%% file: arXiv.tex
\documentclass{article}

\usepackage[T1]{fontenc}
\usepackage[english]{babel}
\input{macros}

\usepackage[letterpaper,top=2cm,bottom=2cm,left=3cm,right=3cm,marginparwidth=1.75cm]{geometry} 
\usepackage{enumitem}

\usepackage{nicefrac}

\usepackage{amsmath}
\usepackage{amssymb}
\usepackage{amsthm}
\usepackage{thmtools}
\usepackage{graphicx}
\usepackage{physics}
\usepackage[colorlinks=true, allcolors=teal, hypertexnames=false]{hyperref}
\usepackage[capitalise]{cleveref}
\usepackage{url}
\usepackage{bm}
\usepackage{comment}
\usepackage[dvipsnames]{xcolor}
\usepackage{enumitem}
\setlist[enumerate,itemize]{itemsep=0.4ex,topsep=0.8ex,leftmargin=1.5em}

\usepackage{titlesec}
\titleclass{\subsubsubsection}{straight}[\subsection]

\newcounter{subsubsubsection}[subsubsection]
\renewcommand\thesubsubsubsection{\thesubsubsection.\arabic{subsubsubsection}}

\titleformat{\subsubsubsection}
  {\normalfont\normalsize\bfseries}{\thesubsubsubsection}{1em}{}
\titlespacing*{\subsubsubsection}
{0pt}{3.25ex plus 1ex minus .2ex}{1.5ex plus .2ex}
\makeatletter
\def\toclevel@subsubsubsection{4}
\def\l@subsubsubsection{\@dottedtocline{4}{7em}{4em}}
\makeatother
\usepackage[giveninits=true,maxbibnames=99,style=alphabetic,maxalphanames=4,minalphanames=3,isbn=false]{biblatex}

\AtEveryBibitem{%
  \clearlist{language}%
}
\renewbibmacro*{doi+eprint+url}{%
  \iftoggle{bbx:doi}
    {\printfield{doi}}
    {}%
  \newunit\newblock
  \ifboolexpr{togl {bbx:eprint} and test {\iffieldundef{doi}}}
    {\usebibmacro{eprint}}
    {}%
  \newunit\newblock
  \ifboolexpr{togl {bbx:url} and test {\iffieldundef{doi}}  and test {\iffieldundef{eprint}}}
    {\usebibmacro{url+urldate}}
    {}}

\usepackage{algorithm}
\usepackage{algpseudocode}

\newtheorem{theorem}{Theorem}[section]
\newtheorem{definition}[theorem]{Definition}

\newtheorem{lem}[theorem]{Lemma}
\newtheorem{assumption}[theorem]{Assumption}
\newtheorem{proposition}[theorem]{Proposition}
\newtheorem{remark}[theorem]{Remark}
\newtheorem{corollary}[theorem]{Corollary}
\newtheorem{conjecture}[theorem]{Conjecture}

\crefname{lem}{Lemma}{Lemmas}
\crefname{paragraph}{Paragraph}{Paragraphs}
\crefname{subsubsubsection}{Section}{Sections}

\DeclareMathOperator{\arcsinh}{arcsinh}

\usepackage{authblk}

\title{Energy, Bosons and Computational Complexity}

\author[1]{Dorian Rudolph}
\author[2]{Arsalan Motamedi}
\author[1]{Dhruva Sambrani}
\author[1]{Hamid Reza Naeij}
\author[3]{\authorcr Ulysse Chabaud}
\author[1]{Sevag Gharibian}
\author[4]{Saeed Mehraban}

\affil[1]{\small Paderborn University and PhoQS, Warburger Stra{\ss}e 100, 33098 Paderborn, Germany}
\affil[2]{University of Waterloo, Ontario, Canada}
\affil[3]{DIENS, \'Ecole Normale Sup\'erieure, PSL University, CNRS, INRIA, 45 rue d’Ulm, Paris, 75005, France}
\affil[4]{Tufts University, Medford, MA, USA}

\begin{document}
\maketitle

\begin{abstract}
We investigate the role of energy, i.e. average photon number, as a resource in the computational complexity of bosonic systems. We show three sets of results: (1. Energy growth rates) There exist bosonic gate sets which increase energy incredibly rapidly, obtaining e.g. infinite energy in finite/constant time. We prove these high energies can make computing properties of bosonic computations, such as deciding whether a given computation will attain infinite energy, extremely difficult, formally undecidable. (2. Lower bounds on computational power) More energy ``='' more computational power. For example, certain gate sets allow poly-time bosonic computations to simulate PTOWER, the set of deterministic computations whose runtime scales as a tower of exponentials with polynomial height. Even just exponential energy and $O(1)$ modes suffice to simulate NP, which, importantly, is a setup similar to that of the recent bosonic factoring algorithm of [Brenner, Caha, Coiteux-Roy and Koenig (2024)]. For simpler gate sets, we show an energy hierarchy theorem. (3. Upper bounds on computational power) Bosonic computations with polynomial energy can be simulated in BQP, ``physical'' bosonic computations with arbitrary finite energy are decidable, and the gate set consisting of Gaussian gates and the cubic phase gate can be simulated in PP, with exponential bound on energy, improving upon the previous PSPACE upper bound.  Finally, combining upper and lower bounds yields no-go theorems for a continuous-variable Solovay--Kitaev theorem for gate sets such as the Gaussian and cubic phase gates.
\end{abstract}

\tableofcontents

\section{Introduction}\label{scn:intro}

Traditionally, much of theoretical quantum algorithms development, from Shor's integer factoring algorithm~\cite{shor1994algorithms} in 1994 to 
the Quantum Singular Value Transform framework (QSVT)~\cite{gilyenQuantumSingularValue2019} in 2019, has taken place in the \emph{Discrete-Variable (DV)} computational model, i.e.~based on a system of finite-dimensional qu\emph{d}its for some $d\in O(1)$. And this is arguably intentional --- not only is the DV model ``closer'' to traditional classical gate-based computing, it also isolates algorithm designers from the underlying physics of whichever experimental platform an algorithm might eventually be implemented on. However, with the advent of the Noisy Intermediate Scale Quantum (NISQ) computing era, this paradigm has begun to shift --- current experimental limitations on qubit counts, coherence, and circuit depth have made it clear that the computational models used in algorithmic theory arguably need to become ``more aligned'' with the underlying physics of experimental platforms.  

To do so, attention in the algorithms community has slowly begun to shift back to the origins of quantum information theory --- \emph{Continuous-Variable (CV)} systems, such as those used in Wiesner's~\cite{wiesnerConjugateCoding1983} quantum money scheme of the late 1970's. 
Here, a CV system is roughly one which has continuous degrees of freedom, and is thus infinite-dimensional. Standard examples include fermionic and bosonic systems. This paper, in particular, studies the computational power of \emph{bosonic} systems.

\paragraph*{Why study the power of quantum computation on CV systems?} First, many quantum systems are inherently CV in Nature. For example, the two recent experimental platforms for quantum advantage implementations of Gaussian Boson Sampling~\cite{aaronsonComputationalComplexityLinear2011,hamiltonGaussianBosonSampling2017,zhongExperimentalGaussianBoson2019} and Random Circuit Sampling~\cite{aruteQuantumSupremacyUsing2019,boulandComplexityVerificationQuantum2019} are based on quantum photonics (which is CV) and superconducting qubits (which use CV systems to simulate DV systems). Second, experimental platforms such as photonics have a decades-long history of successful implementation in the field, for example for Quantum Key Distribution~\cite{bennettQuantumCryptographyPublic1984,}. Third, \emph{hybrid} DV-CV platforms (e.g.~superconducting, trapped ion, and neutral atom; see Fig. 3 of the recent survey~\cite{liu2024hybrid}) have been rising as a potentially viable experimental platform, which has led to intriguing theoretical algorithmic developments, such as the CV quantum factoring algorithm of Brenner, Caha, Coiteux-Roy and Koenig (BCCK)~\cite{brennerFactoringIntegerThree2024}. The latter requires only three oscillators and a qubit, i.e.~a \emph{constant} number of quantum systems, which at first glance is surprising.

\paragraph*{Energy.} This brings us to the elephant in the room for bosonic systems, and our focus here --- \emph{energy}, or more formally, average photon number. For clarity, while by energy we are not referring to Hamiltonian energy levels or physical energy consumption, the latter is nevertheless closely related to average photon number. Specifically, since\footnote{We thank Eugene Tang for pointing this out.} $E= n \omega$, average photon count is a lower bound on average power consumption, even if all other components are lossless. Thus, our results yield relationships between power in the thermodynamic sense and power in the complexity-theoretic sense. 

Continuing with our discussion in terms of energy as average photon number --- that the BCCK CV factoring algorithm requires only $O(1)$ space raises the question --- \emph{where did the exponential degrees of freedom in an $n$-qubit DV quantum system go?} Well, into the \emph{energy} $E$, i.e.~into exponentially many degrees of freedom in the photon number/Fock basis. Indeed, it has long been known that the infinite-dimensional nature of CV systems allows one to do seemingly impossible things, such as the dense coding channel capacity tending to infinity as $E\rightarrow\infty$~\cite{braunsteinDenseCodingContinuous2000}. There even remains an on-going debate about the CV model solving \emph{undecidable} problems \cite{Kieu2003} (see \Cref{sscn:techniques} for details). What is one to make of this? Experimentally, exponential energy in the lab would appear infeasible, let alone \emph{infinite} energy. But can one formally \emph{prove where the boundary between ``plausibly feasible'' and ``intractable'' energy levels} lies? And what might be the \emph{implications for experimental CV setups, e.g.~gate sets used}?

\paragraph*{Studying energy via computational complexity.} A natural framework for addressing these questions is computational complexity theory, which by definition characterizes which resources (e.g.~time, space, communication) are necessary and/or sufficient to solve certain computational tasks. Our goal in this paper is to consider the role of \emph{energy} as a resource in bosonic computation. 

For clarity, this is not the first work to formally study ``bosonic quantum computational complexity''; the latter was initiated by~\cite{chabaudBosonicQuantumComputational2025}, though the focus there was arguably not energy per se. At a high level, the computational model we adopt, dubbed the CV model, generalizes that of \cite{lloyd_quantum_1999,chabaudBosonicQuantumComputational2025}. It compares to a standard DV setup/``DV model'' as follows (formal definitions in \Cref{scn:preliminaries}): 
\begin{itemize}
    \item (State space) $(\bbC^2)^{\otimes n}$ is replaced by the infinite-dimensional space $(\mathbb C^{\infty})^{\otimes n}\cong\ell^2 (\NN_0^{n},\CC)$, the set of square-summable complex sequences with $n$-dimensional index set.
    
    \item (Initial state) $n$-qubit $\ket{0}^{\otimes n}$ is replaced by the vacuum state on $n$ modes, denoted $\ket{0^n}$.  
    
    \item (Gates) $k$-qubit unitary $e^{iHt}$ for Hamiltonian $H$, evolution time $t$, and locality $k\in O(1)$ remains the same, except $H$ is now a $O(1)$-degree polynomial in the position $\hX$ and momentum operators $\hP$ on $k$ modes with constant coefficients. Examples include the non-Gaussian\footnote{Gaussian gate Hamiltonians are defined as having degree at most $2$.} generator $H=\hX^3$ and Gaussian $H=\hX^2+\hP^2$. The total runtime of a sequence of gates $\Pi_{j=1}^m e^{it_jH_j}$ is $\sum_{j=1}^m t_j$. We assume all gates are efficiently uniformly generated\footnote{Roughly, given input $x\in \{0,1\}^n$, a  description of the sequence of bosonic gates to be applied can be generated in $\poly(n)$ time.}.
    
    \item (Measurement) A single qubit standard basis measurement is replaced with a single mode measurement via the number operator $\hN$, which measures the number of photons in said mode.
\end{itemize}
Note that feedforward, e.g.~as in the Knill-Laflamme-Milburn (KLM) or Gottesman-Kitaev-Preskill (GKP) setups~\cite{knillSchemeEfficientQuantum2001,GKP2001}, is not part of this model. \cite{chabaudBosonicQuantumComputational2025} showed (among other results): (1) The CV model with Gaussian initial states (e.g.~vacuum state), Gaussian gates (i.e.~Hamiltonians of degree $\leq 2$) which are additionally logspace uniformly generated, polynomial runtime, and single mode measurement in the position basis, equals the complexity class Bounded-Error Quantum Logspace ({\sf BQL}). (2) There exists a non-Gaussian gate set so that the CV model with polynomial runtime contains \BQP; the reverse containment is not known. There exists a different gate set (Gaussian gates and the cubic phase gate $H=X^3$) so that the CV model with polynomial runtime is contained in \EXPSPACE. The latter has since been improved to \PSPACE~\cite{upreti2025bounding}.

\paragraph*{Interpreting these facts.} That $\BQP\subseteq \AWPP$~\cite{fortnowComplexityLimitationsQuantum1999}, whereas as stated above, the best upper bound known on poly-time CV models with certain gate sets is $\PSPACE$ (which is believed to strictly contain $\BQP$~\cite{vyalyiQMAPPImplies2003}), speaks to the difficulty of analyzing infinite-dimensional systems. Even worse, the CV model upper bounds above are only for \emph{specific} gate sets. Indeed, in contrast to the DV setting, a CV analogue of the Solovay--Kitaev theorem~\cite{dawson2005solovay} is only known to hold for specific subclasses of gate sets \cite{becker2021energy,arzani2025can}. This raises many questions: \emph{What role does energy play in all of this? Can energy be used to disprove versions of a CV Solovay--Kitaev theorem, e.g.~by differentiating the power of different gate sets? Can one obtain formal evidence as to whether constant mode, exponential energy setups such as BCCK~\cite{brennerFactoringIntegerThree2024} are ``reasonable'' or not? And is it possible that ``reasonable'' poly-time CV computations are more powerful than \BQP?} 



\subsection{Results}

We organize our results into three interconnected themes: (1) Why energy is a problem to begin with (i.e.~quantifying the rate of energy growth in bosonic computations), (2) what energy buys us computationally (i.e.~complexity-theoretic lower bounds), and (3) when we can bypass the energy barrier (i.e.~simulation algorithms and complexity-theoretic upper bounds). Combining results from (2) and (3) will yield a no-go for certain versions of a CV Solovay--Kitaev theorem, which is discussed last.  A formal review of CV computations is given in \Cref{scn:preliminaries}; most relevant at this point is the formal definition of energy of a state $\ket{\psi}$, which is $\bra{\psi}\hN\ket{\psi}$.

\paragraph*{1. Quantifying the rate of energy growth.} We begin by formally proving that energy in CV systems can grow at an incredible rate with the right gate sets, even when restricted to polynomial or \emph{constant}(!) time. We show: 

\begin{theorem}[Energy growth rates (informal; see each bullet point for formal reference)]\label{thm:energygrowthinformal}
    Each statement below considers CV computations:
    \begin{enumerate}
        \item Any single-mode circuit of size $t$ with Gaussian gates and the Kerr\footnote{The Kerr gate can be used to generate non-Gaussian resource states such as cat states.} gate $H=\hN^2$ has energy at most $e^{O(t)}$ (\Cref{{prop:N2-G-growth}}).
          \item There exists a single-mode circuit using the Fourier transform $e^{i\frac{\pi}{4}(\hX^2+\hP^2)}$ and cubic phase gate $H=\hX^3$ which yields energy $\sim (2^t)^{2^t}$ (\Cref{lem:cubic-energy-growth}). Even for a dissipative system, energy keeps growing doubly exponentially fast so long as the rate of noise grows slightly slower than a threshold $\gamma_{th} \sim t$, and ceases to grow if it is slightly faster than $\gamma_{th}$.
        \item The two-mode state $\ket{\psi} = e^{-itH}\ket{0}_0\ket{0}_1$ with $H = \frac{i}{2}\X_0^4(\a_1^2-\a_1^{\dagger2})$ has infinite energy for any $t>0$ (\Cref{prop:inf_three}; see \Cref{sscn:infinitefinite} for further examples).
    \end{enumerate}
\end{theorem}
\noindent These are to be interpreted as follows: (1) demonstrates that, in contrast to what is to follow, not all non-Gaussian gate sets lead to drastic energy growth; Gaussian gates and the Kerr give at most exponential energy in poly-time. (2) then rigorously proves that even well-studied gate sets such as Gaussian gates $+$ cubic phase gate~\cite{Gottesman2001,sefi2011decompose} already yield provably doubly exponential energy in polynomial time.  (3) demonstrates the most extreme case; even \emph{constant} time can suffice to achieve infinity energy, i.e.~$\bra{\psi}\hN\ket{\psi}$ diverges. Two follow-up comments are important for (3): (3a) Any normalized state having infinite energy is close in Euclidean distance to a state of \emph{finite} energy up to some photon number cutoff in the Fock basis (\Cref{sscn:infinitefinite}). (3b) Nevertheless, some caution is warranted: Unlike the DV case, in the CV case a small trace distance between $\rho$ and $\sigma$ does \emph{not} necessarily imply $\trace(\rho M)\approx \trace(\sigma M)$, as operator $M$ can have unbounded spectral norm. In sum, \Cref{thm:energygrowthinformal} shows that energy growth rates, even in small circuits, indeed need be reckoned with carefully.

\paragraph*{2. Complexity-theoretic lower bounds.} Having formally demonstrated rapid energy growth in the CV model, we next prove complexity-theoretic implications thereof. These can be categorized into three classes. \\

\vspace{-1mm}
\noindent\emph{a. Complexity class lower bounds.} Define $\CVBQP$ (\Cref{def:cvbqp}) as the set of promise problems solvable in the CV model with (1) \emph{polynomial} runtime and (2) whereas in principle no energy bounds are assumed \emph{during} the computation, at the time of measurement we are promised\footnote{More accurately, we are promised that with constant probability, the Fock basis state we obtain upon measurement is polynomial with high probability.} the energy is at most \emph{polynomial}. Since universal gate sets are not known for the CV model, we further use $\CVBQP[\calG]$ to denote $\CVBQP$ with finite gate set $\calG$.

\begin{theorem}[Complexity class lower bounds (informal; see each bullet point for formal reference)]\label{thm:complexityclassinformal}
For each bullet below, there exists a finite gate set $\calG$ such that:
\begin{enumerate}
    \item $\NP \subseteq \CVBQP[\calG]$ with at most exponential energy and $O(1)$ modes (\Cref{thm:CVBQP-NP}).
    \item For all $k\in \NN$, $\NTIME(\exp^{(k)}(n))\subseteq \CVBQP[\calG]$ with $O(k)$ modes (\Cref{thm:CVBQP-ELEMENTARY}). \item $\PTOWER \subseteq \CVBQP[\calG]$ (\Cref{thm:CVBQP-TOWER}).
\end{enumerate}
\end{theorem}
\noindent These are interpreted as follows: (1) shows \CVBQP\ with at most exponential energy and $O(1)$ modes already suffices to capture NP. This yields formal evidence that the BCCK CV factoring algorithm's~\cite{brennerFactoringIntegerThree2024} computational model is perhaps \emph{not} ``reasonable'', as recall BCCK \emph{also} uses exponential energy and $O(1)$ modes. We caution, however, that \Cref{thm:complexityclassinformal}'s results are gate set-specific, and our gate sets above differ from BCCK. (2) shows that with only $k$ modes but potentially unbounded energy, \CVBQP 's power blows up to non-deterministic $\exp^{(k)}(x)$ time, for $k$-fold iterated exponential $\exp^{(k)}(x):=\exp(\exp(\dotsm\exp(x)\dotsm))$. (2) now yields (3), which shows that with no energy or sub-polynomial mode restrictions on \CVBQP, we can even solve $\PTOWER$, the set of languages decidable in deterministic time $\exp^{(\poly)}(x)$, i.e.~a power tower of polynomial height. In sum, the ``wrong'' choice of gate set can yield rapid energy growth, which in turn provides immense computational power.\\

\vspace{-1mm}
\noindent\emph{b. Almost gate set independent lower bounds via an energy hierarchy.} We next give almost gate set \emph{independent} lower bounds, in that the only gates required are basic: The identity, beamsplitter, and two-mode squeezing. The result is in essence an oracle-separation-based \emph{energy hierarchy theorem}, similar in spirit to time hierarchy theorems such as $\DTIME(n^k)\subsetneq\DTIME(n^{k+1})$ for $k\geq 1$. For this, define $\varepsilon$-\textsc{BeamSplitPrec} as the problem of distinguishing whether a given oracle $O$ implements the identity or a beam-splitter with angle $\varepsilon$ (specified in binary), using a single query to $O$ (\Cref{def:beamsplitter}).

\begin{theorem}[Energy hierarchy]
$\varepsilon$-\textsc{BeamSplitPrec} cannot be solved in $\CVBQP$ with energy $E = o(\varepsilon^{-1})$, but can be solved in $\CVBQP$ for some energy $F = O(\varepsilon^{-2})$.
\end{theorem}
\noindent Thus, even with mild assumptions on the gate set, more energy provably yields greater computational power in the oracle setting.\\

\vspace{-1mm}
\noindent\emph{c. Undecidable properties of CV computations.} We finally show that, not only does high energy yield high computational power, but it also makes computing \emph{properties} of CV computations difficult --- undecidable, in fact! 

\begin{theorem}[Undecidable properties (informal; see each bullet point for formal reference)]{\label{thm:undecidableinformal}}
It is undecidable whether:
\begin{enumerate}
    \item a Gaussian state evolved under a polynomial Hamiltonian has finite energy (c.f.~\Cref{thm:energy-undecidable}).
    \item the evolution of an essentially self-adjoint polynomial Hamiltonian preserves the Schwartz space. (c.f.~\Cref{thm:adjointUndecidable})
    \item a symmetric polynomial Hamiltonian is essentially self-adjoint on the Schwartz space (c.f.~\Cref{thm:schwartz}).
\end{enumerate}
\end{theorem}
\noindent For clarity, a ``polynomial Hamiltonian'' is the class of Hamiltonians allowed in the CV model ($O(1)$-degree polynomials in $\hX$ and $\hP$). The Schwartz space may be viewed as a ``reasonable'' computational space for CV models; formally, it corresponds to the set of states with finite moments $\langle\N^k\rangle$ for all $k\in\NN$ (\Cref{lem:schwartz}), which in particular implies finite energy. Interpretations: (1) adds insult to injury --- not only can the right gate sets rapidly yield infinite energy (e.g.~\Cref{prop:inf_three}), but even deciding if a gate set has this property to begin with is undecidable. With that said, the reader should not lose all hope --- we show shortly, e.g.~that any finite (but unknown) energy bound on \CVBQP\ with ``physical computations'' suffices for decidability. (2) shows that, while polynomial Hamiltonians (viewed as linear maps) always preserve the Schwartz space, the corresponding \emph{unitary} evolution may not, and detecting this is undecidable. (3) shows that the infinite-dimensional analog of a trivial task in finite dimensions, i.e.~checking if a DV operator is self-adjoint, is also undecidable, even if restricted to the ``nicer'' Schwartz space.

\paragraph*{3. Simulation algorithms and upper bounds.} Our third category of results give upper bounds, specifically conditions under which the CV model can be simulated within certain complexity classes. \\

\vspace{-1mm}
\noindent\emph{a. \CVBQP\ with polynomial energy.} We begin with the natural question --- is \CVBQP\ with polynomial energy simulatable by \BQP?
\begin{theorem}[$\BQP$ simulations of $\CVBQP$ (informal; \Cref{thm:polyenergypolysim})]\label{thm:polyenergypolysim_informal}
    $\CVBQP$ with polynomial energy and any gate set is in $\BQPspoly$. If an additional block encoding assumption holds for the gate set used (\cref{asmp:oracle-for-truncated-unitary}), which holds for example for the Kerr gate, Gaussians, Jaynes--Cumming interactions, and cubic phase gate, then the bound improves to $\BQP$.
\end{theorem}


\noindent Above, $\BQPspoly$ is \BQP\ with polynomial bits of classical advice. This theorem shows, for example, for the often studied gate set of Gaussians and cubic phase gate, $\CVBQP$ with polynomial energy is indeed simulatable by BQP. We remark that the formal statement of \Cref{thm:polyenergypolysim_informal} is more general, showing that \emph{any} CV model circuit on  $n$ modes, $k$-local gates, evolution time $T$ and energy bound $E$ can be simulated with some DV quantum circuit of width $O(n \log\frac{E^\ast T}{\delta})$ and depth $O(T E^\ast{}^{2k})$.\\

\vspace{-1mm}
\noindent\emph{b. $\CVBQP$ with arbitrary finite energy.} We next study the opposite extreme: \emph{How large can the energy become with $\CVBQP$ remaining decidable?} For this, we study \CVBQP\ for ``physical computations'', which informally means we assume all states throughout the computation have uniformly finite (but arbitrarily large and unknown) energy. Note we do \emph{not} assume any particular gate set here.

\begin{theorem}[Informal; see \Cref{thm:decidable}]\label{thm:decidableInformal}
  $\CVBQP$ for physical computations is decidable, i.e. $\CVBQP\subseteq \class{R}$, for $\class{R}$ the set of recursive languages. 
\end{theorem}
\noindent Thus, finite energy does not suffice to solve undecidable problems such as the Halting problem, at least for the reasonable superset of experimentally plausible computations we call ``physical''. We conjecture, moreover, that $\CVBQP$ remains decidable even for unbounded (i.e.~possibly infinite) energy and any gate set (see \Cref{app:conjecture} for a statement and rationale).\\

\vspace{-1mm}
\noindent\emph{c. $\CVBQP[\X^3]$ with $\exp$ bounds.} Finally, we study the Gaussian and cubic phase gate set with $\exp$ energy bound. Note that our result actually has wider scope; it also applies to any other gate which can be compiled exactly to this set, such as the quartic gate \cite{kalajdzievski2021exact}. Furthermore, if the cubic gate depth (c.f.~\Cref{prop:E-growth-cubic-depth}) scales as $\polylog$, the energy upper bound is always satisfied.
\begin{theorem}[Informal; see \Cref{thm:XCubedInPP}]\label{thm:XCubedInPPInformal}
    \CVBQP\ with Gaussian gates, $H=\X^3$ and an $\exp$ energy bound is in $\mathsf{PP}$.
\end{theorem}
\noindent This improves on the previous $\PSPACE$ upper bound~\cite{upreti2025bounding} in a key way: It implies that there is at least one common CV gate set, Gaussian and cubic phase gates, whose computational power is provably ``not too far'' from BQP (recall $\BQP\subseteq\AWPP\subseteq\PP$) even with very conservative energy bounds. Whether the two classes coincide, we leave as an open question. 

We note that our result relates the depth of the circuit to the complexity of simulation in the following manner: circuits of constant depth have expected energy at most polynomial in the size of the circuit, hence can be simulated in \BQP. Circuits of polylogarithmic depth, on the other hand, have energy at most exponential, and therefore, using our result can be simulated in \PP. The best known simulation algorithm for circuits of polynomial depth remains \EXPSPACE, as developed in \cite{chabaudBosonicQuantumComputational2025}.

Equally interesting is the novel proof technique for \Cref{thm:XCubedInPPInformal} (see \Cref{sscn:techniques} for further details), which: (1) develops a finite energy CV state injection gadget for the cubic gate, analogous to those for $T$ gates in DV setting, (2) decomposes the non-Gaussian input state to a superposition of Gaussian states, (3) decomposes the output projector into a sum of coherent states, and (4) evaluates each path with a polynomial-size circuit. 

\paragraph{4. Combining lower and upper bounds: No-go for certain versions of CV Solovay--Kitaev.} The Solovay--Kitaev theorem is a fundamental result for DV systems, which roughly states that quantum circuits consisting of polynomially many $1$- and $2$-qubit gates can be simulated within error $\epsilon$ by finite universal gate sets with just $\polylog(1/\epsilon)$ overhead. This overhead is remarkably efficient. In contrast, combining \Cref{thm:complexityclassinformal} (complexity class lower bounds) with \Cref{thm:XCubedInPPInformal} for the CV model shows that, unless $\PTOWER\subseteq \EXPSPACE$ (which we know is false because of the time hierarchy theorem), the Gaussian and cubic phase gate set cannot be an \emph{efficient} universal gate set in the manner of CV Solovay--Kitaev! In other words, any attempt to simulate the gate set behind \Cref{thm:complexityclassinformal}'s \PTOWER\ result via Gaussian and cubic phase gates \emph{must} incur a super-exponential blowup in overhead.

\subsection{Techniques}\label{sscn:techniques}
We now discuss our techniques; for brevity, we focus on selected results.\\

\vspace{-1mm}
\noindent\emph{1. Infinite energy in finite time (\Cref{thm:energygrowthinformal}).} We employ what we amongst ourselves call ``the loophole which breaks physics'' --- namely, controlled squeezers. Briefly, the \emph{squeezing} operation $S(r)$ (\Cref{eq:simple-gaussian}) is a Gaussian operation which allows one to increase the energy of a mode, so that energy grows exponentially in the squeezing parameter $r$. Using polynomial Hamiltonians, we first design a two-mode \emph{controlled}-squeezer which takes the photon number or position of another mode as a squeezing parameter (akin to a DV controlled-$U$ gate for some single qudit unitary $U$). Then, we first squeeze mode $1$ with parameter $r$, and subsequently use mode $1$ as a control to squeeze mode $2$. This allows us to manufacture heavy tails, so that the expected photon number (i.e.~energy) becomes infinite.\\

\vspace{-1mm}
\noindent\emph{2. Complexity class lower bounds (\Cref{thm:complexityclassinformal}).} To solve \NP\ in \CVBQP\ with only exponential energy and $O(1)$ modes, the idea is to use CV adiabatic evolution to solve diophantine equations, beginning with an $\NP$-complete problem related to such equations~\cite{MA78}. Specifically, we perform adiabatic evolution using a time-independent Hamiltonian with a continuous clock register, so that time steps are encoded in the position basis. Coefficients of the diophantine equation are controlled using displaced squeezed states encoding integers; these are ideally encoded via infinitely squeezed states, which we approximate with finitely squeezed states with polynomial squeezing parameter. The approximation error in doing so is rigorously bounded. 
    
To analyze the spectral gap of the adiabatic path (which dictates the speed of evolution), we engineer our Hamiltonians so that the evolution preserves photon number, allowing for a rigorous analysis of the gap. 
By including number operators in the Hamiltonian, we further enlarge the gap arbitrarily by inputting sufficiently squeezed states, so that the adiabatic evolution can be completed within constant time.
An additional challenge is that our adiabatic evolution initially encodes the Diophantine equation's solution in high energy Fock states, whereas \CVBQP\ only allows polynomial energy at the point of measurement. To address this, we design a gadget based on the \emph{detuned degenerate parametric amplifier} \cite{WC19}, which couples the higher energy Fock state encoding the solution with a low energy output mode, so that the output state has polynomial energy. 

To extend this to solving harder classes like \PTOWER, the observation is that the evolution time of our algorithm scales in some sense with the number of possible solutions. Thus, to go beyond $\NP$, we need to expand the search space, since to solve $\NTIME(T)$-hard problems, we need roughly triply exponential energy in $T$ \cite{AM76}. To achieve such high energy, we use our controlled squeezing Hamiltonian repeatedly, which technically produces infinite energy states. However, we can still show that with high probability, the photon number is within a suitable range.
    
\emph{Remark.} Previously, Kieu \cite{Kieu2003} proposed a CV adiabatic algorithm to solve diophantine equations, whose correctness appears to remain the subject of debate~\cite{Hodges05,kieu2005hypercomputabilityquantumadiabaticprocesses,Smith06,kieu2006identificationgroundstatebased}. Specifically, there appears to be no rigorous proof for the gappedness of the Hamiltonian, and thus for required evolution time. In the absence of such a bound, it may be Kieu's algorithm requires infinite time. In contrast, our construction circumvents this issue by using time-independent Hamiltonians and only requires a finite evolution time. Our Hamiltonians also do not depend on the input, and maintain the evolution in a finite-dimensional subspace, allowing our rigorous bounds on the spectral gap. \\

\vspace{-1mm}
\noindent\emph{Simulating \CVBQP\ with arbitrary finite energy (\Cref{thm:decidableInformal})}. To show that \emph{physical} $\CVBQP$ computations (i.e.~computations that admit a uniform bound on average photon number and its moments) are still decidable, even without knowing the bound a-priori, we develop a truncation-based simulation. Running the truncated simulation on sufficiently small time-slices lets us bound the error on each slice. We run the slices with truncation $E+k$, and truncate each slice back down to $E$. If the error in that truncation is sufficiently small, we know that we have chosen a sufficiently large cutoff $E$.\\

\vspace{-1mm}
\noindent\emph{$\CVBQP[\X^3]$ with $\exp$ energy bounds (\Cref{thm:XCubedInPPInformal})}. We prove that the output probabilities of a polynomial-size circuit $\hC$ consisting of cubic and Gaussian gates can be estimated in $\PP$ if the energy is upper bounded by $\exp$. To do so, we apply a gate injection protocol, following \cite{GKP2001}, which allows all cubic gates to be shifted to the beginning of the circuit as a tensor product of squeezed cubic phase states. We then approximate this tensor product as a superposition of $R$ Gaussian states, where $R$ scales polynomially with the energy of the computation. This allows us to write amplitudes as a sum of exponentially many terms, each computable in polynomial time. 
We then prove that our teleportation gadget works on input states of exponential energy, if we increase the number of Gaussians in each teleportation step, by an exponential factor. One of our main technical contributions is to demonstrate that the very high-energy input states required by the magic injection gadget can be reliably substituted with nearby states (in trace distance) that have substantially lower energy. This is particularly challenging in the state injection setting, where post-selecting on two previously close states can potentially lead to states far in trace distance.\\

\subsection{Discussion and open questions}
We have studied the role of energy, i.e.~average photon number, in the CV model of quantum computing. Our results highlight that energy is a formidable \emph{resource} for computational power: Energy can grow rapidly even with constant size CV circuits, allows one to already solve \NP\ problems with just exponential energy and $O(1)$ modes, and can make answering even basic questions about CV systems undecidable. On the other hand, such explosive growth in energy appears tied to specific gate sets, meaning certain gate sets such as Gaussian and Kerr gates do not.

We thus return to the question: \emph{What is one to make of this?} On the one hand, our lower bounds give formal evidence that certain gate sets are simply experimentally infeasible to implement to high precision. Such gate sets also appear to make the computational landscape of CV computing more difficult to navigate, in that they can be used to rule out efficient Solovay--Kitaev-style universal gate set simulations, and to make computing properties of CV computations potentially undecidable. On the other hand, no experimental setup is perfect, and it is of course plausible that sufficiently noisy versions of even our ``most troublesome'' gate sets can be implemented in the lab. The question is then presumably: Could such noisy implementations still be \emph{useful} for solving computational problems?

The next natural question is whether a variant of the CV model would be ``more reasonable'' from a complexity-theoretic standpoint, i.e.~naturally avoid infinite energy blowups or $\PTOWER$ containments? In \Cref{sec:CV-DV}, for example, we study the hybrid CV-DV setup (used e.g.~in~\cite{brennerFactoringIntegerThree2024}) to show that by coupling CV systems with qubits and using only the latter for information readout, one can restrict in some sense the ``usefulness'' of ultra-high energy CV states, in that the number of qubits naturally restricts the bits of precision one can extract. This is, however, not a watertight argument in favor of the hybrid CV-DV model, as e.g.~\Cref{thm:complexityclassinformal}'s proof shows that in some settings information stored in high-energy Fock space can successfully be moved to low-energy Fock space, and thus read out via polynomially many qubits in the CV-DV setup! 

An alternative approach for the same problem is probing the computational power of gate sets that are constrained by physical realizability in a closed system with a bounded energy budget.  This model is motivated by practical implementations such as the squeezing gate, where a nonlinear crystal is pumped by a coherent laser, explicitly limiting the energy within the closed system. More precisely, what is the computational capability of a gate set restricted to effective gates implementable with a coherent (or thermal) state of energy $E$ and governed by energy-preserving Hamiltonians? Are such energy-constrained gate sets still computationally powerful? How does the computational power of these gate sets scale with the available energy $E$, and the interaction complexity? 

Other open questions include: We have shown that \CVBQP\ with polynomial energy bounds and certain gate sets (e.g.~Gaussian and cubic phase) are contained in \BQP. Does the reverse containment also hold for these same gate sets? (\cite{chabaudBosonicQuantumComputational2025} shows containment of \BQP\ in \CVBQP\ but with a \emph{different} gate set.) If so, does \CVBQP\ with Gaussian and cubic gates with even $\exp$ bounds \emph{equal} \BQP\ (recall we showed the former is in $\PP$ (\Cref{thm:XCubedInPPInformal}))? What about the case without apriori energy bounds? What other CV gates may benefit from development of a state injection gadget, as done here for \Cref{thm:XCubedInPPInformal}? Is $\CVBQP$ with logarithmic energy contained in Bounded-Error Quantum Logspace ($\class{BQL}$)? Does our \Cref{conj:decidable} that \CVBQP\ is decidable for any gate set and without energy bounds hold? 

 During the completion of this work, we became aware of two related works, one by Alex Maltesson et al.\ that was submitted to the arXiv simultaneously, and one by Brenner et al.\ recently submitted to the arXiv~\cite{brenner2025trading}.

 \paragraph{Operator domains and connections with physical implementation}

We end the introduction with a discussion on operator domains and connections between our model and physical realization. It is well-known that valid examples of quantum states, i.e., square summable complex series, with unbounded average photon number. For instance, consider the quantum state
$$
\ket{\psi} = \sqrt{\frac{6}{\pi^2}} \sum_{n\geq 1} \frac{1}{n} \ket{n}
$$
which is properly normalized but has unbounded second moments, e.g., $\bra{\psi} \hat N \ket{\psi} = \infty$. To avoid such divergences we can work with a dense subset of the Hilbert space with bounded quadrature moments, implying thereby that energy does not diverge. See \cref{def:Shwartz} for more details. 

The standard way we define unitaries over infinite-dimensional Hilbert spaces is through the spectral theorem.

However one may face difficulties in describing unitary operators as exponentiation of Hermitian Hamiltonian operators if the Hamiltonian operator does not satisfy a condition called essential self-adjointness~\cite{nelson1959analytic}. In simple words, essential self-adjointness is the property that the domain of a Hermitian symmetric operator is different from the domain of the Hermitian conjugate. For instance, consider the Hamiltonian $P^2 - X^4$. We can show that, while Hermitian symmetric, this operator under exponentiation does not define a valid unitary because the adjoint has imaginary eigenvalues, thereby it does not allow a valid a spectral decomposition (see \cite{hall2013quantum, reed2012methods}). 

Operators we consider in this work are essentially self adjoint therefore they generate a valid unitary evolution. However they do not need to preserve the Schwartz space. As a result, we can find examples of unitary dynamics that generate unbounded average photon number in finite time. The main claim of our work is not that such operators have physical realization. What we emphasize is that if we consider the standard model of continuous quantum computation based on polynomial Hamiltonians defined in \cite{lloyd_quantum_1999} then one faces divergence in the computational domain. We note that energy divergence phenomena for unitaries generated by essentially self-adjoint Hamiltonians were known-known prior to our work based on operators such as higher order squeezing \cite{hong1985higher, hillery1987amplitude}. Our main contribution is to translate them to the computational domain. 

We also emphasize that the classically algorithm described in \cref{scn:simulations_upperbounds} utilizes squeezing of arbitrarily large values. On the experimental front, there is a limit on how much squeezing can be achieved in the laboratory setting. For instance \cite{vahlbruch2016detection} achieves 15 dB squeezing but going beyond this limit faces specific challenges. Hence, one should not view the algorithm in \cref{scn:simulations_upperbounds} as an experimentally viable algorithm, but mainly as a mathematical way of making comparison between the computational capabilities of various models in the asymptotic limit. Similarly we emphasize that the no-go theorem in the SK which we present is about compiling essentially self-adjoint polynomial Hamiltoninas into more tangible sets such as Kerr and Gaussian studied by \cite{lloyd_quantum_1999}. 

How should we choose the definition of a physically realistic operation? The answer to this question is eventually subjective and depends on the context in which we want to study these models. For instance one could define a physical operation to mean an operation which maps finite energy states to finite energy states (see, e.g., \cite{arzani2025can}). However even for these families of operations we expect similar complexity divergences to hold, mainly because, even though bounded, the target energy bound can still be extremely high. In this work we choose essentially self-adjointness as the mathematical definition of allowed operations. This however does not imply a feasible experimental implication.

In order to obtain a realistic model one needs to also include the effect of noise. In \cref{sec:growth} we show that in the case of cubic and Gaussian gate set, energy blowups persist under linear loss. However more extensive study of noise is needed and we leave it to future work. 

In conclusion the main message of this work is that to avoid the complexity divergences presented in this work, there is a need for a more robust and concrete formulation of universal computation over continuous variables. 

In \cref{sec:growth} we quantify energy growth for unbounded operators. We emphasize that energy divergence for unbounded operators on Fock spaces, including finite-time energy blow-up, is a known phenomenon (see for instance \cite{reed2012methods}). Note the notation of essential self-adjointness was formalized in a seminal work of Nelson \cite{nelson1959analytic}.
 




\section{Preliminaries and notation}\label{scn:preliminaries}

\subsection{Continuous variable quantum systems}

We denote the Hilbert space corresponding to $n$ qudits with $\mathcal{Q}^{(d)}_n \cong (\mathbb C^d)^{\otimes n}$; for simplicity of notation we denote $\mathcal{Q}_n = \mathcal{Q}^{(2)}_n$. We denote the Hilbert space corresponding to $n$ continuous modes with $\mathcal{M}_n$ as the set of unit vectors in $(\mathbb C^{\infty})^{\otimes n}$; $\mathcal{M}_n$ is isomorphic to the set of square summable complex sequences $\ell^2 (\mathbb{C}^{n})$ which is isomorphic to $\mathcal{L}^2(\mathbb{R}^n)$.

\paragraph*{Position and momentum operators:} For a particle in mode $j$, the position and momentum operators are denoted by $\hX_j$ and $\hP_j$, respectively. They satisfy the algebra $[\hX_j, \hP_k] = i \delta_{j,k} \hI$ and $[\hX_j, \hX_k] = [\hP_j, \hP_k] = 0$. When talking about a single mode we will drop the subscript. 
$\hX$ and $\hP$ have continuous spectra corresponding to eigenbasis $\{\ket{x}\}_{x \in \mathbb{R}}$ and $\{\ket{p}\}_{p \in \mathbb{R}}$ with $\hX \ket{x} = x \ket{x}$ and $\hP \ket{p} = p \ket{p}$. These basis satisfy orthonormality condition $\langle{x|x'}\rangle = \delta (x - x')$ and $\langle{p|p'}\rangle = \delta (p - p')$ and are related to each other according to 
$$
\ket{x} = \frac{1}{\sqrt{2\pi}} \int_{p \in \mathbb{R}} e^{i px} \ket{p} dp, \quad \ket{p} = \frac{1}{\sqrt{2\pi}} \int_{x \in \mathbb{R}} e^{-i px} \ket{x} dx.
$$
The vacuum state can be decomposed in position and momentum bases as 
$$
\ket{0} = \frac{1}{\pi^{1/4}} \int_{x \in \mathbb{R}} e^{-\frac{x^2}{2}} \ket{x} dx = \frac{1}{\pi^{1/4}} \int_{p \in \mathbb{R}} e^{-\frac{p^2}{2}} \ket{p} dp 
$$

For an operator $\hA$ as a polynomial in $n$ position and momentum operators $\hX_1, \ldots, \hX_n; \hP_1, \ldots, \hP_n$ we denote $d(A)$ to be the degree of this polynomial in its minimal representation (with respect the commutation rule $[\hX_j, \hP_k] = i \hI \delta_{j,k}$).

\paragraph*{Ladder operators:}
We also define the ladder (creation and annihilation) operators $\ha, \ha^\dagger$ with the action on Fock basis
$$
\ha \ket{n} = \sqrt{n} \ket{n-1}, \quad \ha^\dagger \ket{n} = \sqrt{n+1} \ket{n+1}
$$
They are related to position and momentum operators via 
$$
\hX = \frac{\ha + \ha^\dagger}{\sqrt{2}}, \quad \hP = \frac{\ha - \ha^\dagger}{\sqrt{2} i}
$$
The number operator $\hN = \ha^\dagger \ha = \frac{1}{2} (\hX^2 + \hP^2 - \hI)$ is diagonal in the Fock basis $\hN \ket{n} = n \ket{n}$.

\paragraph*{Single-mode Gaussian operators:} A unitary operator $G$ is called Gaussian if it is a product of terms of the form $e^{iH}$, where $ H$ is a quadratic Hamiltonian. Over a single mode, we consider the following special single-mode Gaussian processes:
\begin{align}\label{eq:simple-gaussian}
\begin{split}
 \hR(\phi) &= e^{i\phi  \ha^\dagger  \ha}\quad\quad\quad\text{(Rotation)}\\
 \hD(\delta) &= e^{\delta  \ha -  \delta^*  \ha^\dagger} \quad\quad\text{(Displacement)}\\
 \hS(\xi) &= e^{\xi  \ha^2 +  \xi^* { \ha^{\dagger2}} } \quad\;\text{(Squeezing)}.
\end{split}
\end{align}
They have the following action on the creation and annihilation operators:
\begin{align}
\begin{split}
 \hR(\phi)  \ha^\dagger  \hR(\phi)^\dagger &= e^{i\phi}  \ha^\dagger\\
 \hD(\delta)  \ha^\dagger  \hD^\dagger (\delta) &=  \ha^\dagger - \delta^* \hI\\
 \hS(r e^{i \phi})  \ha^\dag  \hS^\dagger(r e^{i\phi}) &= (\cosh r)  \ha^\dagger - e^{- i \phi}(\sinh r)  \ha.
\end{split}
\end{align}
It turns out any single-mode Gaussian unitary operator can be described as a product $ \hS (\xi) \hD (\delta) \hR (\theta)$ with suitable parameters.

It is insightful to view the evolution of these operators under Gaussian operators:
\begin{align}
\begin{split}
\hR(\theta) \hX \hR(\theta)^\dagger &= \cos (\theta) \hX + \sin(\theta) \hP,\\
\hR(\theta) \hP \hR(\theta)^\dagger &= -\sin (\theta) \hX + \cos(\theta) \hP,\\
\hD(\delta) \hP \hD^\dagger (\delta) &= \hP -\sqrt {2}\Re (\delta),\\
\hD(\delta) \hX \hD^\dagger (\delta) &= \hX -\sqrt {2}\Im (\delta),\\
\hS(r e^{i \phi}) \hX \hS^\dagger(r e^{i \phi}) &= (\cosh r - \cos \phi \sinh r) \hX + \sinh r \sin \phi \hP,\\
\hS(r e^{i \phi}) \hP \hS^\dagger(r e^{i \phi}) &= -(\cosh r + \cos \phi \sinh r) \hP + \sinh r \sin \phi \hX.
\label{eq:affine}
\end{split}
\end{align}
It turns out that any Gaussian operation can be decomposed as $\hS (\xi) \hD (\delta) \hR(\theta)$ with suitable parameters $\xi, \delta \in \bbC$ and $\theta \in [0, 2\pi)$.

\paragraph*{Multi-mode Gaussian operations:} In general, a multi-mode Gaussian operation corresponds to exponentiation of a quadratic Hamiltonian. An important family of multi-mode operations is linear optical operations. A linear optical operation over $n$ modes corresponds to an $n \times n$ unitary matrix $\hU$ such that $\hU \ha_j \hU^\dagger = \sum_{k} U_{j,k} \ha^\dagger_k$. A linear optical operation can be synthesized as a sequence of rotations (or phase shifters) and beam splitters (corresponding to $\hB_{j,k} = e^{i \theta \ha_j \ha^\dagger_k}$). Importantly, linear optical operations do not change the photon number in a physical system, hence they are also sometimes called passive operations. It turns out an arbitrary Gaussian operation can be decomposed as $\hU \otimes_j \hG_j \hV$, where $\hU$ and $\hV$ are linear optical unitaries and $\hG_j$ are single-mode Gaussian operations. 

\paragraph*{The symplectic formalism from Gaussian operations:}
For a system composed of $n$ modes, let $\mathbf{\hR} = (\hX_1, \ldots , \hX_n; \hP_1, \ldots , \hP_n)^T$. Let 
$$
\hat \Omega = \begin{pmatrix}
    \hat 0_n & \hI_n\\
    -\hI_n & \hat 0_n 
\end{pmatrix}
$$
then $[\hR_j, \hR_k] = \Omega_{j,k} i \hI$. We define $\mathrm{Symp}_n = \{\mathbf{A} \in \bbC^{n \times n} : \mathbf{A} \hat \Omega = \hat \Omega \mathbf{A}^T \}$. To a Gaussian operation one can associate a symplectic matrix $\mathbf{A} \in \mathrm{Symp}_{n}$ and a real vector $\mathbf{d} \in \bbR^{2n}$ such that under conjugate map $\mathbf{\hR} \mapsto \mathbf{A} \mathbf{\hR} + \mathbf{d}$. For a quantum state $\ket{\psi}$ over $n$ modes we define the \emph{covariance matrix} to be a $(2n) \times (2n)$ matrix $\mathbf{V}$ with entries
$$
V_{jk} = \frac{1}{2}\langle\{\hR_j, \hR_k\}\rangle - \langle{\hR_j}\rangle \langle{\hR_k}\rangle
$$
and a mean vector $\mathbf{\mu} \in \bbR^{2n}$ corresponding to $\boldsymbol {\mu} = \langle\mathbf{\hR}\rangle$. The covariance matrix can be written as
\[
\mathbf{V} = \begin{pmatrix} \hA & \hC \\ \hC^T & \hB \end{pmatrix}, 
\qquad A,B \in \mathbb{R}^{n\times n}, \; C \in \mathbb{R}^{n\times n}.
\]
$A, B, C$ satisfy $\hA \hB - \hC \hC^T = \hI/4$. 

Now suppose we evolve a quantum state according to a Gaussian $\hG$ such that $\hG^{-1}$ corresponds to the $\mathbf{A} \in \mathrm{Symp}_n$ and $\mathbf{b} \in \bbR^{2n}$. Then the mean vector and covariance matrices evolve as 
\asp{
\label{eq:gaussian-update}
\hG : \boldsymbol{\mu} &\mapsto \mathbf{A} \boldsymbol{\mu}  + \mathbf{d}\\
\hG : \mathbf{V} &\mapsto \mathbf{A} \mathbf{V} \mathbf{A}^T
}

\paragraph*{Gaussian quantum states:} A pure quantum state is Gaussian if it can be written as $\hG \ket{0}$ for a Gaussian operation $\hG$. If $\ket{\psi}$ is a pure Gaussian state, then the quantum state is fully specified by $\mathbf{V}$ and $\boldsymbol{\mu}$ above. 

For a \emph{pure Gaussian state}, the position-basis wavefunction takes the form
\asp{
    \label{eq:gaussian-position}
    \psi(\mathbf{x}) = 
    \frac{\det(\Re \mathbf{K})^{1/4}}{\pi^{n/4}}
    \exp\!\left(
    -\tfrac{1}{2} (\mathbf{x}-\bar{\mathbf{x}})^T \mathbf K \, (\mathbf{x}-\bar{\mathbf{x}})
    + i\, \bar{\mathbf{p}}^T (\mathbf{x}-\bar{\mathbf{x}})
    \right),
}
where $K$ is a complex symmetric $n\times n$ matrix with $\Re K > 0$ given by
\[
K = (2A)^{-1} - 2i\, A C .
\]
Here $A$ and $C$ are the position and position-momentum correlation blocks of the covariance matrix $V$, and $\bar{\mathbf{p}}$ is the momentum component of $\boldsymbol{\mu}$. See \cite{gonzalez2021cluster}.

\paragraph*{Squeezed states:} An important family of Gaussian state are squeezed states. We use the notation $\hS_\xi := \hS (\ln (\xi))$ which maps
$$
S_\xi \hat X S^\dagger_\xi = \hat{X}/\xi, \quad S_\xi \hat P S^\dagger_\xi = \hat{P} \xi.
$$
We can show that $\hS_{\xi} \ket{x} = \sqrt{\xi}\ket{\xi x}$ and $\hS_{\xi} \ket{p} = 1/\sqrt{\xi}\ket{p/\xi}$.
The squeezed vacuum state is defined as 
\asp{
\ket{S_\xi} := \hS_{\xi} \ket{0} = \frac{1}{\pi^{1/4} \xi^{1/2}} \int_{x \in \mathbb{R}} e^{-\frac{x^2}{2\xi^2}} \ket{x} dx
}
To approximate the position basis, we use the squeezed vacuum state for very small values of $\xi$
\begin{align}
    \ket{S_\xi (y)} := e^{i y \hat P} \hat S_\xi \ket{0} = \frac{1}{\pi^{1/4} \xi^{1/2}} \int_{x \in \mathbb{R}} e^{-\frac{(x-y)^2}{2\xi^2}} \ket{x} dx
\end{align}

Let $\hF = e^{i \frac{\pi}{4} (\hX^2 + \hP^2)}$ be the Fourier transform which satisfies
$$
\hF^\dagger \hX \hF = \hP, \quad \hF^\dagger \hP \hF = - \hX
$$
We use the following notation to indicate non-Gaussian gates
$$
\hV^{(k)}_j (\theta) := e^{i \theta \hX_j^k/k}
$$

The two-mode squeezing gate is defined as
\begin{align}
\hT_\xi = \exp(\xi(\ha_1\ha_2 - \ha_1^\dagger \ha_2^\dagger)),
\end{align}
and it satisfies
\begin{align}
\hT_\xi\ket{0,0} = \mathrm{sech}(\xi) \sum_{n=0}^\infty \tanh^n(\xi) \ket{n,n}.
\end{align}

\paragraph*{Types of measurement for Bosonic systems:}

\begin{itemize}
    \item Homodyne Detection: A homodyne measurement in position basis, is a projective measurement with projections onto position eigenstates $\ket{q}\bra{q}$. This definition can be extended to a homodyne detection in other "directions" by first applying a rotation gate.
    \item Heterodyne Detection: It is a projective measurement with projections onto a coherent state $\ketbra{q}$.
    \item General Gaussian projections: It is a projection onto an arbitrary guassian state 
    $\ketbra{g}$. The heterodyne detections are a special case.
    \item Photon Number Detections (PNR):  It is a projective measurement onto the Fock states $\ketbra{n}$
    \item Partial Trace: A measurement of a subsystem followed by "forgetting" the outcome.
\end{itemize}

Of these, partial trace, homodyne, heterodyne, and Gaussian projections are all Gaussian measurements, in that the resulting state remains Gaussian\footnote{The position eigenstate can be thought of as the infinite limit of the (anti-)squeezed state}. The PNR measurement results in a non-zero state always.

\paragraph*{Basic notation:}
$\mathbb{N}_0 = \bbZ_{\geq 0} = \{0, 1, 2, \ldots\}$. For operators $\hA$ and $\hB$ we use the notation $\ad_{\hA} (\hB) = [\hA, \hB]$. For a tuple of numbers $\mathbf{T} = (t_1, \ldots, t_m) \in \bbZ_{\geq 0}^m$ let $\nu (\mathbf{T}) = (\nu_0, \nu_1, \nu_2, \ldots)$ where $\nu_j$ is the frequency of number $j \in \bbZ_{\geq 0}$. Furthermore, let $\nu (\mathbf{T}) ! = \nu_0 ! \nu_1 ! \ldots $. For an operator $\hA$ and quantum state $\ket{\psi}$ define $\langle \hA\rangle_\psi := \bra{\psi} \hA \ket{\psi}$. For a quantum state $\ket{\Phi}$, we denote the corresponding density matrix as $\Phi = \ket{\Phi}\bra{\Phi}$.  

We need the following basic lemma
\begin{lem}
    Let $A_1, \ldots, A_k$ be real-valued random variables then 
    $$|\E (A_1, \ldots, A_k)| \leq \sqrt{\E (A_1^2) \ldots \E (A_1^2)}$$.
\label{lem:moment-upperbound}
\end{lem}

\paragraph*{Quantum channels:}
Quantum channels correspond to completely positive and trace-preserving maps (CPTP): a quantum channel $\mathcal{E}$ is completely positive (CP) if $\hat \rho \geq 0$ iff $\mathcal{E} (\hat \rho) \geq 0$. It is trace preserving (TP) if for any linear operator $\hat \rho$, $\Tr (\mathcal{E}[\hat \rho]) = \Tr (\hat \rho)$. A quantum channel $\mathcal{E}$ is unital if $\mathcal{E} (\hI) = \hI$.
For a super operator $\mathcal{A}$, its dual (or conjugate) with respect to the Hilbert-Schmidt inner product is a superoperator $\mathcal{A}^\dagger$ such that for any operators $\hA$ and $\hB$, $\langle{\hA , \mathcal{E} [\hB]}\rangle = \langle{\mathcal{E}^\dagger [\hA] , \hB}\rangle$. We note that for unbounded operators the dual does not necessarily exist.

For open quantum systems with Markovian dynamics, the dynamics can be described by the Lindbladian, the quantum analogue of the Louvillian. 

$$\mathcal{L}(\hat \rho) = i[\hat H, \hat \rho] + \gamma \sum_i \gamma_i \left(\hL_i^\dagger \rho \hL_i - \frac{1}{2}\{\hL_i^\dagger \hL_i, \hat \rho\}\right),$$
where $\{.,.\}$ is the anticomutator.

The $[H,\rho]$ is the unitary dynamics of the system, and the $L_i$ are called the jump terms which govern the non-unitary evolution of a density matrix $\rho$.

The dissipation channel $\mathcal{E}$ corresponds to one-step evolution of the quantum system under the Lindbladian dynamics:
\asp{
\frac{d\hat \rho}{dt} = \mathcal L[\rho] := \gamma (\ha \hat \rho \ha^\dagger - \frac{1}{2}\{\hat \rho, \hN\}).
\label{eq:lindbladian}
}
The dual of this dissipation channel exists and corresponds to the Lindbladian
$$
\mathcal L^\dagger[\rho] := \gamma (\ha^\dagger \hat \rho \ha - \frac{1}{2}\{\hat \rho, \hN\}).
$$

\paragraph*{Useful integrations:}
We will use the following Gaussian integration formula for $A, B \in \mathbb{C}$:
\[
\int_{-\infty}^{\infty} e^{-A x^2 + Bx} dx = \sqrt{\frac{\pi}{A}} \exp\left( \frac{B^2}{4A} \right),
\]

and the multivariate Gaussian integral:

\asp{
    \int_{\mathbb{R}^n} \exp\left( -\frac{1}{2}\mathbf{x}^T \mathbf{A} \mathbf{x} + \mathbf{J}^T \mathbf{x} \right) \,d^n\mathbf{x} = \sqrt{\frac{(2\pi)^n}{\det(\mathbf{A})}} \exp\left(\frac{1}{2}\mathbf{J}^T \mathbf{A}^{-1} \mathbf{J}\right).
\label{eq:MVG-integral}
}

\subsection{Complexity classes}

\begin{definition}[\CVBQP~\cite{chabaudBosonicQuantumComputational2025}]\label{def:cvbqp}
  Let $\calG \subseteq \CC[\X_1,\P_1,\dots,\X_k,\P_k]$ be a set of Hamiltonians, such that each $H_i$ is essentially self-adjoint on a dense domain.\footnote{Essential self-adjointness is required to make the unitary evolution well-defined via Stone's theorem (e.g.~\cite{hall2013quantum}). In infinite-dimensional Hilbert spaces, Hermitian operators are not necessarily self-adjoint. Indeed, $\P$ is not self-adjoint, only \emph{essentially} self-adjoint, i.e.~its closure is self-adjoint \cite{hall2013quantum}.}
  We say a bosonic $m$-mode circuit $C$ uses gate set $\calG$, if we can write 
  \begin{equation}
    C = e^{-it_lH_l}\dotsm e^{-it_1H_1},
  \end{equation}
  where $t_1,\dots,t_l\in \RR$ and $H_1\dots,H_l$ are either Gaussian (i.e.~degree $2$) Hamiltonians or belong to $\calG$ (up to renaming modes).

  We say a promise problem $A\in \CVBQP[\calG]$ if there exists a poly-time uniform family of bosonic circuits $\{C_x\}_{x\in\{0,1\}^*}$ using gate set $\calG$, total runtime $\le\poly(\abs{x})$, and integer bounds $0<a<b\le \poly(\abs{x})$
  such that for all $x\in \{0,1\}^*$ and output state $\ket{\psi_x} = C_x\ket{0,\dots,0}$
  \begin{subequations}
    \begin{alignat}{2}
      x&\in\Ayes\quad&\Longrightarrow\quad &\Pr[\N_1 \in [a,b]] \ge 2/3,\\
      x&\in\Ano\quad&\Longrightarrow\quad &\Pr[\N_1 \le a-1 ] \ge 2/3,
    \end{alignat}
  \end{subequations}
  where $\Probability[\N_1=n] = \abs{\braket{n}{\psi_x}}^2$ denotes the probability of measuring $n$ photons in the first mode of the output state of $C_x$ on vacuum input.
\end{definition}

\begin{remark}
    Our definition of \CVBQP\ differs from \cite{chabaudBosonicQuantumComputational2025} in that we let the universal circuit family depend on the input itself and not just its length.
    This is reasonable from an experimental perspective, and stanard for $\BQP$ where it holds without loss of generality (remains open for $\CVBQP$).
    The main advantage is that we do not need one mode for each input bit, and can tackle $\NP$-hard problems with just a constant number of modes (\cref{thm:CVBQP-NP}).

    The probability estimation problem for bosonic circuits of \cite{chabaudBosonicQuantumComputational2025} remains a natural complete problem with our modified definition.
\end{remark}

\begin{definition}[Energy-bounded \CVBQP ($\CVBQP_f$)]\label{def:cvbqpBoundedEnergy}
    For computable function $f:\N\rightarrow\N$, defined as $\CVBQP$, except  the energy satisfies $\bra{\psi_t}\hN\ket{\psi_t}\leq f(n)$ at any intermediate point 
    $\ket{\psi_t}$ in the computation, for $n$ the input size. 
\end{definition}
\noindent Thus, e.g.~$\CVBQPpoly$ and $\CVBQPexp$ indicate \CVBQP with polynomial and exponential bounds on energy, respectively.

Let $\exp^{(n)}(x) = \exp(\exp(\dotsm\exp(x)\dotsm))$ denote $n$-fold exponentiation, and denote tetration by $a\uuarr b$, i.e.,
\begin{equation}
    a\uuarr b \coloneq \underbrace{a^{a^{\cdot^{\cdot^{\cdot^a}}}}}_{b\text{ times}}.
\end{equation}

\begin{definition}\label{def:tower}
    We define the following complexity classes:
    \begin{align}
        \ELEMENTARY &\coloneq \bigcup_{k\in\NN}\DTIME(\exp^{(k)}(n))\\
        \PTOWER &\coloneq \bigcup_{k\in\NN}\DTIME(2\uuarr n^k)\\
        \TOWER &\coloneq \bigcup_{k\in\NN}\DTIME(2\uuarr \exp^{(k)}(n))\quad \textup{\cite{Schmitz16}}
    \end{align}
    In words, \ELEMENTARY\ is the class of languages that can be solved by a Turing machine whose runtime is bounded by an elementary function (i.e.~at most a finite composition of exponential functions).
    In \PTOWER\ the runtime is bounded by a power tower of polynomial height, whereas \TOWER\ allows elementary tower height.
\end{definition}

\begin{remark}
  $\EXP \subsetneq \ELEMENTARY \subsetneq \PTOWER \subsetneq \TOWER$ by the time hierarchy theorem \cite{HS65,HS66}.
\end{remark}

\begin{remark}
  We can replace $\DTIME$ in \cref{def:tower} by $\NTIME$, $\DSPACE$, or even $\NSPACE$ without changing the power of the classes.
\end{remark}

\section{Quantifying energy growth rates in CV systems}\label{sec:growth}

In this section, we show how different concrete gate sets can lead to drastically different energy growth rates, including even infinite energy in finite time.

\subsection{Gaussian gates and the Kerr gate: Exponential upper bound}

We first consider the Gaussian gates and Kerr gate $\hN^2$ and show that, unlike the cubic and Gaussian gate set, energy for this gate set grows at most exponentially fast. 

\begin{proposition}
    Let $\hU$ be a circuit of size
    at most $T$ over a single mode with gates from the Gaussian and $\hN^2$ gate set (each gate specified with constant bits of precision), then $\bra{0} U^\dagger \hN U\ket{0} \leq e^{c T}$ for some constant $c > 0$.
    \label{prop:N2-G-growth}
\end{proposition}

To prove this result, we use the following tool:
\begin{lem}
    Let $\hG$ be a single-mode Gaussian gate with parameters specified with constant bits of precision, and let $\ket{\psi}$ be a CV quantum state with $\bra{\psi} \hN \ket{\psi} = E \geq 1$ and let $\ket{\phi} = \hG \ket{\psi}$, then there exist constant $\alpha, \beta > 0$ such that $\bra{\phi} \hN \ket{\phi} \leq \alpha E + \beta$. 
    \label{lem:G-energy}
\end{lem}

\begin{proof}
Using standard results \cite{weedbrook2012gaussian}, we may decompose
\[
\hat G^\dagger = \hat S(r)\hat D(\delta)\hat R(\theta),
\]
where the parameters \(r,\delta,\theta\) are constants. Since rotations commute
with \(\hat N\), they do not affect the photon number. Hence
\[
\bra{\phi}\hat N\ket{\phi}
=
\bra{\psi}\hat G^\dagger \hat N \hat G\ket{\psi}
=
\bra{\psi}\hat N'\ket{\psi},
\]
where
\[
\hat N'
:=
\hat S(r)\hat D(\delta)\hat N\hat D(\delta)^\dagger \hat S(r)^\dagger .
\]
Using
\[
\hat N=\frac12(\hat X^2+\hat P^2-\hat I),
\]
and the transformation rules
\[
\hat D(\delta)\hat X\hat D(\delta)^\dagger
=
\hat X-\sqrt2\,\delta_r,
\qquad
\hat D(\delta)\hat P\hat D(\delta)^\dagger
=
\hat P-\sqrt2\,\delta_i,
\]
together with
\[
\hat S(r)\hat X\hat S(r)^\dagger=e^r \hat X,
\qquad
\hat S(r)\hat P\hat S(r)^\dagger=e^{-r}\hat P,
\]
we obtain
\asp{
\hat N'
&=
\frac12
\left(
(e^r\hat X-\sqrt2\delta_r)^2
+
(e^{-r}\hat P-\sqrt2\delta_i)^2
-
\hat I
\right)\\
&=
\frac12 e^{2r}\hat X^2
+
\frac12 e^{-2r}\hat P^2
+
|\delta|^2
-
\sqrt2 e^r\delta_r \hat X
-
\sqrt2 e^{-r}\delta_i \hat P
-
\frac12 \hat I .
}
Taking expectation in \(\ket{\psi}\) and applying Cauchy--Schwarz gives
\begin{align}
\bra{\phi}\hat N\ket{\phi}
&\leq
\frac12 e^{2r}\bra{\psi}\hat X^2\ket{\psi}
+
\frac12 e^{-2r}\bra{\psi}\hat P^2\ket{\psi}
+
|\delta|^2 \nonumber\\
&\qquad
+
\sqrt2 e^r|\delta_r|
\sqrt{\bra{\psi}\hat X^2\ket{\psi}}
+
\sqrt2 e^{-r}|\delta_i|
\sqrt{\bra{\psi}\hat P^2\ket{\psi}} .
\end{align}
Since
$\hat X^2+\hat P^2=2\hat N+\hat I$,
we have
\[
\bra{\psi}\hat X^2\ket{\psi},
\bra{\psi}\hat P^2\ket{\psi}
\leq 2E+1.
\]
Because \(E\geq 1\), \(2E+1\leq 3E\) and \(\sqrt{2E+1}\leq \sqrt{3E}\leq \sqrt3 E\). Therefore
\[
\bra{\phi}\hat N\ket{\phi}
\leq
\left[
\frac32(e^{2r}+e^{-2r})
+
\sqrt6\left(e^r|\delta_r|+e^{-r}|\delta_i|\right)
\right]E
+
|\delta|^2.
\]
Thus the claim holds with, for example,
\[
\alpha=
\frac32(e^{2r}+e^{-2r})
+
\sqrt6\left(e^r|\delta_r|+e^{-r}|\delta_i|\right),
\qquad
\beta=|\delta|^2.
\]
Since \(r\) and \(\delta\) are fixed constants determined by the gate \(\hat G\),
\(\alpha,\beta>0\) are constants.
\end{proof}

\begin{proof} [Proof of \cref{prop:N2-G-growth}] The proof is immediate from \cref{lem:G-energy}. Let $E_t$ be the energy of the system after the application of the $t$th Gaussian gate corresponding to parameters $\alpha_t$ and $\beta_t$. Let $\alpha = \max_t \alpha_t$ and $\beta = \max_t \beta_t$, then $E_{t+1} \leq \alpha E_t + \beta$. We deduce that $E_T \leq e^{cT}$ for some constant $c > 0$. 
\end{proof}

\paragraph*{Systems with noise.} Next, we study the dissipative system:
\begin{proposition}
    For a dissipative system, where we apply a Gaussian gate at each step, followed by an $\hbar N^2$ gate, and then send the quantum state through a dissipative channel according to the one-step evolution of the Lindbladian in \cref{eq:lindbladian}. Then, there exists a threshold $\gamma_{th}$, depending on the circuit, such that if $\gamma \geq \gamma_{th}$, the energy of the system remains constant throughout. Furthermore, for any constant $\gamma$, there exists a circuit such that the energy of the dissipative system grows exponentially fast.
\end{proposition}

\begin{proof}
We can show that $\mathcal{E}^\dagger (\hN) = e^{-\gamma} (\hN + \hI/2) - \hI/2$. Therefore, the recursive equation we obtain for energy is 
\asp{E_{t+1} \leq e^{-\gamma}(\alpha E_t + \beta)}
for suitable constants $\alpha, \beta$. If $e^{-\gamma} \leq \alpha$ then the energy of the system does not grow beyond a constant. 

To prove the second part of the lemma, we consider a circuit composed of squeezing gates with a large enough parameter.
\end{proof}
\subsection{Gaussian gates and the cubic phase gate: Doubly exponential lower bound}
\label{sec:doubly-exp-growth}

We next study the rate of energy growth for circuits composed of Gaussian and cubic gates. We first prove the following doubly exponential lower bound on energy growth for circuits composed of Gaussian and cubic gates.
\begin{lem}
    Consider the following unitary 
    $$
    \hU_t = (\hF\hV^{(3)}(\theta) )^t
    $$
    as an alternation between the Fourier transform and cubic gates and let $\ket{\psi_t} = \hU_t \ket{0}$. Then 
    $$
    \bra{\psi_t} \hN \ket{\psi_t} \geq \frac1{4e}(\frac{\theta^2}{2 e})^{2^t-1} 2^{t 2^t}-\frac12.
    $$
    \label{lem:cubic-energy-growth}
\end{lem}

\begin{proof}

Using the Heisenberg picture $\overline{N_t} := \bra{\psi_t} \hN \ket{\psi_t} =: \bra{0} \hN_t \ket{0}$, where $\hN_t := \hU^\dagger_t \hN \hU_t$. We note $\hN_t = \frac{1}{2} (\hX_t^2 + \hP_t^2 - \hI)$, where $\hX_t = \hU^\dagger_t \hX \hU_t$ and $\hP_t = \hU^\dagger_t \hP \hU_t$. Let $\hA(\theta) :=  \hF\hV^{(3)}(\theta)$, therefore $\hN_t = \hA^{\dagger t}(\theta)\hN \hA^t (\theta)$. We observe $\hP_t = (\hA^\dagger (\theta))^{t-1} (\hA^\dagger(\theta) \hP \hA(\theta)) (\hA(\theta))^{t-1} = - \hX_{t-1}$. Therefore, $\overline{N_t} =\frac12( \overline {X^2_t} + \overline {X^2_{t-1}}-\frac12)\ge\frac12\overline {X^2_t}$, where $\overline {X^2_t} = \bra{0} \hX^2_t \ket{0}$. Therefore, it is enough to lower bound $\overline {X^2_t}$. Let $\ket{v_t}  := \hX_t \ket{0}$, then $\overline {X^2_t} = \|v_t\|^2$. Next we derive a recursive equation for $\hX_t$:
\asp{
\hX_{t+1} = \theta \hX_t^2 - \hX_{t-1}.
}
The proof is by observing $\hA^\dagger(\theta) \hX \hA(\theta) = \hP + \theta \hX^2$ and $\hA^\dagger(\theta) \hP_t \hA= - \hX_{t-1}$. Next we show
\begin{proposition}
    For $t \geq 1$, $\hX_t = \theta^{2^t-1} \hX^{2^t} + \hat{g}_t(\hX, \hP)$, where $\hat g$ is a polynomial of total degree at most $2^t-1$ in $\hX$ and $\hP$. Furthermore $\hat g_t$ has degree at most $2^t -2$ in $\theta$. 
\label{prop:recursive-degree}
\end{proposition}

\begin{proof}
    The proof is by induction. Assume $1 \leq t \leq s$, $\hX_t = \theta^{2^t-1} \hX^{2^t} + \hat{g}_t(\hX, \hP)$. This hypothesis is correct for $t = 1$, since $\hX_1 = \theta \hX^2 + \hP$. We now show that the hypothesis is valid for $t = s + 1$. This is correct because 
    \asp{\hX_{s+1} &= \theta \hX^2_s - \hX_{s-1}\\
    &= \theta ( \theta^{2^s-1} \hX^{2^s} + \hat{g}_s(\hX, \hP))^2 - \hX_{s-1}\\
    &= \theta^{2^{s+1}-1} \hX^{2^{s+1}} + \hat{g}_s(\hX, \hP)^2 + \theta^{2^s-1} \hat{g}_s(\hX, \hP)  \hX^{2^s} + \theta^{2^s-1} \hX^{2^s} \hat{g}_s(\hX, \hP)    - \hX_{s-1}\\}
    Since the degree of $\hX_{s-1}$ is at most $2^{s-1}$ in $\hX$ and $\hP$ and degree at most $2^{s-1}-2$ in $\theta$ and that the degree of $\hat g_s$ is at most $2^s-1$ in $\hX$ and $\hP$ and at most $2^s-2$ in $\theta$, it is evident from the last expression that the hypothesis holds for $t = s +1$. 
\end{proof}

Next we show 
\begin{proposition}
    Let $\hQ$ be a polynomial of degree at most $d$ in $\hX$ and $\hP$. Then the support of $\hQ\ket{0}$ on the Fock basis is at most $d$.
    \label{prop:support}
\end{proposition}

\begin{proof}
    Since the degree of $\hQ$ is at most $d$, it can be written in the normal form $\hQ = \sum_{\substack{k,l \geq 0\\k + l = d}}\alpha_{k,l}\ha^k \ha^{\dagger l}$. Therefore $\hQ \ket{0}$ can only have support up to $d$ corresponding to the monomial $\ha^{\dagger d}$. 
\end{proof}

Now we combine \cref{prop:support} and \cref{prop:recursive-degree} we find that $\ket{v_t} = \theta^{2^t-1} \hX^{2^t} \ket{0} + \ket{w_t}$ such that the support of $\ket{w_t}$ is on Fock basis up to $2^t-1$. Now let $\hat \Pi_{\geq d}$ be the projector on Fock basis states with Boson number at least $d$. Therefore
\asp{
\overline{X^2_t} &= \bra{v_t}\ket{v_t}\\
&\geq \bra{v_t}\Pi_{\geq 2^t}\ket{v_t}\\
&= \theta^{2^{t+1}-2} \bra{0} \hX^{2^t}\hat \Pi_{\geq 2^t} \hX^{2^t}\ket{0}.
}
Next, we observe 
\asp{\hat \Pi_{\geq 2^t} \hX^{2^t} \ket{0} &= \hat \Pi_{\geq 2^t} (\frac{\ha + \ha^\dagger}{\sqrt{2}})^{2^t} \ket{0},\\
&= (\frac{\ha^\dagger}{\sqrt{2}})^{2^t} \ket{0},\\
&= \sqrt{\frac{(2^t)!}{2^{2^t}}} \ket{2^t}.}
Therefore
\asp{\overline{X^2_t} \geq \theta^{2^{t+1}-2}  \frac{(2^t)!}{2^{2^t}}.}
Hence, the bound claimed in the lemma using $n! \geq (n/e)^n$.
\end{proof}

We now show that even under the effect of dissipation under the Lindbladian in \cref{eq:lindbladian}, the photon number grows at least doubly exponentially fast at any constant rate $\gamma$. 

\begin{proposition}
    Let $\hU = \hF \hV^{(3)} (\theta)$, and let $\mathcal A (\cdot) = \mathcal{E} \circ (\hat{U} (\cdot) \hat{U}^\dagger)$ and let 
    $\overline{N_t} := \Tr (\hN \mathcal{A}^t (\ket{0}\bra{0}))$, then
    $$
    \overline{N_t} \geq \Omega((\theta e^{-\gamma})^{2^t} 2^{t 2^t}).
    $$
\end{proposition}

\begin{proof}
    The proof is very similar to the proof of \cref{lem:cubic-energy-growth}. We need to write a new recursive formula. We first note that under $\mathcal E$ an operator $\hQ$ evolves as $\frac{d\hat Q}{dt} = \mathcal L^\dagger [\hat Q] = \gamma (\ha^\dagger \hat Q \ha - \frac{1}{2} \{\hQ, \hN\})$. We can show that for any $\mu, \nu \in \mathbb{Z}_{\geq 0}$
    \asp{
    \mathcal{L}^\dagger [\ha^{\dagger \mu} \ha^{\nu}] &= \gamma (\ha^{\dagger (\mu+1)} \ha^{\nu+1} - \frac{1}{2} (\ha^\dagger \ha \ha^{\dagger \mu} \ha^{\nu} + \ha^{\dagger \mu} \ha^{\nu}\ha^\dagger \ha))\\
    &= \frac{-\gamma (\mu + \nu)}{2} ( \ha^{\dagger \mu} \ha^{\nu}).\\
    }
Therefore $\mathcal{E}^\dagger [\hX] = e^{-\frac{\gamma}2} \hX$ and $\mathcal{E}^\dagger [\hP] = e^{-\frac{\gamma}2} \hP$ and $\mathcal{E}^\dagger [\hX^2] = e^{-\gamma} \hX^2 + (1 - e^{-\gamma})\hI$.

Monomial observables $\hX^\mu\hP^\nu$ evolve as
$$
\mathcal{E}^\dagger [\hX^\mu \hP^\nu] = e^{- \gamma \frac{\mu + \nu}{2}} \hX^\mu \hP^\nu.
$$
In particular for any polynomial $\hat g$ in $\hX$ and $\hP$ we observe $\mathcal E^\dagger (\hat g (\hat X, \hP)) = \hat g (\mathcal E^\dagger (\hat X), \mathcal E^\dagger (\hP))$. This implies that for any $s \geq 0$, $\mathcal A^s (\hat g (\hat X, \hP)) = \hat g (\mathcal A^s (\hat X), \mathcal A^s (\hP))$.

Now we write 
\asp{
\hX_t &= \mathcal {A}^t [\hX]\\
&= \mathcal {A}^{t-1} [\mathcal A(\hX)]\\
&= \mathcal {A}^{t-1} [\mathcal E^\dagger(\hP + \theta \hX^2)]\\
&= \mathcal {A}^{t-1} [e^{-\frac{\gamma}{2}}\hP + \theta e^{-\gamma}\hX^2 + (1 - e^{-\gamma}) \hI]\\
&= -\mathcal {A}^{t-2} [e^{-\gamma}\hX] + \theta e^{-\gamma}\mathcal{A}^{t-1}[\hX^2] + (1 - e^{-\gamma}) \hI\\
&= - e^{-\gamma}\hX_{t-2} + \theta e^{-\gamma}\mathcal{A}^{t-1}[\hX^2]+ (1 - e^{-\gamma}) \hI.\\
\label{eq:rec-diss-cubic-1}
}
Next, we use the following results about super operators: First, we use the following lemma to argue that $\mathcal{E}^\dagger$ is completely positive and unital.
\begin{lem}
  $\mathcal{E}^\dagger$ is CP and unital. 
    \label{lem:unital-CP}
\end{lem}

\begin{proof}
    $\mathcal{E}$ has explicit Kraus decomposition $\mathcal{E} (\cdot) = \sum_j \hK_j (\cdot) \hK_j^\dagger$. The conjugate channel corresponding to $\mathcal{E}^\dagger (\cdot) = \sum_j \hK^\dagger_j (\cdot)  \hK_j$ is CP because $\mathcal{E}$ is CP. CP $\mathcal{E}$ with Kraus operators $\hK_j$ is TP iff $\sum_j \hK^\dagger_j  \hK_j = \hI$. Hence CP $\mathcal{E}$ is TP iff $\mathcal{E}^\dagger$ is unital since $\mathcal{E}^\dagger (\hI) = \sum_j \hK^\dagger_j  \hK_j = \hI$. 
\end{proof}
Next, we use the following result that indicates that if a superoperator is CP and unital, then it satisfies a quantum analog of Jensen's inequality:
\begin{lem} [Kadison's inequality]
    If $\mathcal{A} : L (A) \rightarrow L(A)$ is CP and unital then for a Hermitian matrix $\hR$, $\mathcal{A} (\hR^2) \geq (\mathcal{A} [\hR])^2$.
    \label{lem:Q-Jensen}
\end{lem}

\begin{proof}
    If $\mathcal{A}$ is CP and unital, then using Stinespring's dilation, there exists an isometry and a Hilbert space $\hV : A \rightarrow B$ satisfying $\hV^\dagger \hV = \hI$ and a linear representation $\hat \pi : B \rightarrow B$ such that $\mathcal{A} (\cdot) = \hV^\dagger \hat \pi (\cdot) \hV$. An operator $\hat \pi$ is a linear representation if $\hat{\pi}$ is linear and it preserves multiplication, i.e.~$\forall \hX, \hY\in L (B)$ $\hat{\pi} (\hX \hY) = \hat{\pi} (\hX) \hat{\pi} (\hY)$, and conjugation preserving, i.e.~$\forall \hX \in L(B), \hat{\pi} (\hX^\dagger) = \hat{\pi} (\hX)^\dagger$. 
    
    Now for a Hermitian matrix $\hR$ define 
    \asp{\hY (\hR) = \pi (\hR) \hV -\hV \mathcal{A} (\hR).}
    Therefore
    \asp{0 &\leq \hY^\dagger (\hR) \hY (\hR)\\
    &= ( \hV^\dagger \pi (\hR) - \mathcal{A} (\hR) \hV^\dagger) (\pi (\hR) \hV -\hV \mathcal{A} [\hR])\\
    &= \hV^\dagger \pi (\hR^2) \hV - \hV^\dagger \hat{\pi} (\hR) \hV \mathcal{A} [\hR] - \mathcal{A} [\hR] \hV^\dagger \hat{\pi} (\hR) \hV + \mathcal{A} [\hR]^2\\
    &= \mathcal{A} [\hR^2] -  \mathcal{A}^2 [\hR].}
\end{proof}

By applying \cref{lem:unital-CP} and \cref{lem:Q-Jensen} to \cref{eq:rec-diss-cubic-1} we obtain
\asp{\hX_t &\geq - e^{-\gamma}\hX_{t-2} + \theta e^{-\gamma} \hX_{t-1}^2.}
The rest of the proof is almost identical to the proof of \cref{lem:cubic-energy-growth}.

\end{proof}

 This result shows that there exists a constant $c$ such that if $\gamma < (1-\epsilon) t$ for some $0< \delta <1$ then the expected number of particles grows $\gtrsim 2^{\delta 2^t}$, i.e.~doubly exponentially fast in the number of cubic gates. Next, we show that if the rate of dissipation is slightly larger than linear, then the expected particle number of the system does not grow. In other words, for a circuit composed of Gaussian gates and $t$ cubic gates, there exists a constant $c'$ such that if $\gamma >  t + c'$ then energy ceases to grow.

\begin{proposition}
    For a dissipative system, where at each step we apply a Gaussian gate, followed by a dissipative channel according to the one-step evolution of the Lindbladian, then a cubic gate, then a dissipative channel, and so on there exists a constant $c'$  depending on the circuit such that if threshold $\gamma \geq c' + t$ then the energy of the system remains bounded by a constant.
\end{proposition}

\begin{proof}
We work in the Heisenberg picture with the monomial basis for an observable as

\begin{equation}
\hat{O}_{j,k} :=  \hat{a}^j(\hat{a}^\dagger)^k, \qquad \deg(\hat{O}_{j,k}) := j+k
\end{equation}
For the amplitude damping channel of rate $\gamma$, the adjoint map acts diagonally 
\begin{equation}
\mathcal{L}^\dagger(\hat{O}_{j,k})= e^{-\frac{\gamma}{2}(j+k)} (\hat{O}_{j,k})
\end{equation}
Let $m$ be the number of cubic gates encountered. Each Gaussian gate preserves degree, and each cubic gate can at most double it \cite{chabaudBosonicQuantumComputational2025}.  
After $m$ cubic gates, one has
\begin{equation}
  \deg(\hat{N}) \le 2^m
  \label{eq:deg}
\end{equation}
\begin{equation}
\hat{N} = \sum_{j,k} c^{(m)}_{j,k} \hat{O}_{j,k},
\qquad
A_m := \sum_{j,k} |c^{(m)}_{j,k}|
\end{equation}
Suppose a polynomial observable of degree $\le d$ is conjugated by any Gaussian or cubic gate whose parameters are polynomially bounded. In that case, the total coefficients grow by at most $A \rightarrow e^{C_0 d} A$ for some constant $C_0>0$ (depending only on the range of gate parameter). Applying the loss channel then multiplies the contribution of degree-$d$ terms by $e^{-\frac{\gamma}{2}d}$. Thus, if the degree after the $t$-th cubic gate is $2^t$, one obtains the recursion
\begin{equation}
  A_t \le \exp\Bigl((C_0 - \frac{\gamma}{2}) 2^t\Bigr) A_{t-1}.
  \label{eq:recursion}
\end{equation}
Iterating,
\begin{equation}
A_m \le
\exp\Bigl((C_0 - \frac{\gamma}{2}) \sum_{t=1}^m 2^t \Bigr) A_0 
\label{eq:Am}
\end{equation}
Since $A_0=1$ and $\sum_{t=1}^m 2^t = 2^{m+1}-2$, we have
\begin{equation}
  A_m \le 
  \exp\Bigl((C_0 - \frac{\gamma}{2})(2^{m+1}-2)\Bigr)
  \label{eq:Am-final}
\end{equation}
For the photon number expectation, only the diagonal terms $\hat{O}_{d,d}=\hat{a}^d \hat{a}^{\dagger d}$ survive on $|0\rangle$.  
They satisfy
\begin{equation}
  \langle 0| (\hat{a})^d \hat{a}^{\dagger d} |0\rangle = d! \le d^d 
  \label{eq:vac}
\end{equation}
Hence
\begin{align}
\langle \hat{N}\rangle&= \langle 0|\hat{N}|0\rangle=\sum_d c^m_{d,d} d! \notag \\
&\le\ \Big(\max_{d\le 2^m} d!\Big ) \sum_d c^m_{d,d} \notag \\
&\le\ A_m \cdot \Big(\max_{d\le 2^m} d! \Big )  \notag \\
&\le\ A_m \cdot (2^m)^{2^m}
\label{eq:N-bound}
\end{align}
Combining \eqref{eq:Am-final} and \eqref{eq:N-bound} yields
\begin{align}
\langle \hat{N}\rangle & \le \exp\Bigl((C_0 - \frac{\gamma}{2})(2^{m+1}-2)\Bigr)\cdot \exp\bigl(2^m \log(2^m)\bigr) \notag \\
&= \exp\Bigl(-\frac{\gamma}{2}(2^{m+1}-2)+ C_0 (2^{m+1}-2)+ m 2^m\Bigr)
\label{eq:N-bound1}
\end{align}
where we used $(2^m)^{2^m}=\exp\bigl(2^m \log(2^m)\bigr)=\exp(m2^m)$. Now, we want to obtain $\gamma_{th}$ for the energy upper bound.

Starting from \eqref{eq:N-bound1}, let the exponent be
\begin{equation}
B(m)=-\tfrac{\gamma}{2}(2^{m+1}-2)+C_0(2^{m+1}-2)+m 2^m.
\end{equation}
Rewrite $2^{m+1}-2=2\cdot 2^m-2$ and factor $2^m$, we have
\begin{align}
B(m)
&=\Big(-\tfrac{\gamma}{2}+C_0\Big)(2\cdot 2^m-2)+m 2^m \notag\\
&=\underbrace{\big(-\gamma+2C_0+m\big)}_{\text{coeff.\ of }2^m}2^m
+\underbrace{(\gamma-2C_0)}_{\text{constant in}m}.
\end{align}
If we choose 
\begin{equation}
\gamma \ge 2C_0 + m + c\qquad(c>0)
\end{equation}
then
\begin{equation}
B(m)\le -c 2^m + (m+c),
\end{equation}
which is negative for all sufficiently large $m$ (where $-c2^m$ dominates $m+c$, so, $B(m) <0$). Hence, the RHS of \eqref{eq:N-bound1} decays uniformly and the energy remains bounded. Therefore we have
\begin{equation}
\gamma \ge 2C_0 +m +O(1).
\end{equation}
Thus, a linear scaling $\gamma = \Theta(m)$ suffices to ensure that
\begin{equation}
  \sup_m \langle \hat{N}\rangle = O(1).
\end{equation}
\end{proof}

Finally, we refine the upper bound on the expected particle number of circuits composed of cubic and Gaussian gates. Specifically, Proposition 4.3 of \cite{chabaudBosonicQuantumComputational2025} shows that the expected particle number in cubic and Gaussian circuits grows at most doubly exponentially fast in the number of cubic gates. Here we show that this number grows at most doubly exponentially fast in the cubic depth meaning the number of layers of cubic gates in a multi-mode CV circuit. Therefore expected particle number in circuits of polylogarithmic depth grows at most exponentially fast, and in circuits of polyloglog depth it grows at most polynomially. This result is captured in the following proposition:
\begin{proposition}\label{prop:E-growth-cubic-depth}
    Let $C$ be a circuit composed of $m$ modes with cubic-depth $d$, with each layer composed of at most $s$ gates (specified with constant bits of precision), then the expected particle number is at most $2^{O(sd 2^d)}$.
\end{proposition}

\begin{proof}
    We evolve the total particle number operator $\hN_{tot} := \hN_1 + \ldots + \hN_m$. We are interested in $\overline{N_{tot}} := \bra{0^m}\hC^\dagger \hN_{tot} \hC \ket{0^m}$. For $1 \leq k \leq d$, let $\hN_{tot}^{(k)}$ be the expected particle number after the application of the $k$'th layer of cubic gates. We consider the following expansion of this operator into the particular creation and annihilation operators
    $$
    \hN_{tot}^{(k)} =: \sum_{\boldsymbol{\mu, \nu} \in \bbZ_{\geq^m}} N^{(k)}_{\boldsymbol{\mu, \nu}} \ha^{\dagger \boldsymbol{\mu}} \ha^{\boldsymbol{\nu}}.
    $$
    We bound the degree $D_k$ and the largest coefficient $\alpha_k$ $\hN_{tot}^{(k)}$. Then $\overline{N_{tot}} = N^{(d)}_{\boldsymbol{0, 0}} \leq \alpha_d$. We first show $D_k \leq 2^k$. This is done by induction on $k$. Each layer consists of a layer of Gaussians and a layer of cubic gates. The Gaussian layer does not change the degree, and the cubic layer at most doubles the total degree. Next, we use this bound on the degree to bound the magnitude of the coefficients. Under the action of any gate $\hat g$ (cubic or Gaussian), the monomials get mapped according to 
    \asp{
    \hat g : \ha^{\dagger \boldsymbol{\mu}} \ha^{\boldsymbol{\nu}} \mapsto &(\hat g \ha^\dagger_1  \hat{g}^\dagger )^{\mu_1} \ldots  (\hat g \ha^\dagger_m  \hat{g}^\dagger )^{\mu_m} (\hat g \ha_1  \hat{g}^\dagger )^{\nu_1} \ldots  (\hat g \ha_m  \hat{g}^\dagger )^{\nu_m}\\
    &= \sum_{\substack{\boldsymbol{\alpha, \beta} \in \bbZ_{\geq^m}\\ \|\alpha\|_1 + \|\beta\|_1 \leq 2 D_k}} A_{\boldsymbol{\alpha, \beta}} \ha^{\dagger \boldsymbol{\alpha}} \ha^{\boldsymbol{\beta}}.
    }
    We claim that $\max_{\alpha, \beta} A_{\boldsymbol{\alpha, \beta}} \leq D_k^{O(D_k)}$. Therefore, if the number of gates in each layer is $s$, the largest coefficient in the expansion of $\hN^{(k)}_{tot}$ gets multiplied by $O(D_k!)^s$. Therefore the largest coefficient after $d$ layers becomes $(D_1^{O(D_1)} \ldots D_d^{O(D_d)})^s \leq 2^{O(sd 2^d)}$.  It remains to show that $\max_{\alpha, \beta} A_{\boldsymbol{\alpha, \beta}} \leq O(D_k!)$. Without loss of generality, we can assume that each Gaussian gate is either a linear optical element, a rotation, a displacement, or a squeezing. If $\hat g$ is a linear optical element, then it does not change the number operator. If it is a displacement, then it preserves the normal ordering, and since each gate is specified with constant bits of precision, it only introduces a multiplicative factor of $e^{O(D_k)}$. Now, if we apply a squeezing operation on mode $j$, we map $\ha_j^{\mu_j} \mapsto (\alpha \ha_j + \beta \ha^\dagger_j)^{\mu_j}$ which multiplies coefficients by at most $D_k! e^{O(D_k)}$. We can show this using induction. Let $\alpha_l$ be the largest magnitude of the coefficients of  $(\alpha \ha_j + \beta \ha^\dagger_j)^{l}$ in the normal order. Let $(\alpha \ha_j + \beta \ha^\dagger_j)^{l} = \sum_{\mu, \nu} \alpha_{\mu \nu} \ha_j^{\dagger \mu}  \ha_j^{\nu}$. Then $(\alpha \ha_j + \beta \ha^\dagger_j)^{l+1} = \sum_{\mu, \nu} \alpha_{\mu \nu} \ha_j^{\dagger \mu}  \ha_j^{\nu} (\alpha \ha_j + \beta \ha^\dagger_j)$. Since $[\ha^\nu, \ha^\dagger] = \nu \ha^{\nu-1}$, we get $\alpha_{l+1} \leq \max\{|\alpha|, |\beta|\} \cdot \nu \alpha_l$. We can use a similar reasoning for the cubic gate. We note that under the action of the cubic gate $\ha \mapsto \ha + \gamma/\sqrt{2} \hX^2 = \ha + \ha^2 \frac{\gamma}{2 \sqrt 2} + \ha^{\dagger 2} \frac{\gamma}{2 \sqrt 2} + \frac{\gamma}{2\sqrt 2} (2 \ha^\dagger \ha + 1)$. Now suppose that $(\ha + \gamma/\sqrt{2} \hX^2)^{l} = \sum_{\mu, \nu} \beta_{\mu \nu} \ha_j^{\dagger \mu}  \ha_j^{\nu}$ and $\beta_l$ be the maximum magnitude of a coefficient. Using a reasoning similar to what we did in the case of the squeezing gates, $\beta_{l+1} \leq \gamma D_k^2 \beta_l$. Therefore $\beta_k \leq k^{O(k)}$. 
\end{proof}

\subsection{Infinite energy in finite time}\label{sscn:infinitefinite}

We now show various settings in which one can achieve infinite energy in finite time.
Note, we are not claiming that this is a realistic physical model.
Rather, these statements show that one cannot hope to realize general polynomial Hamiltonians in the real world.

\begin{proposition}\label{prop:infinite_energy}
  The state $\ket{\psi} = e^{-itH}\hS_0(r)\ket{0}_0\ket{0}_1$ with $r>0$ and $H = \frac{i}{2}\N_0(\a_1^2 - \a_1^{\dagger2})$ has infinite energy for $t \ge -\frac12 \ln(\tanh r)$.
\end{proposition}
\begin{proof}
  First, observe that $H$ is essentially self-adjoint since $H = \sum_{n} n\ketbra{n}\otimes \frac i2(\a_1^2 - \a_1^{\dagger2})$, which is effectively a direct sum of Gaussian Hamiltonians.
  Therefore $\ket{\psi}$ is well-defined.
  Recall that the squeezed state is
  \begin{equation}
    \hS(r)\ket{0} = e^{\frac12(r\a^2 - r\a^{\dagger2})}\ket{0} =\sum_{n=0}^\infty c_{2n}\ket{2n}, \quad c_{2n} = \frac{1}{\sqrt{\cosh r}}(-\tanh r)^n \frac{\sqrt{(2n)!}}{2^nn!},
  \end{equation}
  with average photon number $\sinh^2(r)$.
  Therefore,
  \begin{equation}
    \ket{\psi} = \sum_{n} c_{2n}\ket{2n}\otimes \hS(2tn)\ket{0}.
  \end{equation}
  With Stirling's approximation $n! = \sqrt{2\pi}n^{n+1/2}e^{-n} e^{\Theta(1/n)}$ \cite{Robbins55}, we get (assuming $r=\Theta(1)$)
  \begin{equation}
    \abs{c_{2n}}^2 = \frac{1}{\cosh r}\frac{(2n)!}{(n!)^22^{2n}}\tanh^{2n} r \in \Theta\!\left(\frac{\tanh^{2n}r}{\sqrt{n}}\right).
  \end{equation}
  The average photon number in mode $1$ is then
  \begin{equation}
    \ev{\N_1}{\psi} = \sum_{n} \abs{c_{2n}}^2\sinh^2(2nt) \in \Theta\!\left(\sum_{n} \frac1{\sqrt n}(\tanh^2 r \cdot e^{4t})^n\right),
  \end{equation}
  which diverges iff $\tanh^2r \cdot e^{4t}\ge 1 \Leftrightarrow t \ge -\frac12 \ln(\tanh r)$.
\end{proof}

\begin{proposition}\label{prop:Dinf}
  The state $\ket{\psi} = e^{-itH}\hD_0(1)\ket{0}_0\ket{0}_1$ with $H = \frac{i}{2}\N_0^2(\a_1^2 - \a_1^{\dagger2})$ has infinite energy for any $t>0$.
\end{proposition}
\begin{proof}
  The displaced vacuum state is given by
  \begin{equation}
    \hD(\alpha)\ket{0} = e^{\alpha \a^\dagger - \alpha^*\a}\ket{0} = e^{-\abs{\alpha}^2/2} \sum_{n=0}^\infty \frac{\alpha^n}{\sqrt{n!}}\ket{n}.
  \end{equation}
  Thus,
  \begin{equation}
    \ket{\psi} = e^{-1/2}\sum_{n=0}^\infty \frac{1}{\sqrt{n!}}\ket{n}\otimes \hS(n^2t)\ket{0}.
  \end{equation}
  The average photon number in mode $1$ is given by
  \begin{equation}
    \ev{\N_1}{\psi} = e^{-1}\sum_{n=0}^\infty\frac{\sinh^2(n^2t)}{n!},
  \end{equation}
  which diverges for any $t>0$ as $e^{n^2t}$ grows much faster than $e^{n\ln n}=n^n\ge n!$.
\end{proof}

We can even achieve infinite energy from the vacuum state with just a single gate.

\begin{proposition}\label{prop:inf_three}
  The state $\ket{\psi} = e^{-itH}\ket{0}_0\ket{0}_1$ with $H = \frac{i}{2}\X_0^4(\a_1^2-\a_1^{\dagger2})$ has infinite energy for any $t>0$.
\end{proposition}
\begin{proof}
  The wave function of the vacuum state is given by $\phi(x) = \pi^{-1/4}e^{-x^2/2}$.
  Then 
  \begin{equation}
    \ket{\psi} = \int_{\RR}dx\, \phi(x)\ket{x}\otimes \hS(tx^4)\ket{0},
  \end{equation}
  and we have energy in mode $1$
  \begin{equation}
    E=\ev{\N_1}{\psi} = \int_{\RR}dx\, \phi^2(x)\ev{\hS^\dagger(tx^4)\N \hS(tx^4)}{0} = \frac{1}{\sqrt{\pi}}\int_{\RR} dx\, e^{-x^2}\sinh^2(tx^4).
  \end{equation}
  For $u\ge 1$, $\sinh^2 u = (e^{2u} -2+e^{-2u})/4 \ge e^{2u}/8$.
  Let $B = (1/t)^{1/4}$, so that $tx^4\ge 1$ for $x\ge B$.
  Thus,
  \begin{equation}
    E \ge \frac{1}{8\sqrt\pi} \int_{x\ge B} dx\,e^{-x^2} e^{2tx^4} \ge \frac{1}{8\sqrt\pi} \int_{x\ge B} dx\,e^{tx^4} = \infty,
  \end{equation}
  as the last integral trivially diverges since the integrant is always $\ge1$.
  $H$ is essentially self-adjoint on the Schwartz space, since it is the tensor product of two essentially self-adjoint operators $\X^4$ and $i(\a^2-\a^{\dagger2})$ \cite[Theorem VIII.33]{Reed1981-kj}.
\end{proof}

Note that these results are not robust to small perturbations, in the following sense. Let $\ket\phi$ be a normalized state and $\hat U$ a unitary operation mapping $\ket\phi$ to a normalized state $\ket\psi$ with infinite energy. Since $\ket\psi$ is normalized, given $\epsilon>0$ there exists $N\in\mathbb N$ such that the truncated state $\ket{\psi_N}$ in Fock basis is $\epsilon$-close to $\ket\psi$ in trace distance. By construction, $\ket{\psi_N}$ has finite energy and $\hat U^\dag\ket{\psi_N}$ is $\epsilon$-close to $\ket\phi$. This shows that even if $\hat U$ maps $\ket\phi$ to a state of infinite energy, there exist states arbitrarily close to $\ket\phi$ mapped to states of finite energy.

\section{Complexity-theoretic lower bounds: Energy is all you need}\label{scn:lowerbounds}

\subsection{Lower bounds via complexity classes}\label{sscn:energydorian}

\input{adiabatic2.tex}

\subsection{Lower bounds via an energy hierarchy theorem}\label{sscn:energyarsalan}

In this section, we develop an energy hierarchy theorem for CV quantum systems. This is a tool that shows us
\begin{center}
\textit{With a higher energy budget, we can compute more.}
\end{center}

Our objective is to separate computations that use energy at most $f(n)$ and those that use $g(n)$ when $f(n) \gg g(n)$. We start this study in the ``unitary'' oracle setting. To do so, we consider the following problem:
\begin{definition}[Beam-splitter Precision Problem]\label{def:beamsplitter}
The $\epsilon$-beam-splitter precision problem \\
($\varepsilon$-\textsc{BeamSplitPrec}) is to decide (with probability at least $2/3$) if a given oracle $O$ is the identity unitary, or a beam-splitter with angle $\varepsilon$, using only a single query to $O$.
Note that the input to the problem is $\varepsilon$ (in binary).
\end{definition}

In what follows, we show that
\begin{theorem}
$\varepsilon$-\textsc{BeamSplitPrec} cannot be solved in $\mathsf{CVBQE}(E)$ if $E = o(\varepsilon^{-1})$, but it can be solved in $\mathsf{CVBQE}(F)$ for some $F = O(\varepsilon^{-2})$.
\end{theorem}
\begin{proof}
We split the proof into two parts:
\begin{itemize}
\item \textit{We show we cannot solve this problem with using $o(\varepsilon^{-1})$ energy}: We begin by proving a bound on the energy-constrained diamond norm of a beam-splitter. Recall that for a linear map $V$, the energy-constrained diamond norm is be defined as
\begin{align}
\norm{V}_{E,\diamond}:= \sup_{\psi: \bra{\psi}N\ket{\psi}\leq E} \norm{(\mathbb I \otimes V) \ket\psi}.
\end{align}
We then show that:
\begin{lem}\label{lem:diamond-norm-bound}
Let $B(\varepsilon) = \exp(i\varepsilon (a_1^\dagger a_2 + a_1 a_2^\dagger))$ represent a beam-splitter with angle $\varepsilon$. Then
\begin{align}
\norm{B(\varepsilon) - \mathbb I}_{E,\diamond} \leq O(\varepsilon E)^{\frac13},
\end{align}
which is $o(\varepsilon)$ whenever $E = o(\varepsilon^{-1})$.
\end{lem}
\begin{proof}
Note that by quantum Markov inequality \cite[Lemma 4.1]{chabaudBosonicQuantumComputational2025} for any $\psi$ such that $\bra{\psi}N\ket{\psi}\leq E$, we have
\begin{align}\label{eq:asymptot-qme}
\norm{\Pi_{E'} \ket\psi - \ket\psi}  \leq \sqrt{E/E'}.
\end{align}

Now, let $W = \mathbb I \otimes (B(\varepsilon) - \mathbb I)$, and note that
\begin{align}\label{eq:bound-via-Pi-E}
\begin{split}
\norm{W\ket\psi} &= \norm{W(\ket\psi - \Pi_{E'}\ket\psi + \Pi_{E'}\ket\psi)}\\
&\leq \norm{W(\ket\psi - \Pi_{E'}\ket\psi)} + \norm{W\Pi_{E'}\ket{\psi}}\\
&\overset{(i)}{\leq} 2 \sqrt{E/E'}+ \norm{W\Pi_{E'}\ket{\psi}},
\end{split}
\end{align}
where $(i)$ is due to $\norm{W}_\infty \leq 2$ together with \eqref{eq:asymptot-qme}. This inequality gives
\begin{align}\label{eq:bound-diamond}
\sup_{\psi: E_\psi \leq E} \norm{W\psi} \leq 2\varepsilon + \sup_{\psi\in\mathcal S_{E'}} \norm{W\psi}
\end{align}
Where we have used the notation $\mathcal S_{E'} = \mathrm{span}\{\ket{n}: \sum_i n_i \leq E'\}$. Next, note that
\begin{align}
\norm{\Pi_{E'}(a_1^\dagger a_2 + a_1 a_2^\dagger) \Pi_{E'}} \leq 2 \norm{\Pi_{E'}a_1^\dag a_2 \Pi_{E'}} \leq 2(E'+1).
\end{align}
To prove the above inequality, note that $a_1^\dagger a_2 \ket{n,m} = \sqrt{(n+1) m} \ket{n+1, m-1}$ and $\sqrt{(n+1) m} \leq E'+1$ provided $n,m\leq E'$. Moreover, $\Pi_{E'} a_1^\dagger a_2 \Pi_{E'}$ is a $1$-sparse matrix. Using the Schur test below, we get the desired bound on the largest singular value.
\begin{lem} Given a $m\times n$ complex-valued matrix $A$, we have:
\begin{align}
\sigma^2_1(A) \leq r \times c,
\end{align}
where
\begin{align}
r = \max_{i\in[m]} \sum_{j\in[n]} |A_{ij}|, \quad c = \max_{j\in[n]} \sum_{i\in[m]} |A_{ij}|.
\end{align}
\end{lem}
\begin{proof}
Recall the following form of the largest singular value
\[
\sigma_1 = \max_{x\in\mathbb R^{m}_m, y\in\mathbb{R}^n} | x^\dagger A y|.
\]
We may now upper-bound the expression on the right hand side
\begin{align}
\begin{split}
|x^\dagger A y| &= \left| \sum_{i\in[m],j\in[n]} x^\ast_i A_{ij} \, y_j \right|\\
&\leq \sum_{i\in[m],j\in[n]} \left| x^\ast_i\, A_{ij} \, y_j \right|\\
&= \sum_{i\in[m],j\in[n]} \left| x^\ast_i A_{ij}^{1/2}\,  A_{ij}^{1/2} y_j \right|\\
&\overset{(a)}{\leq} \sqrt{\sum_{i\in[m],j\in[n]}  |x_i|^2 |A_{ij}|} \, \sqrt{\sum_{i\in[m],j\in[n]}  |y_j|^2 |A_{ij}|}\\
&\leq \sqrt{\sum_{i\in[m]}  |x_i|^2 r} \, \sqrt{\sum_{j\in[n]}  |y_j|^2 c}\\
&= \sqrt{r} \, \sqrt{c},
\end{split}
\end{align}
where (a) follows from the Cauchy-Schwarz inequality.
\end{proof}

We also point out that since $[\Pi_{E'},a_1^\dagger a_2] = 0$ is a photon-preserving map, we have\footnote{Note that if $\Pi$ is a projection and $A$ is any operator such that $[\Pi,A]=0$, we can write $\Pi f(A) \Pi = f(\Pi A \Pi)$. To see this note that $\Pi A^2 \Pi = \Pi A^2 \Pi^2 = \Pi A \Pi A \Pi = \Pi A \Pi^2 A \Pi = (\Pi A \Pi)^2$. We can inductively show that $\Pi A^k \Pi = (\Pi A \Pi)^k$.}
\begin{align}
\Pi_{E'} \left(\exp(i\varepsilon(a_1^\dagger a_2 + a_1 a_2^\dagger)) - \mathbb I \right) \Pi_{E'} = \exp(\varepsilon \Pi_{E'}(a_1^\dagger a_2 + a_1 a_2^\dagger)\Pi_{E'}) - \Pi_{E'},
\end{align}
and hence
\begin{align}\label{eq:last-piece-of-lem41}
\begin{split}
\norm{\Pi_{E'} \left(\exp(i\varepsilon(a_1^\dagger a_2 + a_1 a_2^\dagger)) - \mathbb I \right) \Pi_{E'}} &= \norm{\exp(\varepsilon \Pi_{E'}(a_1^\dagger a_2 + a_1 a_2^\dagger)\Pi_{E'}) - \Pi_{E'}}\\
&\overset{(i)}{\leq} \exp(\varepsilon \norm{\Pi_{E'}(a_1^\dagger a_2 + a_1 a_2^\dagger})\Pi_{E'}) - 1\\
&\leq \exp(2\varepsilon (E'+1)) - 1 = O(\varepsilon E'),
\end{split}
\end{align}
where $(i)$ follows by Taylor expanding the exponential, using the triangle inequality, and summing the terms back again. This implies that 
\begin{align}\label{eq:bound-2-diam}
\sup_{\psi\in\mathcal S_{E'}} \norm{W\psi} \leq O(\varepsilon E).
\end{align}

Putting \eqref{eq:bound-2-diam} together with \eqref{eq:bound-via-Pi-E} and \eqref{eq:bound-diamond} gives
\begin{align}
\norm{W}_{E, \diamond} \leq 2\sqrt{E/E'} + O(\varepsilon E').
\end{align}
Choosing $E'=\frac{E^{\frac13}}{\varepsilon^{\frac{2}{3}}}$ gives us the desired result.
\end{proof}
The rest of the proof is clear. Assume you have a circuit with energy at most $E$, and making a query to $O$. The possible outputs of the circuits (i.e.~with either $O = \mathbb I$ or $O = B(\varepsilon)$) are at most $O((\varepsilon E)^{\frac13})$-far. Hence, there is a constant $c>0$ such that no circuit with $E \leq c(\varepsilon^{-1})$ cannot solve the problem with a single query.

\item \textit{We can solve the problem with a single query, using $c'\varepsilon^{-2}$ energy}: To this end, note that
\begin{align}
\ket{n,0} = \frac{a_1^n{}^\dagger}{\sqrt{n!}} \ket{0,0},
\end{align}
and since $B(\varepsilon)^\dagger a_1 B(\varepsilon) = \cos(\varepsilon) a_1 + \sin(\varepsilon) a_2$, we have
\begin{align}
\begin{split}
B(\varepsilon)\ket{n,0} &= \frac{1}{\sqrt{n!}}(a_1^\dagger \cos(\varepsilon) + a_2^\dagger \sin(\varepsilon))^n \ket{0,0}\\
&= \frac{1}{\sqrt{n!}} \sum_{k} {n\choose k} \cos^k(\varepsilon)\sin^{n-k}(\varepsilon) a_1^k{}^\dagger a_2^{n-k}{}^\dagger \ket{0,0},
\end{split}
\end{align}
and therefore
\begin{align}
\begin{split}
\bra{n,0} B(\varepsilon)\ket{n,0} &= \frac{1}{\sqrt{n!}} \sum_{k} {n\choose k} \cos^k(\varepsilon)\sin^{n-k}(\varepsilon) \bra{n,0} a_1^k{}^\dagger a_2^{n-k}{}^\dagger \ket{0,0}\\
&= \cos^{n}(\varepsilon),
\end{split}
\end{align}
where the last line follows from the fact that the only non-zero term on the line above is $k=n$ (i.e.~$\bra{n,0} a_1^k{}^\dagger a_2^{n-k}{}^\dagger \ket{0,0} = 0$ whenever $k<n$). Finally, note that for $|\varepsilon| \leq 2.7$:
\begin{align}
\cos^n(\varepsilon) \leq (1-\frac{\varepsilon^2}{4})^n \leq \exp(-\frac{n\varepsilon^2}{4}),
\end{align}
showing that choosing $n\geq \frac{2}{\varepsilon^2} \log(\frac{1}{\delta})=:N$ solves the problem with probability at least $1-\delta$. 

Recall that 
\begin{align}\label{eq:tms}
T_\xi\ket{0,0} = \exp(\xi(a_1 a_2 - a_1^\dagger a_2^\dagger))\ket{0,0} = \mathrm{sech}(\xi) \sum_{n} \tanh(\xi)^n \ket{n,n}.
\end{align}
and therefore, letting $\xi = O(\log \frac{1}{N})$ and measuring one mode of it, we would, with probability $\geq 2/3$ get the photon number larger than $c(\frac{1}{N})$. Note that we do not really need to measure any mode here, but we can model the measurement of the first mode in \cref{fig:bs-solver} as if there was a high-energy state coming to the bottom mode of the beam-splitter. Therefore, reading the first mode, will reveal if the gate was a beam-splitter or an identity gate.
 
\begin{figure}[t]
    \centering
\begin{quantikz}
\lstick{$\ket{0}$} & \qw            & \ctrl{1}   & \meter{} \\
\lstick{$\ket{0}$} & \gate[wires=2]{\exp(\xi(a_2a_3 - a_2^\dagger a_3^\dagger))} & \targX{}   &  \\
\lstick{$\ket{0}$} &                & \qw        & \qw
\end{quantikz}
\caption{The circuit solving the $\varepsilon$-\textsc{BeamSplitPrec} problem with $O(\varepsilon^{-2})$ energy. Note that the first gate is a two-mode squeezer with a parameter $\xi$ that produces large enough photon number on average (c.f. \eqref{eq:tms}). The second gate is a beam-splitter, with angle either $\varepsilon$ (No case) or angle $0$ (Yes case). The measurement is performed in the number basis. Note that we accept if $n=0$ as the measurement outcome, and otherwise, we reject.}
    \label{fig:bs-solver}
\end{figure}

\end{itemize}
\end{proof}

\subsection{Undecidable problems related to CV systems}\label{sscn:infenergydorian}

\input{infinite-energy.tex}

\section{Complexity-theoretic upper bounds: Classical and quantum simulation algorithms}
\label{scn:simulations_upperbounds}

\subsection{\CVBQP\ with polynomially bounded energy}
\label{sscn:cvbqppoly}
We now show that \CVBQP\ with polynomial energy bounds is in $\BQPspoly$, i.e.~$\CVBQPpoly\subseteq\BQPspoly$, regardless of the gate set. With an additional bounded energy block encoding assumption for our gate set (\Cref{asmp:oracle-for-truncated-unitary}), this is strengthened to $\CVBQPpoly\subseteq\BQP$.

We begin by stating our low energy block encoding assumption. We assume we have a parameterized gate set of bosonic unitaries and that for each gate $g$, we can efficiently construct a block-encoding (as defined in, e.g.~\cite{gilyenQuantumSingularValue2019}) representing the truncation of that unitary with polynomial energy as follows.

\begin{assumption}\label{asmp:oracle-for-truncated-unitary}
Consider a parametrized bosonic gate set $\mathcal G = \{g_1(\vec{\theta_1}), \cdots, g_n(\vec{\theta_n}) \}$, each acting on constantly many sites. We assume that we have a subroutine that takes as input the description of the gate with its parameters, say $g_i(\vec{\theta_i})$, an energy bound $E$, and prints out an $\epsilon$-accurate description of a unitary $U_i$
\begin{align}
U_i = \begin{pmatrix}
\Pi_E^{\otimes k} g_i(\vec\theta_i) \Pi_E^{\otimes k} & \cdot \\
\cdot & \cdot
\end{pmatrix},
\end{align}
in time $\mathsf{poly}(E, \frac1{\epsilon})$. Here, $\Pi_E$ denotes the projector onto the span of Fock states $\ket{0},\ldots,\ket{E}$, and $\cdot$ indicates that we do not care about those entries (i.e.~they can be arbitrary). Note that $U_i$ is defined on a Hilbert space of dimension $2 \cdot E^{k}$. 
\end{assumption}

\noindent Note that \cref{asmp:oracle-for-truncated-unitary} holds as long as we can efficiently compute a given entry of our gate sets up to polynomially large energies. There are several gate sets that fall under \cref{asmp:oracle-for-truncated-unitary}: Kerr gate, Gaussians, Jaynes--Cumming interactions, and cubic phase gate \cite[Appendix D]{miatto2020fast}:

\begin{lem}[{\cite[Lemma 4.2]{chabaudBosonicQuantumComputational2025}}]\label{lem:cubic-blocks}
There exist constants $c,c'$ such that there exists a polynomial of degree $d \leq c (m+n)$ in $s^{-\frac13}$ which approximates $e^{-\frac{2}{3s^2}} \bra{m} e^{isX^3/3} \ket{n}$ to within additive error $\frac{1}{d^{c'd}}$, assuming $|\log s| = O(m+n)$.
\end{lem}

\noindent Using this lemma, we can prepare a subroutine outputting an $\epsilon$-accurate block-encoding of cubic gates $e^{isX^3/3}$ in time $\mathsf{poly}(E,\log\frac{1}{\epsilon})$. To do so, we can use the following lemma
\begin{lem}
Let $A\in\mathbb C^{n\times n}$ be a matrix such that $\norm{A}\leq 1$. The matrix
\begin{align}
W = \begin{pmatrix}
A & \sqrt{1-AA^\dagger}\\
\sqrt{1-A^\dagger A} & -A^\dagger
\end{pmatrix}
\end{align}
is a unitary.
\end{lem}
\begin{proof}
Let us consider the singular value decomposition of $A = U D V$. Then
\begin{align}
W &= \begin{pmatrix}
U D V & U \sqrt{1-D^2} U^\dagger\\
V^\dagger \sqrt{1-D^2}V & -V^\dagger A U^\dagger
\end{pmatrix}\\
&= \begin{pmatrix}
0 & V^\dagger\\
U & 0
\end{pmatrix}
\begin{pmatrix}
D & \sqrt{1-D^2}\\
\sqrt{1-D^2} & -D
\end{pmatrix}
\begin{pmatrix}
0 & V\\
U^\dagger & 0
\end{pmatrix}.
\end{align}
It is straightforward to verify that $\begin{pmatrix}
D & \sqrt{1-D^2}\\
\sqrt{1-D^2} & -D
\end{pmatrix}$ is a unitary, and since $W$ is a multiplication of three unitaries, we conclude that $W$ is also a unitary.
\end{proof}

We now present the simulation result.

\begin{theorem}\label{thm:polyenergypolysim}
For any bosonic circuit on $n$ modes, with $k$-local gates, total evolution time $T$, and uniform energy bound $E^\ast$, there exists a DV quantum circuit simulating it of width $O(n \log\frac{E^\ast T}{\delta})$ and depth $O(T E^\ast{}^{2k})$ (consisting of generic $2$-qubit gates) with error at most $\delta$ (in TV distance). Moreover, the simulation succeeds with probability at least $1-2\delta$, outputting a flag for failure. If our gate set satisfies \cref{asmp:oracle-for-truncated-unitary}, a classical description of the DV quantum simulator can be obtained in time $O(\mathsf{poly}(E^\ast{}^k,\frac{1}{\delta}, T))$. 
\end{theorem}

\begin{proof}
Our simulation is based on \cite[Proposition 4.2]{chabaudBosonicQuantumComputational2025}: we follow the same simulation idea, but we do matrix multiplications via block-encoding. We summarize the simulation in the following
\begin{enumerate}
\item Approximate each gate $G_i$ with its projected one $G_{i,E} =\Pi_E^{\otimes k} G_i \Pi_E^{\otimes k}$, and choose $E = 2E^\ast \frac{T^2}{\epsilon^2}$ to get an $\epsilon$-approximation of the entire circuit.
\item Let $\ket{\phi_i} = G_{i,E} \cdots G_{1,E}\ket{0}$ denote the approximated evolution state at time $i$. We need to construct $\ket{\phi_i}$ quantumly. To do so, we apply each approximated gate $G_{i,E}:= \Pi_E^{\otimes k} G_i \Pi_E^{\otimes k}$ via its block-encoding circuit:
\begin{align}\label{eq:be}
\ket{\phi_i}=(\bra{0}\otimes \mathbb I) U_i (\ket{0}\ket{\phi_{i-1}}).
\end{align}
Note that this means we pick an ancilla qubit at state $\ket{0}$, apply $U_i$, and then, we post-select on measuring the ancilla at $\ket{0}$. Using quantum Markov inequality (or the gentle measurement lemma) we show that this step succeeds with probability at least $1-2\sqrt{\frac{2E^\ast}{E}}$. Using a union bound, all gates succeed with probability at least $1-2T\sqrt{\frac{2E^\ast}{E}}$.
\item It remains to compile $U_i$ via elementary gates. We can reduce $U$ into $O( E^\ast{}^{2k})$-many $2$-qubit unitary gates via unitary Gaussian elimination. At this stage, we are describing each mode (that was previously truncated by $2E^\ast \frac{T^2}{\delta^2}$ many levels) by $\log  \frac{2 E^\ast T^2}{\delta^2}$ qubits.

\item We highlight that a number measurement is canonically implemented. We measure qubits representing the measured mode, and the binary representation of the measurement string represents the measured number $n=0,1,\cdots$.
\end{enumerate}

We now proceed to prove each step above carefully.
\begin{enumerate}
    \item Let us recall \cite[Proposition 4.2]{chabaudBosonicQuantumComputational2025}, whose proof generalizes straightforwardly to any gate set over multiple modes:
    \begin{proposition}
\label{prop:error}
   Let $ G= G_T \ldots G_1$ be a quantum circuit of size $T$. Let\\ $E^* = \max_{i \in \{0,1, \ldots , T\}} \bra{\psi_i} N \ket{\psi_i}$, where $\ket{\psi_i} = G_i \ldots G_1 \ket {\psi_0}$. Then 
   $$
   \|G_E \ket{\psi_0} - G \ket{\psi_0}\| \leq T\cdot \sqrt{\frac{2E^*}{E}} .
   $$
   Here, $G_E := G_{T, E} \ldots G_{1, E}$. 
\end{proposition}
With the above proposition at hand, we can readily move on to the next item.

\item Recall that for the exact simulation, we are using the notation $\ket{\psi_i} = G_i \cdots G_1 \ket{\psi_0}$, and for the approximate one we use $\ket{\phi_i} = G_{i,E} \cdots G_{1,E} \ket{\psi_0}$. All we need to show here, is the success probability lower bound. Let $p$ denote the entire success probability (of all $T$ steps). We have
\begin{align}
p = \norm{\ket{\phi_{T}}}^2 \geq \left(\norm{\ket{\psi_T}} - \norm{\ket{\psi_T} - \ket{\phi_T}}\right)^2 \geq 1 - 2T\sqrt{\frac{2E^\ast}{E}}.
\end{align}
where the first inequality is a triangle inequality, and the second inequality comes from \cref{prop:error}. Note that whenever the ancilla in application of \eqref{eq:be} is measured as $\ket 1$, it is a failure, and the measurement outcome serves as a flag.

\item Note that each $U_i$ has size $2\cdot E^\ast{}^{k}$. Using unitary Gaussian elimination, one can write any $N\times N$ unitary as a product of $N^2\log^2N$ many single-qubit and $2$-qubit gates (cf.~\cite[Sections 4.5.1 and 4.5.2]{NielsenChuang}).
\end{enumerate} 
\end{proof}

\noindent As a result, we immediately get that $\CVBQPpoly \subseteq \BQPspoly$, and if \cref{asmp:oracle-for-truncated-unitary} holds for our circuit (such as for the Gaussian and cubic phase gate set), we then have $\CVBQPpoly \subseteq \BQP$.

\subsection{Decidable upper bounds for any finite energy}\label{sscn:decidableupper}
In this section, we show that under any finite (but unknown) uniform energy bounds, \CVBQP\ is decidable. More precisely, we consider ``physical computations'', in the sense that for all the states throughout the computation have uniforlmy finite energy moments: 

\begin{theorem}\label{thm:decidable}
  $\CVBQP_{\mathrm{phys}}\subseteq \R$, where $\CVBQP_{\mathrm{phys}}$ is the class of $\CVBQP$ circuits where the computation is physical, in the sense that for all $k\in\mathbb N$ it holds that $M_k = \max_{t\in[0,T]} \langle N^{k}\rangle_t < \infty$ at any time step $t$ in the computation.
\end{theorem}

\noindent We note that the result holds under the weaker assumption that there exists a bound on energy moments up to some large enough $k$, depending on the largest degree of the Hamiltonians involved in the computation. We conjecture that $\CVBQP$ remains decidable even without energy bounds and discuss our rationale in \cref{app:conjecture}.

\begin{proof}

We start by introducing a general simulation strategy based on cutting off the Hamiltonian generating the evolution:

\begin{lem}\label{lem:trunc}
  Let $H\in\CC[\a_1,\a_1^\dagger,\dots,\a_m,\a_m^\dagger]$ be an essentially self-adjoint Hamiltonian of degree $d$.
  Consider an input state $\ket{\psi}\in\calD(H)$ and define $\ket{\psi(t)}\coloneq e^{itH}\ket{\psi}$.
  Let $\Pi_N$ be the projector onto the $N$-photon subspace, $H_N \coloneq \Pi_N H \Pi_N$, and $\ket{\psi_N(t)}\coloneq e^{itH_N} \ket{\psi}$.
  If $\norm{\Pi_N^\perp\ket{\psi(s)}}\le \epsilon$ for all $s\in[0,t]$, then
  \begin{equation}
   D(t) \coloneq \norm{\ket{\psi(t)} - \ket{\psi_N(t)}} \le \epsilon\left(2+t\norm{\Pi_N^\perp H\Pi_N}\right)= O(\epsilon (1+t N^{d/2})).
   \label{xeq:D(t)}
  \end{equation}
\end{lem}
\begin{proof}
  Due to the assumption $\ket{\psi}\in\calD(H)$, we have that $\ket{\psi(t)}$ is continuously differentiable (see e.g. \cite[Lemma 1.3]{EN2000}).
  We analyze the error in the interaction picture by writing $H = H_N + (H-H_N)$, where $(H-H_N)$ is the perturbation.
  We have
  \begin{align}
    \ket{\psi_I(t)} &= e^{-itH_N}\ket{\psi(t)},\\
    \frac{d}{dt}\ket{\psi_I(t)} &= \bigl(-iH_N e^{-itH_N}\bigr)\ket{\psi(t)} + e^{-itH_N}\bigl(iH\ket{\psi(t)}\bigr) = ie^{-itH_N}(H-H_N)\ket{\psi(t)}.
  \end{align}
  The Fundamental Theorem of Calculus gives
  \begin{equation}
    \ket{\psi_I(t)} - \ket{\psi_I(0)} = \int_0^t ds\left(\frac{d}{ds}\ket{\psi_I(s)}\right) = i\int_{0}^t ds\, e^{-isH_N}(H-H_N)\ket{\psi(s)}.
  \end{equation}
  Since $e^{itH_N}(\ket{\psi_I(t)}- \ket{\psi_I(0)})=\ket{\psi(t)}-\ket{\psi_N(t)}$, we have $D(t) \le \norm{\int_0^t \dotsi ds}$.
   Rewrite
  \begin{equation}
     H-H_N = (\Pi_N+\Pi_N^\perp)H- \Pi_N H \Pi_N = \Pi_N^\perp H + \Pi_NH(I-\Pi_N) = \Pi_N^\perp H + \Pi_NH\Pi^\perp_N.
  \end{equation}
  Note that $\Pi_NH\Pi^\perp_N$ is effectively finite-dimensional with $\norm{\Pi_NH\Pi^\perp_N} = O(N^{d/2})$ for $N\gg d$.
  Thus, it remains to bound $\int_0^t ds\,\Pi_N^\perp H\ket{\psi(s)}$.
  The Schrödinger equation gives
  \begin{equation}
    \frac{d}{ds}\bigl(\Pi_N^\perp \ket{\psi(s)}\bigr) = \Pi_N^\perp\frac{d}{ds} \ket{\psi(s)} = i\Pi_N^\perp H\ket{\psi(s)},
  \end{equation}
  where differentiability holds because $\Pi_N^\perp$ is bounded and $\ket{\psi(s)}$ is differentiable.
  Substituting back into the integral, and applying $e^{-isH_N}\Pi_N^\perp = \Pi_N^\perp$,
  \begin{equation}
    i\int_{0}^t ds\, e^{-isH_N} \Pi_N^\perp H\ket{\psi(s)} = \int_0^t ds\left(\frac{d}{ds}\Pi_N^\perp \ket{\psi(s)}\right) = \Pi_N^\perp \bigl(\ket{\psi(t)} - \ket{\psi(0)}\bigr),
  \end{equation}
  which is bounded in norm by $2\epsilon$, and integrability follows from \cite[Lemma 1.3(iv), Proposition C.3]{EN2000}
  Putting everything together gives
  \begin{equation}
    D(t) \le 2\epsilon + \epsilon t\norm{\Pi_N H\Pi_N^\perp}.
  \end{equation}
\end{proof}

For a gate generated by a polynomial Hamiltonian $H$ of degree $d$, \cref{lem:trunc} allows us to simulate the unitary evolution $e^{-iTH}$ by $e^{-iTH_N}$, where $H_N=\Pi_N H\Pi_N$ is a truncated version of $H$ in Fock basis. Denoting $\ket{\psi(t)}=e^{-itH}\ket{\psi(0)}$ the state throughout the computation and $\ket{\psi_N(t)}=e^{-itH_N}\ket{\psi(0)}$ its approximation, the error depends on $\|\Pi_N^\perp\ket{\psi(t)}\|$, which can be bounded as follows:
\begin{equation}\label{xeq:QEnew}
    \begin{aligned}
    \|\Pi_N^\perp\ket{\psi(t)}\|^2 &= \sum_{\abs{\bfn}> N} \abs{\braket{\bfn}{\psi(t)}}^2 \\&= \sum_{\abs{\bfn}> N} \frac{n^k}{n^k} \abs{\braket{\bfn}{\psi(t)}}^2\\
    &\le N^{-k} \sum_{\abs{\bfn} > N} \abs{\bfn}^k\abs{\braket{\bfn}{\psi(t)}}^2\\
    &\le N^{-k} M_k.
    \end{aligned}
\end{equation}
The simulation error via \cref{lem:trunc} then scales as 
\begin{equation}
   \norm{\ket{\psi(T)} - \ket{\psi_N(T)}}\le O(TN^{(d-k)/2}).
\end{equation}
Hence, choosing $k>d$ and $N$ large enough ensures that the simulation is successful. However, how large to choose $N$ depends on the constant $M_k$, which is unknown a priori, so we need a way to check that $N$ is indeed chosen large enough.

To that end, we split $e^{-iTH_N}$ into $R$ shorter gates $e^{-itH_N}$ of time $t=T/R$ and after each, we truncate to $E$ using the projector $\Pi_E$ onto the $(\le E)$-photon subspace, for $E\le N$:
\begin{equation}
    \ket{\phi_N(E,t,R)}\coloneqq\Bigl(\Pi_{E}e^{-itH_N}\Pi_E\Bigr)^R\ket{\psi(0)}.
\end{equation}
We also define the corresponding state evolved according to the original Hamiltonian:
\begin{equation}
    \ket{\phi(E,t,R)}\coloneqq\Bigl(\Pi_{E}e^{-itH}\Pi_E\Bigr)^R\ket{\psi(0)}.
\end{equation}
Using $\ket{\phi_N(E,t,R)}$ rather than $\ket{\psi_N(T)}$, the approximation error is bounded as
\begin{equation}\label{eq:spliterror}
    \|\ket{\psi(T)}-\ket{\phi_N(E,t,R)}\|\le\|\ket{\phi(E,t,R)}-\ket{\phi_N(E,t,R)}\|+\|\ket{\psi(T)}-\ket{\phi(E,t,R)}\|.
\end{equation}
To conclude the proof, we bound both terms on the right hand side. Namely, we show hereafter that for the right choice of parameters $R$, $t$ and $N$ as a function of $E$ and $T$: 
\renewcommand{\theenumi}{\roman{enumi}}
\begin{enumerate}
    \item the first error term goes to $0$ as $E$ increases, independently of the unknown constant $M_k$, owing to the short-time nature of the evolutions;
    \item the second error term has an upper bound $\delta_E$ which rapidly goes to $0$ as $E$ increases;
    \item $\delta_E$ has a computable estimate for any given additive precision.
\end{enumerate}

The simulation algorithm is then as follows: given $\epsilon>0$, pick initial values of $E$, $N$ and $R$ such that the first error term is smaller than $\epsilon/3$, and compute an estimate of $\delta_E$ up to additive precision $\epsilon/3$. If the estimate is smaller than $\epsilon/3$, then the simulation is successful, i.e.~$\|\ket{\psi(T)}-\ket{\phi_N(E,t,R)}\|\le\epsilon/3+\epsilon/3+\epsilon/3=\epsilon$ where $\ket{\phi_N(E,t,R)}$ is a computable vector; otherwise, increase $E$ and $N$. The rapid convergence of $\delta_E$ to $0$ ensures that the algorithm terminates.

\medskip

i. We consider the first error term $\|\ket{\phi(E,t,R)}-\ket{\phi_N(E,t,R)}\|$ in \cref{eq:spliterror}.
We rely on the following technical result, which bounds the ``leakage'' of a finite superposition of number states to higher photon numbers under a polynomial Hamiltonian evolution:
\begin{lem}\label{lem:leakage}
    Let $H\in\CC[\a_1,\a_1^\dagger,\dots,\a_m,\a_m^\dagger]$ be an essentially self-adjoint Hamiltonian of degree $d$ and let $E\in\mathbb N$.
    Consider an input state $\ket{\phi}$ such that $\|\Pi_E^\perp\ket\phi\|=0$ and define $\ket{\phi(s)}\coloneq e^{-isH}\ket{\phi}$. Then for all $k\in\mathbb N$ there exists a computable constant $C$ independent of $E$ such that
    \begin{equation}
        \|\Pi_{E+kd}^\perp\ket{\phi(s)}\|\le  C_{k,d}s^{k/2}E^{kd/4}.
    \end{equation}
\end{lem}

\begin{proof}
   Let us write $Q_E(s)\coloneqq\|\Pi_E^\perp\ket{\phi(s)}\|^2$. Its derivative is given by
  \begin{equation}
    \dot{Q}_E(s) = \frac d{ds}\ev{\Pi_E^\bot}{\phi(s)} = i\ev{[\Pi_{E}^\bot,H]}{\phi(s)},
  \end{equation}
  where
  \begin{align}
    [\Pi_{E}^\bot,H] &= \Pi_E^\bot H-H\Pi_E^\bot = \Pi_E^\bot H (I-\Pi_E^\bot)-(I-\Pi_E^\bot) H \Pi_E^\bot \nonumber\\
    &= \Pi_E^\bot H \Pi_E-\Pi_E H \Pi_E^\bot.
  \end{align}
  Let us write $\ket{\phi(s)} = \sum_{\bfj}\phi_{\bfj}(s)\ket{\bfj}$.
  Then for some $C>0$ (independent of $E$) and letting $H_{\bfk\bfl}\coloneq\mel{\bfk}{H}{\bfl}$ (note that this coefficient vanishes for $||\bfk|-|\bfl||>d$ since $H$ is a polynomial Hamiltonian of degree $d$),
  \begin{equation}\label{eq:Qdot1}
    \begin{aligned}
    \abs{\dot Q_E(s)} &\le 2 \sum_{|\bfk| > E, |\bfl| \le E} \abs{H_{\bfk\bfl}}\abs{\phi_{\bfk}}\abs{\phi_{\bfl}} \\
    &\le \sum_{\abs{\bfk}\in [E+1,E+d], \abs{\bfl}\in [E-d,E]}\abs{H_{\bfk\bfl}}\left(\abs{\phi_{\bfk}}^2+\abs{\phi_{\bfl}^2}\right).
    \end{aligned}
  \end{equation}
Writing $H=\sum_{0\le|\bm\mu|+|\bm\nu|\le d}h_{\bm\mu\bm\nu}\bm a^{\bm\mu}\bm a^{\bm\nu\dag}$ we obtain
  \begin{equation}
    \begin{aligned}
    |H_{\bfk\bfl}|&\le\sum_{0\le|\bm\mu|+|\bm\nu|\le d}|h_{\bm\mu\bm\nu}|\bra{\bfk}\bm a^{\bm\mu}\bm a^{\bm\nu\dag}\ket{\bfl}\\
    &\le\sum_{\max(|\bfk|,|\bfl|)\le|\bm n|\le (|\bfk|+|\bfl|+d)/2}|h_{\bm n-\bfk,\bm n-\bfl}|\sqrt{\frac{\bm n!}{\bfk!}\frac{\bm n!}{\bfl!}},
    \end{aligned}
  \end{equation}
  so for large $E$,
  \begin{equation}\label{eq:maxHkl}
    \max_{\abs{\bfk}\in [E+1,E+d], \abs{\bfl}\in [E-d,E]}
    |H_{\bfk\bfl}|\le K_dE^{d/2},
  \end{equation}
  where $K_d$ is a constant independent from $E$ computable from the description of $H$. Hence, \cref{eq:Qdot1} yields
    \begin{align}
    \abs{\dot Q_E(s)} &\le K_dE^{d/2}\left(\sum_{\abs{\bfk}\in [E+1,E+d]}\abs{\phi_{\bfk}}^2 + \sum_{\abs{\bfl}\in [E-d+1,E]}\abs{\phi_{\bfl}}^2\right)\\
    \label{eq:Qdot2}&\le K_dE^{d/2}Q_{E-d}(s)\\
    \label{eq:Qdot3}&\le K_dE^{d/2},
    \end{align}
    where we used $|Q_{E}(s)|\le1$ in the last line. Since $Q_E(0) = 0$ by assumption, integrating \cref{eq:Qdot3} we obtain
  \begin{equation}
    Q_{E}(s)\le K_dsE^{d/2}.
  \end{equation}
  By plugging this expression into \cref{eq:Qdot2} (replacing $E$ by $E+d$) and integrating again we obtain
  \begin{equation}
    Q_{E+d}(s)\le K_d^2\frac{s^2}2E^{d/2}(E+d)^{d/2}.
  \end{equation} 
  Repeating this operation $k$ times gives
  \begin{equation}
    Q_{E+kd}(s) \le K_d^{k}\frac{s^k}{k!}\prod_{j=1}^k(E+jd)^{d/2}\le C_{k,d}^2s^kE^{kd/2},
  \end{equation}
  for $E$ large enough, where $C_{k,d}$ is a computable constant.
\end{proof}

Applying \cref{lem:leakage} for $\ket{\phi}=\ket{\phi(E,t,l)}$ gives for all $k,s$
\begin{equation}
    \|\Pi_{E+kd}^\perp e^{-isH}\ket{\phi(E,t,l)}\|\le C_{k,d}s^{k/2}E^{kd/4}.
\end{equation}
Setting $N=E+kd$, \cref{lem:trunc} implies for all $l\le R$
\begin{equation}\label{eq:boundintertrunc}
    \|(e^{-itH}-e^{-itH_N})\ket{\phi(E,t,l)}\|=Ct^{k/2}E^{kd/4}\cdot(1+tE^{d/2}),
\end{equation}
for some computable constant $C$ independent of $E$. Hence, we obtain for all $l\le R$:
\begin{equation}\label{eq:boundtrunc0}
    \begin{aligned}
    \|\ket{\phi(E,t,l)}-\ket{\phi_N(E,t,l)}\|&=\|\Bigl(\Pi_{E}e^{-itH}\Pi_E\Bigr)^l\ket{\psi(0)}-\Bigl(\Pi_{E}e^{-itH_N}\Pi_E\Bigr)^l\ket{\psi(0)}\|\\
    &\le\|e^{-itH}\Bigl(\Pi_{E}e^{-itH}\Pi_E\Bigr)^{l-1}\ket{\psi(0)}-e^{-itH_N}\Bigl(\Pi_{E}e^{-itH_N}\Pi_E\Bigr)^{l-1}\ket{\psi(0)}\|\\
    &=\|(e^{-itH}-e^{-itH_N})\ket{\phi(E,t,l-1)}\\
    &\quad\quad\quad+e^{-itH_N}(\ket{\phi(E,t,l-1)}-\ket{\phi_N(E,t,l-1)})\|\\
    &\le Ct^{k/2}E^{kd/4}\cdot(1+tE^{d/2})+\|\ket{\phi(E,t,l-1)}-\ket{\phi_N(E,t,l-1)}\|\\
    &\le\dots\\
    &\le Clt^{k/2}E^{kd/4}\cdot(1+tE^{d/2}) \\
    &\le CTt^{k/2-1}E^{kd/4}\cdot(1+tE^{d/2}),
    \end{aligned}
\end{equation}
for all $k$, where we used the fact that the norm is non-increasing under projections in the second line, \cref{eq:boundintertrunc} and the triangle inequality in the fourth line, and $T=Rt$ in the last line. Setting $k=8$ and $t=E^{-d}$ we have $N=E+8d$ and
\begin{equation}\label{eq:boundtrunc01}
\|\ket{\phi(E,t,l)}-\ket{\phi_N(E,t,l)}\|\le2CTE^{-d},
\end{equation}
for all $l\le R$. Setting $l=R$, this implies that, given a target precision $\delta>0$, we can compute $E$ such that $\|\ket{\phi(E,t,R)}-\ket{\phi_N(E,t,R)}\|\le\delta$.

\medskip

ii. We now consider the second error term $\|\ket{\psi(T)}-\ket{\phi(E,t,R)}\|$ in \cref{eq:spliterror}. 
On the one hand, we have for all $l\le R$:
\begin{equation}\label{eq:boundtrunc1}
    \begin{aligned}
    \|\ket{\psi(lt)}&-\ket{\phi(E,t,l)}\| \\
    &= \|\Bigl(e^{-iltH}-\Pi_{E}\Bigl(e^{-itH}\Pi_E\Bigr)^l\Bigr)\ket{\psi(0)}\|\\
    &=\|(I-\Pi_E)\Bigl(e^{-itH}\Pi_{E}\Bigr)^{l}\ket{\psi(0)}+e^{-itH}\Bigl(e^{-i(l-1)tH}-\Pi_{E}\Bigl(e^{-itH}\Pi_{E}\Bigr)^{l-1}\Bigr)\ket{\psi(0)}\|\\
    &\le\|\Pi_E^\perp\ket{\psi(lt)}\|+\|\Bigl(e^{-i(l-1)tH}-\Pi_{E}\Bigl(e^{-itH}\Pi_{E}\Bigr)^{l-1}\Bigr)\ket{\psi(0)}\|\\
    &\le\dots\\
    &\le\sum_{j=0}^{l}\|\Pi_E^\perp\ket{\psi(jt)}\|\\
    &\le(l+1)\sqrt{M_k}E^{-k/2},
    \end{aligned}
\end{equation}
where we used the triangle inequality and the fact that the norm is non-increasing under projections and unitary evolutions in the third line, and \cref{xeq:QEnew} in the last line. Taking $l=R$ in \cref{eq:boundtrunc1} implies that the approximation error $\|\ket{\psi(T)}-\ket{\phi(E,t,R)}\|$ rapidly decreases to $0$ as the cut-off parameter $E$ increases. 

On the other hand, we also have:
\begin{equation}\label{eq:boundtrunc2}
    \begin{aligned}
    \|\ket{\psi(T)}&-\ket{\phi(E,t,R)}\|\\
    &= \|\Bigl(e^{-iRtH}-\Bigl(\Pi_{E}e^{-itH}\Bigr)^R\Pi_E\Bigr)\ket{\psi(0)}\|\\
    &= \|(I-\Pi_E)e^{-iRtH}\ket{\psi(0)}+\Pi_Ee^{-itH}\Bigl(e^{-i(R-1)tH}-\Bigl(\Pi_{E}e^{-itH}\Bigr)^{R-1}\Pi_E\Bigr)\ket{\psi(0)}\|\\
    &\le\|\Pi_E^\perp e^{-itH}\ket{\phi(E,t,R-1)}\|+\|\Bigl(e^{-i(R-1)tH}-\Bigl(\Pi_{E}e^{-itH}\Bigr)^{R-1}\Pi_E\Bigr)\ket{\psi(0)}\|\\
    &\le\dots\\
    &\le\sum_{l=1}^{R}\|\Pi_E^\perp e^{-itH}\ket{\phi(E,t,l-1)}\|+\|\Pi_E^\perp\ket{\psi(0)}\|\\
    &=\sum_{l=1}^{R-1}\sqrt{\|\ket{\phi(E,t,l-1)}\|^2-\|\ket{\phi(E,t,l)}\|^2}+\sqrt{1-\|\Pi_E\ket{\psi(0)}\|^2}=:\delta_E,
    \end{aligned}
\end{equation}
where we used the triangle inequality and the fact that the norm is non-increasing under projections and unitary evolutions in the third line, as well as the Pythagorean identity in the last line.

\cref{eq:boundtrunc2} provides a way to track the approximation error by tracking instead the values of the norms $\|\Pi_E\ket{\psi(0)}\|$ and $\|\ket{\phi(E,t,l)}\|$, $l\le R$. By \cref{eq:boundtrunc1} these norms all converge rapidly to $1$ since $|\|a\|-\|b\||\le\|a-b\|$ and $\ket{\psi(lt)}$ is normalized for all $l\le R$, so $\delta_E$ rapidly goes to $0$ when $E$ increases.

\medskip

iii. While we can compute $\|\Pi_E\ket{\psi(0)}\|$, it is not a priori clear that the norms $\|\ket{\phi(E,t,l)}\|$ (and thus $\delta_E$) are computable.
Instead, we show that these are well-approximated by the norms $\|\ket{\phi_N(E,t,l)}\|$:
\begin{equation}\label{eq:boundtrunc3}
    \begin{aligned}
    \left|\|\ket{\phi(E,t,l)}\|-\|\ket{\phi_N(E,t,l)}\|\right|&\le\|\ket{\phi(E,t,l)}-\ket{\phi_N(E,t,l)}\|\\
    &\le 2CTE^{-d},
    \end{aligned}
\end{equation}
where we used \cref{eq:boundtrunc01} in the last line and where $C$ is a computable constant. Since the norms $\|\ket{\phi_N(E,t,l)}\|$ are computable, this implies that, given $\delta>0$, we can compute a estimate of $\delta_E$ up to additive precision $\delta$, thus concluding the proof.

\medskip

The simulation strategy for sequences of gates generated by polynomial Hamiltonians $H^{(1)},\dots,H^{(L)}$ is analogous, using a state of the form:
\begin{equation}
    \Bigl(\Pi_{E_L}e^{-itH_{N_L}^{(L)}}\Pi_{E_L}\Bigr)^{R_L}\dots\Bigl(\Pi_{E_2}e^{-itH_{N_2}^{(2)}}\Pi_{E_2}\Bigr)^{R_2}\Bigl(\Pi_{E_1}e^{-itH_{N_1}^{(1)}}\Pi_{E_1}\Bigr)^{R_1}\ket{\psi(0)}.
\end{equation}
\end{proof}

\subsection{Gaussian and cubic gate set without energy restrictions}\label{sscn:simulation}

In this section, we assume a cubic and Gaussian gate set. Note that our results directly generalize to other gates which may be recompiled exactly into this gate set, such as the quartic gate \cite{kalajdzievski2021exact}.

\subsubsection{Teleporting cubic gate with finite squeezed state}
Following \cite{Gottesman2001}, one can use gate teleportation to replace cubic gates at the middle of the computation with ``magic'' input states, at the beginning of the computation, followed by gadgets made of Gaussian gates and adaptive intermediate measurements. However, the exact construction of this gate teleportation can only be done using magic states that are not normalizable; using normalized magic states, one can only come up with approximate constructions. To see this, suppose we want to apply the cubic gate $V^{(3)}_1 (\theta)$ to the quantum state $\ket{\phi}$ and obtain $\ket {\phi (\theta)}$. Let $\ket{S_\xi} = \int_{-\infty}^{\infty} \frac{e^{- \frac{x^2}{2\xi}}}{\pi^{1/4} \xi} \ket {x} dx$ be the squeezed vacuum state. Define the $\xi$-squeezed cubic state as 
\begin{align}
    \begin{split}
    \ket{V^{(3)} (\theta); \xi} &:= V^{(3)} (\theta) \ket {S_{\xi}}\\
    &= \int_{x \in \mathbb{R}} \frac{e^{ i \frac{\theta x^3}{3} - \frac{x^2}{2\xi}}}{\pi^{1/4} \sqrt{\xi}} \ket{x} dx.
    \end{split}
\end{align}
be the application of a cubic gate on squeezed vacuum. Following \cite{GKP2001} (see also \cite{ghose2007non}) we implement the gadget introduced in \cref{fig:cubic-gadget}.
\begin{figure}
    \centering
    \begin{quantikz}[row sep=1.2em, column sep=1.5em]
\lstick{$\ket{\phi}$} & \gate[wires=2]{\text{SUM}^{-1}} & \gate{G(q)} & \qw \\
\lstick{$\ket{V^{(3)}(\theta);\xi}$} & & \cwbend{-1} & \qw
\end{quantikz}    

\caption{The cubic teleportation gadget introduced in \cite{GKP2001, ghose2007non}. The state $\ket{V^{(3)}(\theta);\xi} = V^{(3)}(\theta)\ket{S_\xi}$ is the cubic phase gate applied to a finitely-squeezed state. We highlight that the measurement is in the $X$ quadrature basis (giving output $q\in\mathbb R$), and $G$ and $\mathrm{SUM}$ are Gaussian gates. In \cref{lem:cubic-teleportation} we prove that this gadget works with high success probability $1-\delta$ and accuracy $\varepsilon$ with finite squeezing parameter $\xi\leq \mathsf{poly}(E,\varepsilon^{-1}) \, \mathsf{qpoly}(\delta^{-1})$, with $E$ being the energy of the input state. We can also flag if the implementation if faulty (occuring with probability $\delta$).}
    \label{fig:cubic-gadget}
\end{figure}
Here $\hat{\text{SUM}} = e^{i \hX_1 \hP_2}$ which maps the position basis states $\ket{x} \ket{y} \mapsto \ket{x} \ket{x + y}$ and
$$
\hat G(q) = e^{- i \theta q^3/3 - i \theta q \hat X ^2 - i \theta q^2 \hat X}.
$$
We first recall how this gadget works when instead of $\ket{V^{(3)} (\theta) ; \xi}$ we use the unnormalizable state $\ket{V^{(3)}(\theta)} := \int_{x \in \mathbb{R}} e^{i \theta x^3} \ket{x} dx$. To demonstrate why this gadget works, we show that it operates correctly for any position basis state $\ket{y}$ on the top mode. This is because we can expand $\ket{\phi} = \int \phi_y \ket{y} dy$. After the application of $\hat{\text{SUM}}^{-1}$, we obtain
\begin{align}\label{eq:approximate-magic}
    \begin{split}
\hat{\text{SUM}}^{-1} \left (\ket{y} \otimes \int_{x \in \mathbb{R}} e^{i \theta x^3/3} \ket{x} dx \right) &=  \int_{x \in \mathbb{R}} e^{i \frac{\theta x^3}{3}} \ket{y}\ket{x - y} dx\\
        &= \int_{\alpha \in \mathbb{R}}  e^{ i \frac{\theta (\alpha+y)^3}{3}} \ket{y}\ket{\alpha} d\alpha\\
        &= \int_{\alpha \in \mathbb{R}}  e^{i\theta y^3/3 +Q_T(y, \alpha)} \ket{y}\ket{\alpha} d\alpha,
    \end{split}
\end{align}
where 
$$
\hQ_\xi(y,\alpha) = i \theta \alpha^3/3 + i \theta \alpha y^2 + i \theta \alpha^2 y.
$$
Measuring the bottom register and obtaining $q$, we can correct the top mode using the following Gaussian gate:
\begin{align}
\hat G(q) = e^{- i \theta q^3/3 - i \theta q \hat X ^2 - i \theta q^2 \hat X}.
\end{align}
In summary
\asp{
(\hG(q) \otimes \bra{q}) (\SUM^{-1}) (\ket{\phi} \otimes \ket{V^{(3)}(\theta)}) = (\hV^{(3)} (\theta) \ket{\phi}) \otimes \ket{q}.
}

In what follows, we analyze the effect of doing this gate teleportation via a finite-energy squeezed state.

\begin{lem}[Cubic phase gate teleportation via finite-energy ancilla]\label{lem:cubic-teleportation}
Let $\ket{\phi}$ be a state with expected photon number at most $E$. Then, for any $\delta > 0$ there exists an constant $c > 0$ such that if
\begin{align}
\xi > c \,\frac{E}{\varepsilon}e^{\frac12\log^2\frac2\delta}
\log^2(1/\delta),
\end{align}
then the magic injection gadget in \cref{fig:cubic-gadget} which uses $\ket{V^{(3)} (\theta) ; \xi}$ as magic input state succeeds
with probability at least $1-\delta$ and outputs a quantum state that is within $\varepsilon$ Euclidean distance of $\ket{\phi(\theta)}:=e^{i\theta \hX^3/3} \ket{\phi}$.
\end{lem}

\begin{proof}
Our proof consists of two steps:
\begin{enumerate}
    \item We first show that the homodyne measurement $q$ on the bottom mode of the gadget in \cref{fig:cubic-gadget} satisfies $|q|\leq 2\xi \log\frac2\delta$ with probability at least $1-\delta$, if $\xi\geq E^{\frac12}\left(\frac{\delta^{-\frac12}}{\log(2\delta^{-1})} \right) $.
    \item Next, we show that that for any $|q|\leq 2\xi \log\frac2\delta$, the resulting output state is $8\varepsilon$-close to the desired target state $\hV^{(3)} (\theta)\ket{\phi}$ within $\ell_2$ distance, if $\xi \geq c \,\frac{E}{\varepsilon}e^{\frac12\log^2\frac2\delta}
\log^2(1/\delta)$.
\end{enumerate}

In what follows, we provide details of the above steps:
\begin{enumerate}
\item Recall that $q$ is the result of the homodyne measurement on the bottom mode in \cref{fig:cubic-gadget}. In this step, we only study the marginal probability distribution on the second mode; hence, without loss of generality, we can assume we measure both modes in the end in the position basis (throwing away the top mode's measurement result gives us a fair sample for $q$). Note that measurement over both modes commutes with the gate (in general measuring all registers in the computational basis can be commuted through a gate that maps computational basis to computational basis $U: \ket{\mathbf x}\mapsto\ket{f(\mathbf x)}$). Therefore, letting $(q', \tilde q)$ represent the samples we get by (hypothetically and) directly measuring the input state and the ancilla, we have that $q = \tilde q-q'$. We now show that $|q|\leq 2\xi \sqrt{2\log\frac2\delta}$ with probability $1-\delta$. Note that $\Pr\left[|\tilde{q}'|\geq \sqrt{\frac{E+1}{\delta}}\right] \leq \delta$.This is due to a simple Markov inequality since $\langle X^2\rangle \leq E + 1$. In particular $\Pr[|\tilde{q}'| \geq d] = \Pr[\tilde{q}'^2 \geq d^2]\leq \frac{\langle X^2\rangle}{d^2} \leq \frac{E +1}{d^2}$. Furthermore, measuring a squeezed vacuum on the top mode does not reveal $|q| \geq d$ from the origin  with probability at least $1-2e^{- 
\frac{d^2}{\xi^2}}$. Hence, $\Pr[|\tilde{q}| \geq \xi \sqrt{2\log\frac2\delta}]\leq \delta$. So, the combined outcome $q = \tilde{q}-\tilde{q}'$ satisfies:
\begin{align}
\Pr\left[|q|\geq d^\ast  \right] \leq 1- \Pr\left[ |\tilde{q}|\leq d^\ast/2,\, \text{and}\; |\tilde{q}'|\leq d^\ast/2\right]\leq 1-(1-\delta)^2\leq 2\delta,
\end{align}
with $d^\ast = 2\max\left( \sqrt{\frac{E+1}{\delta}}, \xi \sqrt{2\log\frac2\delta} \right) = 2\xi \sqrt{2\log\frac2\delta}$ assuming  $\xi\geq (E+1)^{1/2}\left(\frac{\delta^{-1/2}}{\log(2\delta^{-1})} \right) $.

\item Now, assume having measured $q$. Following a similar line of reasoning as outlined in \cref{eq:approximate-magic}, the wave function of the first mode has collapsed onto the state proportional to the following unnormalized state:
\begin{align}
\ket{\phi_\xi(\theta)} = \int_{x \in \bbR} \phi(x)\exp(-\frac{x^2}{2\xi^2}-\frac{xq}{\xi^2}) e^{i\theta x^3/3} \ket{x} dx.
\end{align}

We have
\begin{align}\label{eq:finite-sqz-teleport}
\begin{split}
\norm{\ket{\phi_\xi(\theta)} - \ket{\phi(\theta)} }^2 &= \int_x |\phi(x)|^2 \left| \exp(-\frac{x^2}{2\xi^2}-\frac{xq}{\xi^2}) - 1\right|^2 \, \mathrm dx.
\end{split}
\end{align}
As $|\frac{q}{\xi}|\leq 2\log\frac2\delta$ with probability at least $1-\delta$, we rewrite the above distance as
\begin{align}\label{eq:finite-sqz-teleport-approx}
\begin{split}
\norm{\ket{\phi_\xi(\theta)} - \ket{\phi(\theta)} }^2 &= \int_x |\phi(x)|^2 \left| \exp(-\frac{x^2}{2\xi^2}-\frac{xb}{\xi}) - 1\right|^2 \, \mathrm dx,
\end{split}
\end{align}
where $|b|\leq 2 \log\frac2\delta$ with probability $1-\delta$.
We then use the following lemma.

\begin{lem}\label{lem:random-variable-bound}
Let $X$ be a random variable such that $\mathbb E[X^2]\leq \sigma^2$, and $a \geq 0$. Then, we have that
\begin{align}
\mathbb E\left[ \left|e^{-aX^2+b\sqrt{a}X} - 1\right|^2\right]\leq \varepsilon,
\end{align}
if
\begin{align}
a\leq c\min\left( 1, b^{-4}, \varepsilon^{\frac23}\sigma^{-\frac43},\varepsilon b^{-2}\sigma^{-2},(e^{b^2}+1)^{-2}\sigma^{-4}\varepsilon^2 \right),
\end{align}
for some constant $c>0$.
\end{lem}
\begin{proof}
First, we argue that there exists a small enough $a$ such that $\mathbb E\left[ \left|e^{-aX^2+b\sqrt{a}X} - 1\right|^2\right]$ is arbitrarily close to zero. To see this, let
\begin{align}
Y_a := -aX^2 + b\sqrt{a}X, \quad f(a):= \mathbb E\left[ \left|e^{-aX^2+b\sqrt{a}X} - 1\right|^2\right],
\end{align}
and notice that
\begin{align}
Y_a = -(\sqrt a X - \frac{b}2)^2 + \frac{b^2}{4}\leq \frac{b^2}{4}.
\end{align}
As a result, 
\begin{align}
e^{2Y_a} \leq e^{b^2/2},
\end{align}
and hence, for all $a$ we have
\begin{align}\label{eq:y_a-trivial-bound}
\left(e^{Y_a}-1  \right)^2\leq 2e^{2Y_a} + 2\leq 2(e^{b^2/2}+1),
\end{align}
which is a finite number for any $b$. The second inequality is by Cauchy-Schwarz.  Therefore, by the Dominated Convergence Theorem \cite{evans2018measure}, we have
\begin{align}
\lim_{a\to0^+} \mathbb E\left(e^{Y_a}-1  \right)^2 = f(0) = 0.
\end{align}
Now that we are certain that there exists a small enough $a$ satisfying the statement we want, we go about finding a bound on $a$. The rest of the proof is to find this bound.

Pick a cutoff $M$ (to be determined later) and write
\begin{align}\label{eq:f(a)}
f(a) = \mathbb E\left[|e^{Y_a}-1|^2 \, \mathbf 1_{\{|X|\leq M\}}\right] + \mathbb E\left[|e^{Y_a}-1|^2 \, \mathbf 1_{\{|X| > M\}}\right].
\end{align}
We now bound each term separately:
\begin{enumerate}
    \item Bounding the first term of \eqref{eq:f(a)}: Whenever $|X|\leq M$ we can write
    $Y_a \leq a M^2 + |b|\sqrt{a}M$, and hence, using the elementary bound $|e^z-1|\leq |z|e^{|z|}$ we get
\begin{align}
\mathbb E\left[|e^{Y_a}-1|^2 \, \mathbf 1_{\{|X|\leq M\}}\right]\leq \exp(2aM+2b\sqrt aM) \mathbb E[Y_a^2\mathbf 1_{|X|\leq M}].
\end{align}
Also, note that $Y_a^2 \leq 2a^2X^4+2b^2aX^2$, and that $\mathbb E[X^4 \mathbf 1_{|X|\leq M}]\leq M^2 \mathbb E[X^2 \mathbf 1_{|X|\leq M}]\leq M^2\sigma^2$, which allows us to write
\begin{align}
\mathbb E[Y_a^2\mathbf 1_{\{X\leq M \}}] \leq 2a^2 M^2 \sigma^2 + 2b^2 a \sigma^2.
\end{align}
Therefore, we get
\begin{align}\label{eq:smaller-ineq}
\mathbb E\left[|e^{Y_a}-1|^2 \, \mathbf 1_{\{|X|\leq M\}}\right]\leq \exp(2aM+2b\sqrt aM) (2a^2 M^2 \sigma^2 + 2b^2 a \sigma^2).
\end{align}
\item Bounding the second term of \eqref{eq:f(a)}: We use \eqref{eq:y_a-trivial-bound} to write
\begin{align}\label{eq:larger-ineq}
\mathbb E[|e^{Y_a}-1|^2\mathbf 1_{\{X> M \}}]\leq 2(e^{b^2/2}+1) \mathbb E[\mathbf{1}_{\{|X|>M\}}] = 2(e^{b^2/2}+1)\Pr[|X|>M]\leq 2(e^{b^2/2}+1) \frac{\sigma^2}{M^2},
\end{align}
where the last inequality is a Markov inequality.
\end{enumerate}
Now, let us choose $M = a^{-\frac14}$. Eq \eqref{eq:smaller-ineq} becomes
\begin{align}
\mathbb E\left[|e^{Y_a}-1|^2 
\, \mathbf 1_{\{|X|\leq M\}}\right]
\leq \exp(2a^{3/4} + 2b a^{1/4}) (2\sigma^2 a^{\frac32} + 2b^2 a \sigma^2), 
\end{align}
Note that all terms have $a$, and so by choosing it small enough, we can make the expression arbitrarily small. We choose $a\leq \min(1, b^{-4}, \varepsilon^{\frac23}\sigma^{-\frac43},\varepsilon b^{-2}\sigma^{-2})$ to get
\begin{align}
\mathbb E[|e^{Y_a}-1|^2 \mathbf 1_{\{|X|\leq M\}}] \leq c_1\varepsilon,
\end{align}
for a constant $c_1>0$.

Lastly, setting $M=a^{-\frac14}$ in \eqref{eq:larger-ineq} gives
\begin{align}
\mathbb E[|e^{Y_a}-1|^2 \mathbf 1_{\{|X|> M\}}] \leq 2 (e^{b^2/2}+1) a^{\frac12} \sigma^2,
\end{align}
and hence, choosing $a\leq (e^{b^2}+1)^{-2}\sigma^{-4}\varepsilon^2$ gives
\begin{align}
\mathbb E[|e^{Y_a}-1|^2 \mathbf 1_{\{|X|> M\}}] \leq c_2\varepsilon,
\end{align}
for some constant $c_2\geq 0$. Therefore, we get an overall error of $(c_1+c_2)\varepsilon$.
\end{proof}

We can now directly use \cref{lem:random-variable-bound} to upper bound the right hand side of \eqref{eq:finite-sqz-teleport-approx} and obtain
\begin{align}\label{eq:unnorm-bound}
\begin{split}
\norm{\ket{\phi_\xi(\theta)} - \ket{\phi(\theta)} }^2 &\leq  \varepsilon^2,
\end{split}
\end{align}
for
\begin{align}
\xi \geq c \,\frac{E}{\varepsilon}e^{\frac12\log^2\frac2\delta} 
\log^2(1/\delta).
\end{align}
Finally, note that we have computed the distance of the target state to our unnormalized state here. However, the distance to the normalized state would be different by a mere factor of $2$, by a straightforward application of standard techniques such as \cite[Lemma 13]{berry2017quantum}. For completeness, we re-state it here:
\begin{lem}
Given a normalized state $\psi$ and a potentially unnormalized state $\psi'$,\\
\[
    \norm{\ket\psi-\ket{\psi'}}\leq \varepsilon \implies \norm{\ket\psi - \frac{\ket{\psi'}}{\norm{\ket{\psi'}}}}\leq 2\varepsilon.
\]
\end{lem}
\begin{proof}
\begin{align}
\begin{split}
\norm{\ket\psi - \tfrac{\ket{\psi'}}{\norm{\ket{\psi'}}}} 
&= \norm{\ket\psi - \ket{\psi'} + \ket{\psi'} - \tfrac{\ket{\psi'}}{\norm{\ket{\psi'}}}} \\
&\leq \norm{\ket\psi - \ket{\psi'}} + \norm{\ket{\psi'} - \tfrac{\ket{\psi'}}{\norm{\ket{\psi'}}}} \\
&\leq \varepsilon + \norm{\ket{\psi'} - \tfrac{\ket{\psi'}}{\norm{\ket{\psi'}}}} \\
&= \varepsilon + \frac{\abs{1-\norm{\ket{\psi'}}}}{\norm{\ket{\psi'}}}\norm{\ket{\psi'}} \\
&\leq 2\varepsilon,
\end{split}
\end{align}
where in the last inequality we have used $\abs{1-\norm{\ket{\psi'}}} = \abs{\norm{\psi}-\norm{\ket{\psi'}}} \leq \varepsilon$.
\end{proof}

\end{enumerate}
\end{proof}

\subsubsubsection{Formulation for observable expectation values}
Next, we use the magic injection gadget in \cref{fig:cubic-gadget} to formulate observable expectation values at the output of CV quantum circuits composed of cubic and Gaussian gates. Let $U$ be a polynomial-size CV circuit from Gaussian and cubic gates. Let $\ket{\phi_s}$ be the state of the system after applying $s$ gates with $m$ cubic gates to $n$ modes, starting from the vacuum state. Assume furthermore $\xi$ is a squeezing parameter which we will specify later. We push the cubic gates to the beginning and start with the quantum state
\begin{align}
    \begin{split}
    \ket{\Psi_\xi (\theta_1, \ldots, \theta_m)} &:= \ket{0^n} \ket{V^{(3)} (\theta_1);\xi} \otimes \ldots \otimes \ket{V^{(3)} (\theta_m);\xi}\\
    &= e^{i (\theta_{n+1} \hat X_1^3/3 + \ldots + \theta_{m} \hat X_{n+m}^3/3)}\ket{\Psi_\xi(0)}\\
    &=: \hat \chi (\theta_1, \ldots, \theta_m)\ket{\Psi_\xi(0)},
    \end{split}
\end{align}
where 
$\ket{\Psi_\xi(0)} = \ket{0^n} \ket{S_\xi}^{\otimes m}$.

Let $O$ be a self-adjoint observable, and we are interested in formulating $\alpha_s := \bra{\psi_s} O \ket{\psi_s}$. The degree of $O$ may be finite (e.g.~position, particle number, and momentum) or infinite (e.g.~projectors onto Fock basis). Note in the gate teleportation procedure, upon measuring $q = 0$, we do not need to apply any Gaussian corrections. Suppose the $j$'th cubic gate is applied to mode $A_j$. We can therefore replace
$$
\hV_{A_j}^{(3)} (\theta_m) \ket{\phi_j} = \frac{1}{\sqrt{Z_j}}(\hI_{A_j} \otimes \bra{X = 0}_{A'_j}) \hat{\mathrm{SUM}}_{A_jA'_j}^{-1} (\hI_{A_j} \otimes \ket{V^{(3)} (\theta_m);\xi}_{A'_j}) \ket{\phi_j},
$$
where $A'_j$ is the auxiliary system used to store the cubic magic state corresponding to the $j$'th gadget, and
\asp{Z_j := \int_{y \in \mathbb{R}}
|\phi_j (y)|^2 \frac{e^{-\frac{y^2}{\xi^2}}}{\sqrt{\pi}\xi} dy.
\label{eq:Z}
}
Therefore
\begin{align}
    \begin{split}
        \ket{\psi_s} &= \hV^{(3)} (\theta_m) \hG_m \ldots \hV^{(3)} (\theta_1) \hG_1 \ket{0^n}\\
        &= \frac{1}{\sqrt{Z_1 \ldots Z_m}} I \otimes \bra {x_{n+1} = \ldots = x_{n + m} = 0}\hG \ket{\Psi_T(\theta_1, \ldots, \theta_m)},
    \end{split}
\end{align}
where 
$$
\hG = \prod_{t = 1}^m\hat{\mathrm{SUM}}_{i_t, n + t}^{-1} \hG_t
$$
is the effective Gaussian gate we will have to apply. For simplicity of notation let $\ket{x_{n +[m]} = 0} := \ket{x_{n+1} = \ldots = x_{n + m} = 0}$ and $\Lambda = \ket{x_{n +[m]} = 0}\bra{x_{n +[m]} = 0}$.
We therefore have
\asp{
    \alpha_s &= \bra{\psi_s} O \ket{\psi_s}\\
    &= \frac{1}{\sqrt{Z_1  \ldots Z_m}} \bra{\Psi_\xi(\theta_1, \ldots, \theta_m)}
     \hG^\dagger (\hO \otimes \hat \Lambda) \hG\ket{\Psi_\xi (\theta_1,\ldots, \theta_m)}.
     \label{eq:O-expectation}
     }

 Finally, we argue that if $\xi$ is large enough then $Z_j$ is very close to $1$:

 \begin{lem}\label{lem:post-select-on-zero}
     Suppose the quantum state $\ket{\phi}$ has an expected particle number at most $E$, and $Z$ is defined as in \cref{eq:Z}. Then if $\xi \geq \frac{\sqrt{2E+1}}{\lambda}$ then
     $(1 - \lambda)^2 \frac{1}{\sqrt{\pi}\xi} \leq Z \leq \frac{1}{\sqrt{\pi}\xi}$. Furthermore, this statement can be made robust: if there exists a quantum state $\ket{\psi}$ such that $\|\ket{\psi} - \ket{\phi}\| \leq \delta_1$ and $\bra{\psi} \hN \ket{\psi} \leq E$, then $(1 - 2\lambda - \delta_1) \frac{1}{\sqrt{\pi}\xi} \leq Z \leq \frac{1}{\sqrt{\pi}\xi}$ assuming $\delta_1 < 1 - 2 \lambda$.
     \label{lem:Z-bound}
 \end{lem}

 \begin{proof}
 For the upper bound, observe that 
 \asp{Z\leq \frac{1}{\sqrt{\pi}\xi} \int |\phi(y)|^2 dy \leq \frac{1}{\sqrt{\pi}\xi}.}
 Now set $d^* = \sqrt{\frac{2E+1}\lambda}$ and $\xi \geq \frac{d^*}{ \sqrt{\ln (1/(1-\lambda))}}$. Then
\asp{
    Z \geq \int_{|y| \leq d^*}
|\phi (y)|^2 \frac{e^{-\frac{y^2}{\xi^2}}}{\sqrt{\pi}\xi} dy \geq \frac{1}{\sqrt{\pi} \xi}(1-\lambda) \int_{|y| \leq d^*}
|\phi (y)|^2 dy \geq \frac{1}{\sqrt{\pi} \xi}(1-\lambda)^2.
}

We now prove the robust version. The upper bound is immediate (due to normalization). For the lower bound, we note
\asp{
Z &\geq  |\int_{y}
|\psi(y)|^2 \frac{e^{-\frac{y^2}{\xi^2}}}{\sqrt{\pi}\xi} dy - \int_{y}
|\phi (y) - \psi(y)|^2 \frac{e^{-\frac{y^2}{\xi^2}}}{\sqrt{\pi}\xi} dy|\\
&\geq \frac{1}{\sqrt{\pi} \xi} (1-\lambda)^2 - \frac{1}{\sqrt{\pi} \xi} \delta_1^2\\
&\geq \frac{1}{\sqrt{\pi} \xi} (1-2\lambda- \delta_1^2).\\
}
 \end{proof}

Let $\kappa = \frac{1}{\sqrt{Z_1 \ldots Z_m}}$. If $\xi \geq C \cdot \frac{m \sqrt{E}}{\lambda^2}$, then 
$$
\pi^{m/4} \xi^{m/2} \leq \kappa \leq \pi^{m/4} \xi^{m/2} (1 + \lambda).
$$

\subsubsection{Complexity of the expectation value estimation problem without adaptive measurements}
As a warm-up example, we first study the complexity of expectation value computation when adaptive measurements (or feed forwards) are not allowed. It turns out that when $\deg(O) = n^{O(1)}$, this problem can be solved in polynomial time. In particular,
\begin{theorem}
    Let $\hG$ be a Gaussian circuit with polynomial size $s$ over $n + m$ modes for $m = n^{O(1)}$. Let $\hO$ be an observable of polynomial degree. Then $\alpha = \bra{\Psi_\xi (\theta_1, \ldots, \theta_m)}\hO\ket{\Psi_\xi (\theta_1, \ldots, \theta_m)}$ can be computed in $\mathrm{poly} (n, \log \xi)$ time. 
\end{theorem}

\begin{proof}
   From \cref{eq:O-expectation} we have an expression for $\alpha$:
   \begin{align}
    \alpha = \bra {0^n}\bra{S_\xi}^{\otimes m} O (\boldsymbol{\theta},G) \ket{0^n} \ket{S_\xi}^{\otimes m}.   
   \end{align}
    where $\hO(\boldsymbol{\theta}, G) = \hat \chi (\boldsymbol{\theta})\hG \hO \hG^{-1} \hat \chi^{-1} (\boldsymbol{\theta})$ 
    Since we do not have intermediate measurements, we have excluded the $\Lambda$ term. 

    Let $\mathrm{deg} (O)$ be the degree of $O$ in the position and momentum basis, i.e.~$O = \sum_{\mu, \nu \in \mathbb{Z}^n} o_{\mu, \nu} X^{\mu} P^{\nu}$ such that $|\mu| + |\nu| \leq \deg(O)$ everywhere in the sum. Using an observation made in \cite{chabaudBosonicQuantumComputational2025} we can show that

\begin{lem}
    $\deg(O (\boldsymbol{\theta},G)) \leq 2 \mathrm{deg} (O)$. Furthermore, the coefficients of $\deg(O (\boldsymbol{\theta},G))$ can be computed in polynomial time in the description of $O, G, \theta_1, \ldots, \theta_m$.
    \label{lem:deg-O}
\end{lem} 
\begin{proof}
    The Gaussian $G$ does not change the degree of the observable, because under conjugation, the observable evolution with Gaussians only produces an affine transformation
\begin{align}
G:
\begin{pmatrix}
\vec{X}\\
\vec{P}
\end{pmatrix} \mapsto
S
\begin{pmatrix}
\vec{X}\\
\vec{P}
\end{pmatrix}
+ \vec d.
\end{align}
Next, we show that under the conjugate action of $\chi^\dagger (\boldsymbol{\theta})$, the degree of the observable at most gets doubled. To see this, we study the influence on a monomial $X^\mu P^\nu$ for $\mu, \nu \in \mathbb{Z}^n$. Since $X$ commutes with $\chi (\boldsymbol {\theta})$ we only need to study the impact on $P^\nu = P_1^{\nu_1} \ldots P_{n+m}^{\nu_{n+m}}$. Fortunately, we have an explicit formulation 
\begin{align*}
    \chi^\dagger (\boldsymbol{\theta}) P^\nu \chi (\boldsymbol{\theta}) &= P_1^{\nu_1} \ldots P_n^{\nu_n}\prod_{j=n+1}^{n+m} (V_j^{(3)}(\theta_1)P_j V_j^{(3) \dagger}(\theta_1))^{\nu_j}\\
    &= P_1^{\nu_1} \ldots P_n^{\nu_n}\prod_{j=n+1}^{n+m}  (\hat P_j + \theta_j \hat X_j^2),
\end{align*}
which has a degree at most twice the monomial. 
\end{proof}

   Using \cref{lem:deg-O} above, $\hO (\boldsymbol{\theta},G) = \sum_{|\mu| + |\nu| \leq 2 d_O} \tilde O_{\mu, \nu} X^{\mu} P^{\nu}$ can be computed in polynomial time and has at most $n^{O(d_O)}$ terms. It is enough to show that for each $\mu, \nu \in \mathbb{Z}^n$ the expectation $\bra {0^n}\bra{S_\xi}^{\otimes m} X^\mu P^\nu \ket{0^n} \ket{S_\xi}^{\otimes m}$ can be computed efficiently. However, all we need to do is compute a product like
   $$
   \prod_{j = 1}^{|\mu| + |\nu|} b_j,
   $$
   where each term $b_j$ is either one of  $\bra{0}X^{t_j}\ket{0} = f(t_j), \bra{0}P^{t_j}\ket{0} = f(t_j), \bra{S_\xi}X^{t_j}\ket{S_\xi} = f(t_j) \xi^{t_j},$ or $\bra{S_\xi}P^{t_j}\ket{S_\xi} = f(t_j) \xi^{- t_j}$ for some integer $t_j$. Here for any integer $t$, $f(2t+1) = 0$ and
   $$
   f (2t) = \frac{(2t)!}{4^t t!}.
   $$
All of these operations can be done in polynomial time (in $n$ and $\log \xi$). 
\end{proof}

\subsection{Connecting the general problem to the Gaussian rank of cubic states }

Next, we include adaptive measurements. We need to give a proper approximation to $\Lambda$ as it involves projectors onto the position basis. We approximate the position $\ket{x}$ with squeezed vacuum $\ket{S_{1/\beta} (x)} = e^{i \hat P x} \hS_{1/\beta} \ket{0}$, where $\beta$ is a large squeezing parameter to be specified later. We now express the observable expectation for an observable $\hO$ (with bounded or unbounded degree). We use label $A$ to denote modes $1$ to $n$, and $B$ to denote ancilla modes $n+1$ to $n + m$; $A_i$ and $B_i$ refer to the $i$'th mode in the system $A$ and $B$, respectively. 

Let $\kappa := \frac{1}{{Z_1 \ldots Z_m}}$. Using \cref{lem:Z-bound}, if $\xi = \Omega (m \sqrt{E}/\varepsilon)$ then $1 \leq \kappa \cdot(\sqrt{\pi}\xi)^{m} \leq (1+\varepsilon)$.
\begin{align}
\begin{split}
    \alpha/\kappa &=  \bra{\Psi_{AB} (\theta_1, \ldots, \theta_m)} G_{AB} (O_A \otimes \Lambda_B) G_{AB}^\dagger \ket{\Psi_{AB} (\theta_1, \ldots, \theta_m)}\\
    &=:  \bra{\Phi_{B} (\theta_1, \ldots, \theta_m);\xi} M_B (O) \ket{\Phi_{B} (\theta_1, \ldots, \theta_m); \xi}.
\end{split}
\end{align}
Here 
\begin{align}
    \begin{split}
        \ket{\Phi (\theta_1, \ldots, \theta_m) ; \xi} &:=  \ket{V^{(3)} (\theta_1); \xi} \otimes \ldots \otimes \ket{V^{(3)} (\theta_m); \xi},
    \end{split}
\end{align}
and
$$
M_B (O) := \bra{0^n}_A G_{AB} (O_A \otimes \hat{\Lambda}) G_{AB}^\dagger \ket {0^n}_A.
$$
Our main observation is that if we can properly approximate $\ket{\Phi (\theta_1, \ldots, \theta_m); \xi}$ with the sum of at most $R$ Gaussian density matrices, then we can write $\alpha$ as a sum of $R^2$ terms each computable in polynomial time. Based on this motivation, we define the Gaussian rank of density matrices
\begin{definition}
    For a quantum state $\ket{\psi} \in (L^2(\mathbb{R}))^{\otimes n}$ the $\delta$-approximate Gaussian rank $\mathcal{R}_{\delta} (\ket{\psi})$ is the minimum number $r$ such that there exists Gaussian density matrices $\ket{G_1}, \ldots, \ket{G_r}$ and complex number $c_1, \ldots, c_r$ such that 
    $$
    \| \ket{\psi} - (c_1 \ket{G_1} + \ldots + c_r \ket{G_r})\| \leq \delta.
    $$
    \label{def:Gaussian-rank}
\end{definition}

\begin{theorem}
    Let $\hO$ be an observable with $\|\hO\|_\infty \leq 1$. Let $\varepsilon > 0$. Given a polynomial-size CV circuit $\hC$ over $n$ modes, with Gaussian gates and $m$ cubic gates corresponding to angles $\theta_1, \ldots, \theta_m$, such that the largest expected energy at any time throughout the circuit is at most $E$. Let $\mathcal{R}_\delta (\ket{\Phi (\theta_1, \ldots, \theta_m); \xi}) = R$ for $\delta \leq \varepsilon/ \kappa$ to Gaussian states $\ket{G_1}, \ldots, \ket{G_R}$ and coefficients $c_1, \ldots , c_R$ as defined in \cref{def:Gaussian-rank}.  If $\xi > E e^{\Omega( \log^2 (m/\varepsilon))}$ then 
    $$
    |\alpha - \kappa \sum_{j,k} c^*_j c_k \bra{G_j}\hM(O) \ket{G_k}| \leq O (\varepsilon).
    $$
\end{theorem}

\begin{proof}
    Using \cref{lem:cubic-teleportation} for large enough $C$, then we can use gadgets in \cref{fig:cubic-gadget} to replace cubic gates in the circuit with cubic states such that the resulting architecture outputs a quantum state $\ket{\phi}$ such that $\|\hC \ket{0} - \ket{\phi}\| \leq \varepsilon$. Therefore using the formulation outlined in before
    $$
    \|\sqrt{\kappa} \cdot \hK_{AB} \ket{\phi_\xi (\theta_1, \ldots, \theta_m)} - \hC \ket{0}\| \leq \varepsilon,
    $$
  where we used the notation $\ket{\mathbf{x}_A = 0}$ to denote $\ket{x_{A_1} = \ldots = x_{A_m} = 0}$ and defined $\hK_{AB} :=$$ (\hI_A \otimes \bra{\mathbf{x}_{B} = 0}_B) \hG_{AB}(\ket{0}_A \otimes \hI_B)$.
    
By assumption
    \begin{align}
        \begin{split}
            \ket{\Phi (\theta_1, \ldots, \theta_m); \xi} &= c_1 \ket{G_1} + \ldots + c_R \ket{G_R} + \ket{\delta}\\
            &=: \ket{\bar {\Phi}_R} + \ket{\delta},
        \end{split}
    \end{align}
    for $\|\ket{\delta}\| \leq \delta$. Therefore
    $$
    \sqrt{\kappa} \cdot \hK_{AB} \ket{\bar\Phi_R}_B = \hC\ket{0}_A + \ket{\delta'},
    $$
    where $\|\ket{\delta'}\| \leq \kappa \delta + \varepsilon =: \delta'$.
    
    Therefore
    $$
    |\kappa \cdot \bra{\bar \Phi_R} \hM(O) \ket{\bar \Phi_R} - \bra{0}\hC^\dagger \hO \hC \ket{0}| \leq \| O \|_{\infty} \cdot  (\delta'^2 + 2 \delta') \leq O(\varepsilon).
    $$
\end{proof}

    Next, we analyze the classical complexity of computing $\ \bar\alpha$. We observe 
    \asp{
     \bar \alpha &=  \kappa \sum_{i,j = 1}^R c^*_i c_j \bra{G_i} M(O) \ket {G_j}.
    }
    Let $L = \max_j |c_j|$. We need to 
    evaluate each term $O_{i,j}:= \bra{G_i} M(O) \ket {G_j}$ up to precision $\varepsilon/(\kappa \cdot L^2 \cdot R^2)$.  This is the subject of \cref{lem:hdyne-gaussian-expect} below. Suppose we can compute each $O_{i,j}$ using a classical circuit with size
    $t$ and depth $d$. Then $\bar \alpha$ can be computed using a circuit of size $O(R^2 \cdot t \cdot \log L)$ and depth $O(d \cdot \log (R) \cdot \log \log L)$. Furthermore, we can evaluate $\bar \alpha$ as the number of solutions to a nondeterministic time Turing machine running in time $O(\log (R)) \cdot t \cdot \log L$.

\begin{lem}
    Consider $n$-mode Gaussian states $\ket{G}$ and $\ket{G'}$ are specified with $\poly(n)$ size Gaussian circuits, each specified in the phase space representation using $B$ bits of precision. Furthermore, let $\hat\Pi = \otimes_j\ket{q_j}\bra{q_j}$, where $\ket{q_j}$ is a position basis state specified with $B$ bits of precision. Then $\bra{G} \hat \Pi \ket{G'}$ can be computed using a circuit of size $t = \poly(n, B)$, depth $d = \polylog(n, B)$ and space $s = \poly(n, B)$.
    \label{lem:hdyne-gaussian-expect}
\end{lem}
\begin{proof}
Let $\mathbf{q} = [q_j]$
We have 

\begin{align}
\begin{split}
    \bra{G}\hat{\Pi}\ket{G'} &= \braket{G}{\mathbf{q}}\braket{\mathbf{q}}{G'} \\
    &= \overline{\braket{\mathbf{q}}{G}} \braket{\mathbf{q}}{G'}.
    \end{split}
\end{align}

Now all that is left to do is find $\braket{\mathbf{q}}{G}$ for some Gaussian state $\ket{G}$ described by the Gaussian circuit. The covariance matrix $\mathbf{V}$ and displacement vector $\mathbf{\mu}$ of $\ket{G}$ can be found by updating the description according to \cref{eq:gaussian-update} for the $m$ many gates. Clearly the covariance update is more costly, and each gate requires 2 $2n \cross 2n$ matrix multiplications\footnote{This assumes that the Gaussian gates are arbitrary. If the Gaussian gates operate only on constant number of modes, they act nontrivially only on a linear number of terms in the covariance matrix. The final results remain the same though.}. This procedure requires $t = O(2m(2n)^3B^2) = O(\poly(n, B))$ due to the $n^3$ $B$-bit multiplications, and depth $d = 2m(\log(B)+\log(2n)) = O(\polylog(n, B))$. Similarly the calculation of $K$ as defined in \cref{eq:gaussian-position} is $t = O(\poly(n, B))$ and $d = \polylog(n, B)$.

Now by \cref{eq:gaussian-position} we have 

$$\braket{\mathbf{q}}{G} = \psi_G(\mathbf{q}) =  
    \frac{\det(\Re \mathbf{K})^{1/4}}{\pi^{n/4}}
    \exp\!\left(
    -\tfrac{1}{2} (\mathbf{q}-\bar{\mathbf{x}})^T \mathbf K \, (\mathbf{q}-\bar{\mathbf{x}})
    + i\, \bar{\mathbf{p}}^T (\mathbf{q}-\bar{\mathbf{x}})
    \right),
$$
which can also be computed in $t = \poly(n, B)$ and $d = \polylog(n, B)$ due to the determinant calculation.

We also note that the algorithm requires $\poly(n, B)$ auxillary space since the intermediate matrices of the calculation are all $2n \cross 2n$, and can be calculated sequentially.

Therefore the algorithm takes $t = \poly(n, B)$, $d = \polylog(n, B)$, $s = \poly(n, B)$.
\end{proof}

The same proof is also given for the \cite{diasClassicalSimulationNonGaussian2024}-type simulation in \cref{sec:phase-space-grank}.

\begin{lem}\label{lem:useful-gaussian-overlap}
    Given two $n$ mode Gaussian states $\ket {g_1}, \ket{g_2}$ and single mode Gaussian states $\ket {g_3}, \ket{g_4}$, all described in the phase-space representation to $B$ bits of precision,  $\bra{g_2}({\ketbra{g_4}{g_3}\otimes I_{n-1})\ket{g_1}}$ can be computed using a circuit of size $t = \poly(n, B)$ and depth $d = \polylog(n, B)$ and space $s = \poly(n, B)$.
    \label{lem:Gaussian-complx}
\end{lem}
\begin{proof}
    We first note 
    
    $$\bra{g_2}(\ketbra{g_4}{g_3}\otimes I_{n-1})\ket{g_1} = \tr(\ketbra{g_1}{g_2} (\ketbra{g_4}{g_3}\otimes I_{n-1})) = \int_{-\infty}^{\infty}d\mathbf{r} \mathbf{W}_{\ketbra{g_1}{g_2}}(\mathbf{r}) \mathbf{W}_{\ketbra{g_0}{g_0} \otimes I_{n-1}}(\mathbf{r}),$$

    where $\mathbf{W}_O$ is the Wigner function of operator $O$, and $\mathbf{r} = [x_1, p_1, \dots, x_n, p_n]$ represent the quadratures. 

    Now
    $$\begin{aligned}
        \mathbf{W}_{\ketbra{g_4}{g_3} \otimes I_{n-1}}(\mathbf{r}) =& \mathbf{W}_{\ketbra{g_0}{g_0}}([x_1, p_1])\\
        W_{\ketbra{g_1}{g_2}}(\mathbf{r}) =& \langle g_l | g_k \rangle \\
        &\times \frac{1}{\pi \sqrt{\det(\mathbf{V}_{\text{avg}})}} \cdot \exp\left[ -\frac{1}{2} (\mathbf{r} - \mathbf{d}_{\text{avg}})^T \mathbf{V}_{\text{avg}}^{-1} (\mathbf{r} - \mathbf{d}_{\text{avg}}) \right] \\
        & \times \exp\left[ i (\mathbf{d}_k - \mathbf{d}_l)^T \bm{\Omega} (\mathbf{r} - \mathbf{d}_{\text{avg}}) \right],
    \end{aligned}$$
    with $\mathbf{V}_{\text{avg}} = \frac{\mathbf{V}_1 + \mathbf{V}_2}{2}$, and similarly $\mathbf{d}_\text{avg}$.

    We now note that $\braket{g_1}{g_2}$ can be calculated by integrating the position representations of the states (\cref{eq:gaussian-position}).
    
    Finally, to analyze the complexity of calculating the integrals, we note that all the integrals are multivariate complex \textit{gaussian} integrals, which have an analytic solution \cref{eq:MVG-integral}

    Therefore, the circuit that evaluates this value is composed primarily of matrix multiplications and inversions. Therefore the algorithm takes $t = \poly(n, B)$, $d = \polylog(n, B)$, $s = \poly(n, B)$ from the same analysis as the previous lemma.
    
    The same proof is also given for the \cite{diasClassicalSimulationNonGaussian2024}-type simulation in \cref{sec:phase-space-grank}.
\end{proof}

\subsubsection{Gaussian rank of cubic states}

In this part, we evaluate an upper bound on the Gaussian rank of $R = \mathcal{R}_{\delta} (\ket{\Phi (\theta_1, \ldots, \theta_m); \xi})$. We do this in two steps. In \cref{lem:tensor-product-rank} we first relate the rank of the tensor product of quantum states to the Gaussian rank of each state. Next, in \cref{thm:cubic-Gaussian-rank} we bound the Gaussian rank of a single copy of a cubic state and demonstrate that it grows polynomially with the squeezing parameter. 

\begin{lem}
\label{lem:tensor-product-rank}
    Let $\ket{\phi_1}, \ldots, \ket{\phi_m} \in L^2(\mathbb{R})$. Then
    $$
    \mathcal{R}_{\delta'} (\ket{\phi_1} \otimes \ldots \otimes \ket{\phi_m}) \leq \mathcal{R}_{\delta} (\ket{\phi_1}) \ldots \mathcal{R}_{\delta}(\ket{\phi_m}),
    $$
    for $\delta' = m \delta$.
\end{lem}

\begin{proof}
    Let $\ket{\phi} = \bigotimes_{i=1}^m \ket{\phi_i}$. Let $r_k = \mathcal{R}_{\delta}(\ket{\phi_k})$. By Definition 3.3, for each $k$, there exist Gaussian states $\{\ket{G_{k,1}}, \dots, \ket{G_{k,r_k}}\}$ and complex coefficients $\{c_{k,1}, \dots, c_{k,r_k}\}$ that define an approximation
    \[
        \ket{\tilde{\phi}_{k}} = \sum_{j=1}^{r_k} c_{k,j} \ket{G_{k,j}}.
    \]
    The error term, $\ket{\epsilon_k} = \ket{\phi_k} - \ket{\tilde{\phi}_{k}}$, satisfies $\lVert\ket{\epsilon_k}\rVert \le \delta$.

    We construct an approximation for $\ket{\phi}$ as the tensor product of the individual approximations:
    \[
        \ket{\tilde{\phi}} = \bigotimes_{k=1}^m \ket{\tilde{\phi}_k} = \bigotimes_{k=1}^m \left(\sum_{j=1}^{r_k} c_{k,j} \ket{G_{k,j}}\right) = \sum_{J \in \times_{k=1}^m [r_k]} C_J \ket{G_J},
    \]
    where $J=(j_1, \dots, j_m)$, the coefficients are $C_J = \prod_{k=1}^m c_{k,j_k}$, and the basis states are $\ket{G_J} = \bigotimes_{k=1}^m \ket{G_{k,j_k}}$. Each $\ket{G_J}$ is a tensor product of Gaussian states and is therefore itself a Gaussian state. The total number of terms in this decomposition is $R = \prod_{k=1}^m r_k$. This implies that $\mathcal{R}_{\delta'}(\ket{\phi}) \le \prod_k r_k$ for appropriately chosen $\delta'$. Then
    \asp{
    \|\ket{\phi} - \ket{\tilde{\phi}}\| &= \|\sum_{j=1}^m \otimes_{l \neq j} \ket {\phi_j} \otimes \ket{\epsilon_j}\| \leq m \delta'.
    }

\end{proof}

\begin{theorem}\label{thm:cubic-state-gaussian}
For any $\theta = O(1)$, it is the case that  
\begin{align}
\mathcal{R}_{\delta} \left(\ket{V^{(3)} (\theta) ; \xi}\right) \leq O\left( \frac{\xi^{12} \log^6\frac1\delta}{\delta^{2}} \right).
\end{align}
In particular
\begin{align}
\ket{V^{(3)}(\theta); \xi} = \sum_{j=1}^{R_\delta} c_j \ket{G_j},
\end{align}
where 
\begin{align}
\ket{G_j} = \left(\frac{\pi\varepsilon\xi^2}{\xi^2 + \varepsilon} \right)^{\frac14} e^{\frac{y_j^2}{2(\xi^2+\varepsilon)}}\int_x e^{-\frac1{2\varepsilon}(x-y_j)^2 + ix^2y_j - \frac{x^2}{2\xi^2}}\, \ket{x}\,\mathrm dx,
\end{align}
with $y_j = (\frac{1}{\sqrt{\varepsilon}} + \xi)\sqrt{2\log(\frac2\delta)}(\frac{2j-R_\delta}{R_\delta})$, $\varepsilon \leq O(\frac{\sqrt\delta}{ \xi^4 \log^2 \frac1\delta})$, and $c_j = \frac1{R_\delta} \left(\frac{\xi^2 + \varepsilon}{\pi\varepsilon\xi^2} \right)^{\frac14} e^{-\frac{y_j^2}{2(\xi^2+\varepsilon)}} \leq O(1)$.
\label{thm:cubic-Gaussian-rank}
\end{theorem}

\begin{proof}
    For simplicity, assume $\theta = 1$ (the proof is identical for any other $\theta = O(1)$). We begin with expressing $\ket{V^{(3)} (1) ; \xi}$ in the position basis
    \asp{
    \ket{\psi_{\mathrm{target}}} = \ket{V^{(3)} (1) ; \xi} &= \hV^{(3)} (1) \ket{S_\xi}\\
    &= \hV^{(3)} (1) \int \frac{1}{\pi^{1/4} \sqrt{\xi}} e^{-\frac{x^2}{2\xi^2}} \ket{x} dx\\
    &=  \int \frac{1}{\pi^{1/4} \sqrt{\xi}} e^{-\frac{x^2}{2\xi^2} + i x^3/3 } \ket{x} dx\\
    &=: \int_x \psi(x) \ket{x} dx.
    }
    Next, we use the following identity to decompose cubic phases into quadratic phases (see \Cref{sec:phase-decomp}):

\begin{align}
e^{i\frac{x^3}3 } e^{-\frac{\varepsilon x^4}{18}} = \frac{1}{\sqrt{2\pi\varepsilon}}\int_y e^{-\frac1{2\varepsilon}(x-y)^2 + i\frac{1}{3}x^2 y}\, \mathrm dy,
\end{align}
and hence
\begin{align}\label{eq:limit-x3}
e^{i\frac{x^3}3 } = \lim_{\varepsilon \to 0^+} \frac{1}{\sqrt{2\pi\varepsilon}}\int_y e^{-\frac1{2\varepsilon}(x-y)^2 + i\frac{1}{3}x^2 y}\, \mathrm dy.
\end{align}
Note that on the RHS of \eqref{eq:limit-x3} we only have Gaussian forms in $x$. Our proof strategy is as follows: 
\begin{enumerate}
    \item we show that choosing  
\begin{align}\label{eq:choice-of-eps}
\varepsilon \leq \frac{\sqrt\delta}{4 \xi^4 \log^2 \frac2\delta}
\end{align}
in \eqref{eq:limit-x3} guarantees that the state
\begin{align}\label{eq:gaussian-approx}
\begin{split}
\ket{\psi_{\varepsilon}} &= \frac{1}{(\pi)^{1/4}\sqrt{2\pi\varepsilon\xi}}\int_y \int_x e^{-\frac{x^2}{2\xi^2}} e^{-\frac1{2\varepsilon}(x-y)^2 + i\frac{1}{3}x^2 y}\,  \ket{x} \,\mathrm dx \mathrm dy\\
&=: \int_x \psi_\varepsilon (x) \ket{x} dx
\end{split}
\end{align}
is $\delta$-close to the target state $\ket{\psi_{\mathrm{target}}}$.
\item Note that \eqref{eq:gaussian-approx} consists of only Gaussian states (note that the exponent is quadratic in $x$). Next, we show that one can restrict the integral of \eqref{eq:gaussian-approx} to $y\in[-K,K]$ with
\begin{align}
K = O\left( \frac{\xi^2\log\frac1\delta}{\delta^{1/4}} \right),
\end{align}
and get a $\delta$ approximation to $\ket{\psi_{\varepsilon}}$.
\item Lastly, we show that the restricted integral can be approximated within accuracy $\delta$ by a Riemann summation including
\begin{align}
N = O\left( \frac{\xi^{12} \log^6\frac1\delta}{\delta^{2}} \right)
\end{align}
many terms.
\end{enumerate}   
We will now prove each of the items above:
\begin{enumerate}
\item To show the closeness of $\psi_{\varepsilon}$ to our cubic phase state we note that
\begin{align}
\abs{\psi_\varepsilon(x) - \psi(x)} = |\psi(x)| \left(1-e^{-\frac{\varepsilon x^4}{18}}\right) = \frac{1}{\pi^{1/4}\sqrt{\xi}}e^{-\frac{x^2}{2\xi^2}} \left(1-e^{-\frac{\varepsilon x^4}{18}}\right),
\end{align}
which gives
\begin{align}\label{eq:tar-eps}
\begin{split}
\norm{\psi_{\mathrm{target}} - \psi_{\varepsilon}}^2 &= \int_{x=-\infty}^\infty \frac{1}{\sqrt{\pi}\xi}e^{-\frac{x^2}{\xi^2}} \left(1-e^{-\frac{\varepsilon x^4}{18}}\right)^2\, \mathrm dx\\
&\leq \int_{|x|\geq \xi \sqrt{2\log\frac2\delta}}  \frac{1}{\sqrt{\pi}\xi}e^{-\frac{x^2}{\xi^2}} \mathrm dx + \int_{|x|\leq \xi \sqrt{2\log\frac2\delta}} \frac{1}{\sqrt{\pi}\xi}e^{-\frac{x^2}{\xi^2}} \left( 1 - e^{-\frac{\varepsilon x^4}{18}}\right)^2 \, \mathrm dx.
\end{split}
\end{align}

 We can use the fact that for a Gaussian random variable $X$ with standard deviation $\sigma$ and mean $\mu$, one has
\begin{align}\label{eq:gaussian-tail-bound}
\Pr[|X-\mu|\geq \sigma \sqrt{2\log\frac2\delta}] \leq \delta,
\end{align}
to bound the first term in \eqref{eq:tar-eps} by $\delta$. Furthermore, our choice of \eqref{eq:choice-of-eps} guarantees
\begin{align}
\left( 1-e^{-\frac{\varepsilon x^4}{2}} \right)^2 \leq \delta.
\end{align}
Therefore, we have proved
\begin{align}
\norm{\psi_{\varepsilon} - \psi_{\mathrm{target}}} \leq 2\delta.
\end{align}
Note that re-scaling the dependencies of all parameters with respect to $\delta$ (by a mere factor of $2$), we can get $\delta$ closeness. 
\item Now, we show that one can indeed truncate the integral \eqref{eq:gaussian-approx} over $y$. To do so, we define the function
\begin{align}\label{eq:def-of-f}
f_\varepsilon(x) := \frac{1}{\sqrt{2\pi\varepsilon}}\int_{|y|\leq (\frac{1}{\sqrt\varepsilon} + \xi)\sqrt{2\log\frac2{\delta}}} e^{-\frac1{2\varepsilon}(x-y)^2 + i\frac{1}{3}x^2 y}\, \mathrm{dy}.
\end{align}
We note that for any $x\in\mathbb R$ we have
\begin{align}\label{eq:uniform-bound-on-f}
|f_\varepsilon(x)| \leq \frac{1}{\sqrt{2\pi\varepsilon}}\int_{|y|\leq (\frac{1}{\sqrt\varepsilon} + \xi)\sqrt{2\log\frac2{\delta}}} e^{-\frac1{2\varepsilon}(x-y)^2}\, \mathrm dy \leq \frac{1}{\sqrt{2\pi\varepsilon}}\int_{y=-\infty}^\infty e^{-\frac1{2\varepsilon}(x-y)^2}\, \mathrm dy = 1,
\end{align}
and furthermore that for any $|x|\leq \xi\sqrt{2\log\frac2\delta}$ we have
\begin{align}\label{eq:point-wise-bound-on-delta}
\begin{split}
\abs{f_\varepsilon(x) - \frac{1}{\sqrt{2\pi\varepsilon}}\int_y e^{-\frac1{2\varepsilon}(x-y)^2 + i\frac{1}{3}x^2 y}\, \mathrm dy} &= \abs{\frac{1}{\sqrt{2\pi\varepsilon}}\int_{|y|\geq (\frac{1}{\sqrt\varepsilon} + \xi)\sqrt{2\log\frac2{\delta}}} e^{-\frac1{2\varepsilon}(x-y)^2 + i\frac{1}{3}x^2 y}\, \mathrm dy}\\
&\leq \frac{1}{\sqrt{2\pi\varepsilon}}\int_{|y|\geq (\frac{1}{\sqrt\varepsilon} + \xi)\sqrt{2\log\frac{2}{\delta}}} e^{-\frac1{2\varepsilon}(x-y)^2}\, \mathrm dy \leq \delta,
\end{split}
\end{align}
where the last line uses \eqref{eq:gaussian-tail-bound} again. As a result:
\begin{align}
\begin{split}
&\int_x \frac{1}{\sqrt{\pi}\xi}e^{-\frac{x^2}{\xi^2}}\abs{f_\varepsilon(x) - \frac{1}{\sqrt{2\pi\varepsilon}}\int_y e^{-\frac1{2\varepsilon}(x-y)^2 + i\frac{1}{3}x^2 y}\, \mathrm dy}^2\, \mathrm dx\\
&= \int_{|x|\leq \xi \sqrt{2\log\frac2\delta}} \frac{1}{\sqrt{\pi}\xi}e^{-\frac{x^2}{\xi^2}}\underbrace{\abs{f_\varepsilon(x) - \frac{1}{\sqrt{2\pi\varepsilon}}\int_y e^{-\frac1{2\varepsilon}(x-y)^2 + i\frac{1}{3}x^2 y}\, \mathrm dy}}_{\leq \delta}^2\, \mathrm dx\\
& \qquad + \int_{|x|\geq \xi \sqrt{2\log\frac2\delta}} \frac{1}{\sqrt{\pi}\xi}e^{-\frac{x^2}{\xi^2}}\underbrace{\abs{f_\varepsilon(x) - \frac{1}{\sqrt{2\pi\varepsilon}}\int_y e^{-\frac1{2\varepsilon}(x-y)^2 + i\frac{1}{3}x^2 y}\, \mathrm dy}}_{\leq 2}^2\, \mathrm dx\\
&\leq \delta^2 + 2\delta \leq 3\delta,
\end{split}
\end{align}
where the bound on the first term comes from \eqref{eq:point-wise-bound-on-delta}, whereas the second term comes from using \eqref{eq:uniform-bound-on-f} together with \eqref{eq:gaussian-tail-bound}. Note that we have chosen truncation parameter $K = (\frac{1}{\sqrt\varepsilon} + \xi) \sqrt{2\log \frac{2}{\delta}} = O\left( \frac{\xi^2\log\frac1\delta}{\delta^{1/4}} \right)$.

\item We now analyze the error introduced by approximating the integral form of $f_\varepsilon$ in \eqref{eq:def-of-f} by a Riemann sum. To do so, we start from a lemma:
\begin{lem} [Riemannian discretization]\label{lem:Reimann}
Let $g:\mathbb R \to \mathbb C$ be a differentiable function such that $|g'(x)|\leq L$ for all $x\in[a,b]$. Then the Reimann sum satisfies
\begin{align}
\abs{ \sum_{i=1}^N g(x_i) \left( \frac{b-a}{N} \right) - \int_{x=a}^b g(x)\, \mathrm dx} \leq \frac{L(b-a)^2}{2N},
\end{align}
where $x_i = a + \frac{i}{N}(b-a)$ are
the discretization points.
\end{lem}
\begin{proof}
This is a standard statement, whose proof can be found in standard textbooks about numerical methods of calculations (c.f. \cite{isaacson2012analysis}). Nevertheless, we provide a proof for completeness: for any 
$x\in[x_i, x_{i+1}]$ we can write
\begin{align}
g(x_i) - g(x) = \int_{y=x_i}^x g'(x)\, \mathrm dy,
\end{align}
and hence we can write
\begin{align}
\begin{split}
&\abs{ \int_{x_i}^{x_{i+1}} \left(g(x_i) -  g(y)\right) \mathrm dy } \leq \int_{x_i}^{x_{i+1}} |g(x_i) - g(y)| \mathrm dy\\
&= \int_{x_i}^{x_{i+1}} \abs{\int_{z=x_i}^y g'(z) \, \mathrm dz} \, \mathrm dy\leq L \int_{x_i}^{x_{i+1}} (y-x_i) \, \mathrm dy = \frac{L(x_{i+1} - x_i)^2}{2}.
\end{split}
\end{align}
Given that $(x_{i+1}-x_i) = (b-a)/N$ and the fact that the final difference is bounded by $N$ times the difference above, we get the desired bound.
\end{proof}
Now, we upper bound the largest derivate of our integrand for a given $x$:
\begin{align}
\abs{\partial_y\left(\frac{1}{\sqrt{2\pi\varepsilon}} e^{-\frac1{2\varepsilon}(x-y)^2 + i\frac{1}{3}x^2 y}\right)} \leq \frac{e^{-\frac{1}{2\varepsilon}(x-y)^2}}{\sqrt{2\pi\varepsilon}}  \left(\frac{1}{\varepsilon}|x-y| + \frac{x^2}3 \right)\leq \frac{1}{\sqrt{2\pi\varepsilon}} \left(\frac{|x|}{\varepsilon} + \frac{x^2}{3} + \frac{K}{\varepsilon} \right).
\end{align}
Hence, letting $\hat f_{\varepsilon,N}(x)$ represent the Reimann sum of $f_\varepsilon(x)$ with $N$ discretization points, and plugging into \cref{lem:Reimann} we get
\begin{align}
\abs{f_{\varepsilon}(x) - \hat{f}_{\varepsilon, N}(x)} \leq \frac{2K^2}{N\sqrt{2\pi\varepsilon}} \left(\frac{|x|}{\varepsilon} + \frac{x^2}{3} + \frac{K}{\varepsilon} \right).
\end{align}
We also let $\ket{\hat\psi_{\varepsilon,N}}$ represent the corresponding state (which is a finite sum of Gaussians), and we get
\begin{align}
\abs{\psi_{\varepsilon}(x) - \hat\psi_{\varepsilon,N}(x)} = \frac{e^{-\frac{x^2}{2\xi^2}}}{\pi^{1/4}\sqrt{\xi}} \abs{f_{\varepsilon}(x) - \hat{f}_{\varepsilon, N}(x)}\leq \frac{e^{-\frac{x^2}{2\xi^2}}}{\pi^{1/4}} \frac{2K^2}{N\sqrt{2\pi\varepsilon \xi}} \left(\frac{|x|}{\varepsilon} + \frac{x^2}{3} + \frac{K}{\varepsilon} \right),
\end{align}
which results in
\begin{align}
\norm{\psi_{\varepsilon} - \hat{\psi}_{\varepsilon, N}}^2 \leq \frac{4K^4}{2\pi\varepsilon N^2}\int_{x} \frac{e^{-x^2/\xi^2}}{\sqrt{\pi}\xi} \left( \frac{|x|}{\varepsilon} + \frac{x^2}{3} + \frac{K}{\varepsilon}  \right)^2\leq \frac{4K^4}{2\pi\varepsilon N^2} \cdot O(\frac{K^2}{\varepsilon^2}) = O(\frac{K^6}{\varepsilon^3 N^2}),
\end{align}
where the last line follows from the fact that $\mathbb E[X^{2m}]\leq O(\xi^{2m})$ for $m=O(1)$, and the scaling of $K$ obtained above. Setting $N = O(\frac{K^3}{\delta^{1/2} \varepsilon^{3/2}})$ results in $\delta$-closeness. This scaling is the same as in the statement, after plugging in the bounds for parameters $K$ and $\varepsilon$.
\end{enumerate}

\end{proof}

\begin{corollary}\label{corol:Gaussian-rank-of-cubic}
For $\theta_1, \ldots, \theta_m = O(1)$, it is the case that 
\begin{align}
\mathcal{R}_{m\delta} (\ket{V_1(\theta_1); \xi} \otimes \ldots \otimes \ket{V_m(\theta_m); \xi}) \leq O\left(\left(\frac{\xi^{12}}{\delta^{2}} \log^6\frac{1}{\delta}\right)^{m}\right).
\end{align}
Moreover, the mean, variance, and coefficients of the $i$-th Gaussian can be computed with precision $\varepsilon$ in time $\mathsf{poly}(m,\log(\xi, \delta^{-1}, \varepsilon^{-1}))$.
\end{corollary}

\paragraph{Extension to higher-degree phase gates} 
\label{sec:phase-decomp}

The crux of our formulas for Gaussian decompositions is the Gaussian moment generating function
\begin{align}\label{eq:gaussian-mgf}
\frac{1}{\sqrt{2\pi}\sigma}\int_y e^{-\frac1{2\sigma^2}(\mu-y)^2 + i\omega y} \, \mathrm dy= e^{-\frac12\omega^2\sigma^2} \cdot e^{i\omega \mu},
\end{align}
which immediately gives
\begin{align}
\lim_{\varepsilon\to0^+}\frac{1}{\sqrt{2\pi\varepsilon}}\int_y e^{-\frac1{2\varepsilon}(x-y)^2 + ix^2 y}\, \mathrm dy = e^{ix^3},
\end{align}
allowing us to decompose the cubic phase state into a sum of Gaussians. We can go further and write a general formula for higher-degree phase states. Such states are defined as $\ket{V^{(k)}(\theta)}:= \int_x \exp(i\frac{\theta}{k}x^k) \ket x\, \mathrm dx$. 
\begin{lem}
For any $k>2$, we have
\begin{align}
e^{i\frac{\theta}{k}x^k} = \lim_{\varepsilon\to0^+}\frac{1}{\left(  2\pi\varepsilon\right)^{\frac{k}2-1}}\int\dots\int e^{ix^2 y_1 \cdots y_{k-2}} e^{-\frac1{2\varepsilon}\sum_{i=1}^{k-1}(x-y_{i})^2}\, \mathrm dy_1\cdots\mathrm dy_{k-2}.
\end{align}
\end{lem}

\begin{remark}
    We conjecture the Gaussian rank bound for degree $k$ states can be improved to the following bound:
    \begin{conjecture}
    The Gaussian rank of $\mathcal{R}_{\delta}(\ket{V^{(k)} (\theta); \xi}^{\otimes m}) \leq O(\xi^{m(k-2)/2})$.    
    \end{conjecture}
    The reason for the quadratic improvement is that we might be able to group copies of $\ket{V^{(k)} (\theta); \xi}$ into pairs and bound the Gaussian rank of each pair by $O(\xi^{k-2})$. Similar grouping works well in the case of the non-Gaussian state $\ket{1}^{\otimes m}$, since we can approximate $\ket{11}$ by using only two Gaussian states. 
\end{remark}

\begin{remark}
    We could have tried using the decomposition of cubic states into ``linear phases.'' However, a naive approach based on decomposition in the Fourier basis leads to an exponential bound for the rank in terms of the squeezing parameter. That is due to the heavy-tailed behavior of the Fourier coefficients.  We can write the cubic term in the Fourier basis
    $e^{i x^3/3} = \int  Ai (-p) e^{i p x}  dp$
    where $Ai$ is the Airy function defined as
    $Ai (k) = \frac{1}{2\pi} \int e^{i x^3/3 + i k x} dx$. 
    Therefore, we find the continuous superposition over Gaussian states $\ket{V^{(3)} (\theta) ; \xi} = \int  Ai (-p)  (\ket{S_{\xi,p}})  dp$, where $\ket{S_{\xi,p}} := (e^{i p \hat{X}} \ket{S_{\xi}})$
    Here $\ket{S_{\xi,p}} := (e^{i p \hat{X}} \ket{S_\xi}$ is a Gaussian state. The Airy function $\mathrm{Ai}(x)$ exhibits different asymptotic behavior depending on the sign of $x$: while for large positive $x \rightarrow \infty$, $\mathrm{Ai}(x)$ decays exponentially:
    $\mathrm{Ai}(x) \sim \frac{1}{2\sqrt{\pi}} x^{-1/4} \exp\left(-\frac{2}{3}x^{3/2}\right)$, for large negative $x \rightarrow - \infty$, $\mathrm{Ai}(x)$ oscillates with a slowly decaying amplitude: $\mathrm{Ai}(x) \sim \frac{1}{\sqrt{\pi} |x|^{1/4}} \sin\left(\frac{2}{3}|x|^{3/2} + \frac{\pi}{4}\right)$. Therefore, to discretize the integration, we need a number of samples growing exponentially in $\xi$.

\end{remark}


\subsection{Simulating $\mathsf{CVBQP}[X^3]$ with effective exponential energy in $\mathsf{PP}$}\label{sec:cvbqp-in-pp}

We now present an algorithm to simulate $\mathsf{CVBQP}[X^3]$ circuits without an energy upper bound within $\PP$. In particular, we show
\begin{theorem}\label{thm:XCubedInPP}
    Let $\hC$ be a CV circuit over $n$ modes composed of $\mathrm{poly}(n)$ Gaussian and cubic gates specified with constant bits of precision. We assume that for any precision parameter $\delta=1/\mathsf{exp}(n)$, there exists a state of polynomial computable energy $E = \mathsf{exp}(n)$, such that the state at time $t$ is $\delta$-close to a state of energy $E$. Then, using one query to a $\PP$ oracle\footnote{Note that $\class{P}^{\PP[1]}$, i.e.~$\class{P}$ with one query to $\PP$, equals $\PP$ itself~\cite{FORTNOW19961}.}, one can decide if the output probabilities of $\hC$ according to number basis measurement are in $[a,a+\mathsf{poly}(a)]$ with probability at least $2/3$.
\end{theorem}
\noindent  We note that, as depicted in \cref{sec:doubly-exp-growth}, energy in such circuits may grow up to doubly exponential in the number of cubic gates. Therefore, by assuming an exponential energy promise, we are not simulating $\mathsf{CVBQP}[X^3]$ in its entirety. However, the realization of exponential energy scaling remains a significant experimental challenge, and thus such a bound is still physically relevant. Furthermore, from \Cref{prop:E-growth-cubic-depth} we have that for circuits with $\polylog$ depth, the energy of the circuit is upper bounded by $\exp$, satisfying our promise condition.

Note that this energy bound allows choosing a polynomial-time computable
\begin{align}
\xi \leq \mathsf{exp}(n),
\end{align}
for our teleportation purposes in \cref{lem:cubic-teleportation}. The reason is that we can treat the input state of our teleportation circuit to have a cutoff $E = \mathsf{exp}(n)$ at any time $t\in\{0,\cdots,T-1\}$, and lose only an exponentially small amount of precision in trace distance. We call the truncated states $\Pi_E\ket{\psi_t}$, the \textit{proxy states}, and we can show that having the promise that our state is close enough to a state of low energy $E$ (the proxy state), we can still apply our cubic teleportation gadget. We use 
\begin{align}
K_0 =(\mathbb I \otimes \bra{x=0}) \mathrm{SUM}^{-1}(\mathbb I \otimes \ket{V^{(3)}(\theta); \xi}),
\end{align}
to denote the post-selection on $x=0$ in \cref{fig:cubic-gadget}. We have

\begin{lem}[Cubic gate teleportation with a proxy state promise]\label{lem:lambda}
Consider the teleportation gadget of \cref{lem:cubic-teleportation}. Assume there exists a state $\phi$ of energy $E$ such that
\begin{align}
\norm{\ket\psi - \ket\phi} \leq \varepsilon.
\end{align}
We have
\begin{align}
\norm{\frac{1}{\sqrt{Z_\psi}}K_0 \ket\psi - \frac{1}{\sqrt{Z_\phi}}K_0\ket\phi}\leq 2(1+\lambda)\varepsilon,
\end{align}
whenever $\xi\geq \frac{\sqrt{E+1}}{\lambda}$, for $\varepsilon < 1- 2 \lambda$, indicating an accurate simulation. Here we have used $Z_\psi = \norm{K_0 \ket{\psi}}$
\end{lem}
\begin{proof}
Note that
\begin{align}
\norm{K_0(\ket{\psi} - \ket{\phi})}^2 = \frac{1}{\pi^{\frac12}\xi}\int_x |\psi(x) - \phi(x)|^2\, e^{-\frac{x_1^2}{\xi^2}}\mathrm dx \leq \frac{1}{\pi^{\frac12}\xi} \norm{\ket\psi - \ket\phi}^2\leq \frac{\varepsilon^2}{\pi^{\frac{1}{2}}\xi}.
\end{align}
We then use
\begin{lem}\label{lem:dummy}
Let $u,v$ be nonzero vectors. If $\|u-v\|\le\varepsilon$, then
\[
\Big\|\frac{v}{\|v\|}-\frac{u}{\|u\|}\Big\|
\le 2\frac{\varepsilon}{\max\{\|u\|,\|v\|\}}.
\]
\end{lem}
\begin{proof}
Write $u=\ket a$ and $v=\ket b$ for brevity; assume $u,v\neq 0$.  We have
\[
\frac{u}{\|u\|}-\frac{v}{\|v\|}
=\frac{u-v}{\|u\|}+v\!\left(\frac{1}{\|u\|}-\frac{1}{\|v\|}\right).
\]
Taking norms and using $|\|u\|-\|v\||\le\|u-v\|$ gives
\[
\Big\|\frac{u}{\|u\|}-\frac{v}{\|v\|}\Big\|
\le \frac{\|u-v\|}{\|u\|}+\frac{\|v\|\,|\|u\|-\|v\||}{\|u\|\,\|v\|}
\le \frac{\|u-v\|}{\|u\|}+\frac{\|u-v\|}{\|u\|}
= \frac{2\|u-v\|}{\|u\|}.
\]
\end{proof}
We use \cref{lem:post-select-on-zero} to conclude that whenever $\xi\geq \frac{\sqrt{E+1}}{\lambda}$ we have $\frac{1}{\sqrt{\pi} \xi} (1 - 3\lambda) \leq Z_\psi, Z_\phi \leq \frac{1}{\sqrt{\pi} \xi}$
\begin{align}
\norm{\frac{1}{\sqrt{Z_\psi}}K_0 \ket\psi - \frac{1}{\sqrt{Z_\phi}}K_0\ket\phi} \leq 2(1+\lambda)\varepsilon.
\end{align}
\end{proof}
As an immediate corollary, we get
\begin{corollary}\label{corol:k0}
Assuming the proxy state $\ket\phi$ promise as in \cref{lem:lambda} with $E =\mathsf{exp}(n)$, we have that
\begin{align}
\norm{\frac{1}{\sqrt{Z_\psi}}K_0 \ket{\psi} - V^{(3)}(\theta)\ket{\psi}}\leq \frac{1}{\mathsf{exp}(n)},
\end{align}
with $\xi=\mathsf{exp}(n)$.
\end{corollary}
\begin{proof}
From \cref{lem:cubic-teleportation} we have that $\norm{\frac{1}{\sqrt{Z_\phi}} K_0\ket{\phi} - V^{(3)}\ket{\phi}}\leq \frac{1}{\mathsf{exp}(n)}$. Furthermore, as $\norm{\ket{\phi}-\ket{\psi}}=\norm{V^{(3)}(\theta)\ket{\phi}-V^{(3)}(\theta)\ket{\psi}}\leq \frac{1}{\mathsf{exp}(n)}$, and by using \cref{lem:lambda}, we obtain the result.
\end{proof}

Finally, we show that it is safe to teleport through the sum of Gaussian approximation of the cubic phase state \cref{thm:cubic-Gaussian-rank}:
\begin{lem}\label{lem:single-mode-B}
Let $\ket{\phi}$ be a state such that $|\phi(x)|\leq B$ for all $x\in\mathbb R$ be the input to the gadget in \cref{fig:cubic-gadget}. Then replacing the finite-energy cubic state $\ket{\psi_1}=\ket{V^{(3)}(\theta);\xi}$ of \cref{fig:cubic-gadget} by $\ket{\psi_2}=\sum_{i=1}^{R_\delta} c_i\ket{G_i}$ (as in \cref{thm:cubic-Gaussian-rank}) , and post-selecting on $\ket{x=0}$ on the homodyne outcome, produces a $O(\sqrt{\xi}B\delta)$-accurate to the result of injection with $\ket{V^{(3)}(\theta);\xi}$ after renormalization.
\end{lem}
\begin{proof}
For each model $i = 1,2$, the state after post-selection on $\ket{x=0}$ has the wavefunction given by $\phi(x)\psi_i(x)$. Now, we have that the difference between the two states, after post-selecting on $\ket{x=0}$ but before renormalization, is
\begin{align}
\int_x |\phi(x) (\psi_1(x) - \psi_2(x))|^2\, \mathrm dx \leq B^2 \norm{\psi_1-\psi_2}^2 = B^2\delta^2.
\end{align}
Finally, using \cref{lem:dummy}, and the fact that the normalization factor of teleportation via approximate cubic phase state is $\norm{\psi_1\cdot\phi}^2 = Z = \Theta({\sqrt \xi})$, we conclude the proof.
\end{proof}
We also use an intermediary lemma, which becomes useful later as well
\begin{lem}\label{lem:bounded-overlap}
    Let $\ket{G}$ be a Gaussian state with $\norm{\boldsymbol\Sigma}\leq B$. Then, for any $q\in\mathbb R$ we have
    \begin{align}
    \norm{(\bra{x=q}\otimes \mathbb I) \ket{\psi}}\leq \left({\frac{B}{\pi}}\right)^{\frac14}.
    \end{align}
    \end{lem}
    \begin{proof}
    Note that the condition $\norm{\boldsymbol\Sigma}\leq B$ implies $\sigma_{p_1}^2 \leq B^2$, and hence from Heisenberg uncertainty principle, we get $\sigma_{q_1}^2 \geq \frac{1}{4\sigma_p^2} \geq \frac{1}{4B^2}$. Since any marginal distribution of a Gaussian state is a Gaussian, we get that the maximum of $p(q_1) = \frac{1}{\sqrt{2\pi\sigma_{q_1}}}\leq \sqrt{\frac{B}{\pi}}$. Finally, note that $\norm{(\bra{q}\otimes \mathbb I) \ket{\psi}}^2 = p(q)$.
    \end{proof}
We can now extend \cref{lem:single-mode-B} to a multimode setting:
\begin{lem}
Let $\ket{\phi} = \sum_{i=1}^N d_i \ket{g_i}$ be an $n$-mode state where $|d_i|\leq \mathsf{exp}(n)$, $N\leq \mathsf{exp}(n)$, and $\ket{g_i}$ is a Gaussian with at most exponential covariance and means. Then, replacing the finite-energy cubic state $\ket{\psi_1} = \ket{V^{(3)}(\theta); \xi}$ of \cref{fig:cubic-gadget} by $\ket{\psi_2} = \sum_{i=1}^{R_\delta} c_i \ket{G_i}$ (as in \cref{thm:cubic-Gaussian-rank}), and post-selecting on $\ket{x=0}$ on the homodyne outcome, produces a $O(\sqrt{\xi} \delta \,\mathsf{exp}(n))$.
\end{lem}
\begin{proof}
Let $\rho_1 = \tr_{[2..n]}\left( \ket{\phi}\bra{\phi} \right)$. Let $\sigma_i = \tr_{[2..n]}\left( \ket{g_i}\bra{g_i} \right) $. We first bound $\bra{x}\rho_1\ket{x}$:
\begin{align}
\begin{split}
\bra{x}\rho_1\ket{x} &= \int_y \langle{x,y} \ket{\phi}\bra{\phi}x,y\rangle\, \mathrm dy\\
&= \sum_{i,j=1}^{N} d_i d_j^\ast \int_y \langle x,y \ket{g_i} \bra{g_j}x,y\rangle\\
&\leq \sum_{i,j=1}^N d_i d_j^\ast \sqrt{\int_y |\langle x,y\ket{g_i}|^2} \sqrt{\int_y |\langle x,y\ket{g_j}|^2}\\
&= \sum_{i,j=1}^N d_i d_j^\ast \sqrt{\bra{x}\sigma_i\ket{x}} \sqrt{\bra{x}\sigma_i\ket{x}}\\
&\leq \mathsf{exp}(n),
\end{split}
\end{align}
where the first inequality is due to Cauchy-Schwarz, and in the last line we have used \cref{lem:bounded-overlap} to bound $\bra{x}\sigma_i\ket{x}$ by an exponential in $n$. 
\end{proof}
Later in our proof, we show that the coefficient $c_i$  for the sum of Gaussians state that we use in simulation are bounded by an exponential, and so are the covariance and mean of each Gaussian in our simulation, and hence, choosing $\delta=\frac{1}{\mathsf{exp}(n)}$ suffices to let us teleport through a sum of Gaussians.

Finally, we recall a lemma about the coherent decomposition of number states, before proceeding to the proof:
\begin{lem}[Restating Corollary 1 of \cite{marshallSimulationQuantumOptics2023}]\label{lem:coherent-decomposition-of-number-states}
Let $\ket{n}$ be any number state with $n>1$. There is a real number $C_n \leq  O(\frac1{\sqrt\delta})$ and coherent states $(\ket{\alpha_i})_{i=0}^{n}$ with $|\alpha_i| \leq \sqrt{n}$ such that state $\ket{\tilde n}$ defined as
\begin{align}
\ket{\tilde n} := \frac{C_n}{\sqrt{\mathcal N}} \sum_{k=0}^n e^{-i2\pi kn/(n+1)} \ket{\alpha_i},
\end{align}
satisfies
\begin{align}
\abs{\langle n|\tilde n \rangle}^2 = 1-\delta.
\end{align}
\end{lem}

\begin{proof}
We recall that according to \cite[Corollary 1]{marshallSimulationQuantumOptics2023}, we have for any $r\in\mathbb R_+$ that the state
\begin{align}\label{eq:direct-MA-formula}
\ket{\tilde n_r}:= \frac{1}{\sqrt{\mathcal N}}\frac{\sqrt{n!}}{n+1} \frac{e^{r^2/2}}{r^n} \sum_{k=0}^n e^{-i2\pi kn/(n+1)} \ket{r e^{i2\pi k/(n+1)}},
\end{align}
satisfies
\begin{align}
\abs{ \langle n | \tilde n_r\rangle}^2 = \frac{1}{\mathcal N} = \frac{1}{n! \sum_{k=0}^\infty \frac{r^{2k(n+1)}}{(k(n+1)+n)!}}.
\end{align}
In what follows, we show that
\begin{enumerate}
    \item If $n\geq \log \frac 1\delta$, it suffices to choose $r=\sqrt n$. We also show that $C_n\leq O(1)$.
    \item If $n\leq \log \frac{1}{\delta}$, then we can choose $r = O(\sqrt{n \delta})$. In this case, we show that $C_n \leq O(\frac{1}{\sqrt \delta})$.
\end{enumerate}
Below is the proof for the above items:
Let us begin by choosing $r = \sqrt {n \, e^{-\beta}}$. We get that
\begin{align}
\mathcal N = n! \sum_{k=0}^{\infty} \frac{n^{k(n+1)} e^{-2\beta k(n+1)}}{(k(n+1)+n)!}.
\end{align}
Using Stirling's formula, we have the elementary inequalities
$\sqrt{2\pi \ell} \left(\frac{\ell}{e}\right)^\ell \leq \ell ! \leq 2\sqrt{2\pi \ell} \left(\frac{\ell}{e}\right)^\ell$ \cite{flajolet2009analytic}. Therefore, using Stirling's approximation, we get
\begin{align}\label{eq:bound-on-fidelity-coherent-decomp}
\begin{split}
\mathcal N &= 1 + 2\sum_{k=1}^\infty \sqrt{\frac{n}{k(n+1)+n}}\left(\frac{n}{k(n+1)+n}\right)^{k(n+1)+n} e^{(1-\beta) k(n+1)}\\
&\leq 1 + 2 \sum_{k=1}^\infty  \left( \frac{1}{1+k}  \right)^{k(n+1)+n+\frac12} e^{(1-\beta)k(n+1)}.
\end{split}
\end{align}
We separate the cases here:
\begin{enumerate}
    \item If $n\geq \frac{e}{4}\log \frac{4e}{\delta}$: we choose $\beta = 0$, and \eqref{eq:bound-on-fidelity-coherent-decomp} becomes
    \begin{align}
    \begin{split}
    \mathcal N &\leq 1 + 2\sum_{k=1}^\infty \left( \frac{1}{k+1} \right)^{k(n+1)+n+\frac12} e^{k(n+1)}\\
    &=1 + 2\sum_{k=1}^\infty \left( \frac{e}{k+1} \right)^{k(n+1)} \left( \frac{1}{k+1} \right)^{n+\frac12}\\
    &= 1 + \frac{e}{\sqrt 2} \left( \frac{e}4 \right)^n + 2\sum_{k=2}^\infty \left( \frac{e}{k+1} \right)^{k(n+1)} \left( \frac{1}{k+1} \right)^{n+\frac12}\\
    &\leq 1 + \frac{e}{\sqrt 2} \left( \frac{e}4 \right)^n + 2\sum_{k=2}^\infty \left(\frac{1}{k+1}\right)^{n+1/2}\\
    &\leq 1 + \frac{e}{\sqrt 2}\left( \frac{e}{4}  \right)^n + 2\int_{x=2}^\infty \frac{\mathrm dx}{x^{n+\frac12}}\\
    &= 1 + \frac{e}{\sqrt 2}\left( \frac{e}{4}  \right)^n + \frac{2^{\frac12 - n}}{n - \frac12}\\
    &\leq 1 + 4e \left( \frac{e}{4} \right)^n \leq 1+\delta.
    \end{split}
    \end{align}
\item If $n<\frac{e}{4}\log\frac{4e}{\delta}$, from \cref{eq:bound-on-fidelity-coherent-decomp} we observe that we should introduce a factor of $e^{-\beta k(n+1)}$ for the $k$-th term in the sum. As $k\geq 1$, we get that the error terms are multiplied by a number smaller than $e^{-\beta (n+1)}$. As a result, the entire error term shrinks by a factor of (at least) $e^{-\beta (n+1)}$. Hence, using the bound from the previous condition, we obtain
\begin{align}
\mathcal N \leq 1 + 4e\left( \frac{e}{4} \right)^n e^{-\beta(n+1)} \leq 1 + 4e \, e^{-\beta n},
\end{align}
and hence, choosing
\begin{align}\label{eq:choice-of-beta-for-small-n}
\beta = \frac{1}{n} \log \frac{4e}{\delta}
\end{align}
guarantees $\mathcal N \leq 1 + \delta$.
\end{enumerate}
Now, we bound the coefficient $C_n$: from \cref{eq:direct-MA-formula} we have that for any choice of $r=\sqrt{n \, e^{-\beta}}$, the coefficient $C_n$ is given by
\begin{align}
C_n = \frac{1}{\sqrt{\mathcal N}} \frac{\sqrt{n!}}{n+1} \frac{e^{n/2} e^{e^{-\beta}/2}}{n^{n/2} e^{-\beta n/2}} \leq 4 \frac{(2\pi n)^{1/4}}{n+1} e^{e^{-\beta}/2}\, e^{\beta n/2}.
\end{align}
Now for our two scenarios:
\begin{enumerate}
    \item If $n \geq \frac{e}{4} \log \frac{4e}{\delta}$, we have chosen $\beta=0$ giving
    \begin{align}
    C_n \leq O(\frac{1}{n^{3/4}}),
    \end{align}
    \item If $n< \frac{e}{4} \log \frac{4e}{\delta}$, we would choose $\beta = \frac{1}{n}\log \frac{4e}{\delta}$ (according to \eqref{eq:choice-of-beta-for-small-n}), which guarantees $e^{\beta n/2}\leq O(\frac{1}{\sqrt \delta})$, giving
    \begin{align}
    C_n \leq O(\frac{1}{\sqrt \delta}).
    \end{align}
\end{enumerate}
Hence, we can overall, upper bound $C_n$ by $O(\frac{1}{\sqrt \delta})$.
\end{proof}

\subsubsection{$\mathsf{PP}$ algorithm for exponential energy cubic computation}

We are now ready to state our $\mathsf{PP}$ algorithm.

\begin{proof}[Proof of \cref{thm:XCubedInPP}]
As per \cref{corol:k0} we use the gadget of \cref{fig:cubic-gadget} with exponential squeezing parameter, to get an exponentially accurate simulation. From \cref{corol:Gaussian-rank-of-cubic} we get that our input state (including all non-Gaussian ancillas required for cubic gate teleportation) is a sum of $N = \mathsf{exp}(n)$ many Gaussian states. Hence, we can our state at step $t$ as
\begin{align}
\ket{\psi_t} = \sum_{i=1}^{N_t} c^{(t)}_i \ket{G^{(t)}_i},
\end{align}
where $G_i$ are normalized Gaussian states, with $c_i\in\mathbb C$. We can classically track the evolution of Gaussian states going through Gaussian evolutions (we use \cref{lem:hdyne-gaussian-expect} and \cref{lem:useful-gaussian-overlap}). This is indeed the CV analogue of the Gottesmann-Knill theorem. The way we can do this is by following the update rules of \eqref{eq:gaussian-update}. We note that we can always use the maximum number of $N_T=R=\mathsf{exp}(n)$ many terms, and set $c_i^{(t)}=0$ whenever $i\geq N_t$. We do so for convenience. Below we introduce our proof strategy:

\begin{enumerate}
\item First, we note that the state of system, at any time of the computation is given by
\begin{align}
\ket{\psi_{t}} = \sum_{i=1}^R c_i^{(t)} \ket{G_i^{(t)}},
\end{align}
and showing that $c_i^{(t)}$, the covariance matrix elements of $G_i^{(t)}$, which we denote by $\boldsymbol\Sigma_i^{(t)}$, and the vector of means of $G_i^{(t)}$ which we denote by  $\boldsymbol\mu_i^{(i)}$ satisfy
\begin{align}\label{eq:bounded-norm-growth}
|c_i^{(t)}|, \norm{\boldsymbol\Sigma_{i}^{(t)}}, \norm{\boldsymbol\mu^{(t)}_i}\leq \mathsf{exp}(n),
\end{align}
and therefore, we can store them with $\mathsf{poly}(n)$ bits of precision.

\item We show that each step of updates $c_i^{(t)} \mapsto c_i^{(t+1)}$, $\boldsymbol\Sigma_i^{(t)} \mapsto \boldsymbol\Sigma_i^{(t+1)}$, and $\boldsymbol\mu_i^{(t)} \mapsto \boldsymbol\mu_i^{(t+1)}$ can be performed in $\mathsf{P}$.

\item Finally, we note that the probability of number measurement on mode 1 is 

    $$p = \sum_{n\in A} \braket{\psi_T}{(\hat{n}\otimes I^{\otimes n-1})\mid\psi_T}.$$

    where $\mathbb A\subset \mathbb N$ is the set of acceptable photon number measurements. We note fock state $\ket{k} \approx \sum_{i=0}^{k} \omega_i^{(k)} \ket{\alpha_i}$ for coefficients $\omega_i$ and coherent states $\ket{\alpha_i}$ (c.f. \cref{lem:coherent-decomposition-of-number-states}). Using this fact, we turn our attention to computing $p$.
    Therefore, 

    \begin{align}\label{eq:big-final-overlap}
        p &= \sum_{n\in A}\sum_{i,j}c_i^*c_j\bra{G_i}(\hat{n}\otimes I^{\otimes n-1})\ket{G_j}\\ 
        &= \sum_{n\in A}\sum_{i,j} \sum_{p,q} \overline{c_i}c_j\overline{\omega_p^{(k)}} \omega_q^{(k)} \bra{G_i}(\ketbra{\alpha_p}{\alpha_q}\otimes I^{\otimes n-1})\ket{G_j},
    \end{align}

which is a $\polylog$ depth sum (c.f \cref{lem:Gaussian-complx}) over exponentially many branches. We can hence, use a $\PP$ oracle to determine if $p\geq2/3$.
\end{enumerate}

We now prove each step above.
\begin{enumerate}[label=\arabic*]
    \item[1 and 2.] Our evolution consists of two types of updates: unitary evolutions and post-selections on homodyne measurements. Consider the $t$-th evolution step, and assume its a unitary evolution, described by a symplectic matrix $\mathbf A_t$. The update rule \eqref{eq:gaussian-update} gives
    \begin{align}\label{eq:explicit-update}
    \begin{split}
    &\boldsymbol\Sigma_i^{(t+1)} = \mathbf A_t \boldsymbol\Sigma_i^{(t)} \mathbf A^T_t,\\
    &\boldsymbol\mu_i^{(t+1)}= \mathbf A_t \boldsymbol\mu_i^{(t)} + \mathbf d,\\
    & c_i^{(t+1)} = c_i^{(t)}.
    \end{split}
    \end{align}
    Since the unitary gates in our ciruit have symplectic matrices and $\mathbf d$ vectors of constant size, we get
    \begin{align}
    &\norm{\boldsymbol\Sigma_i^{(t+1)}} \leq a\cdot  \norm{\boldsymbol\Sigma_i^{(t)}},\\
    &\boldsymbol\mu_i^{(t+1)} \leq a \cdot \norm{\boldsymbol\mu_i^{(t)}} + a,\\
    &c_i^{t+1} = c_i^{(t)},
    \end{align}
    for some constant $a>1$.
    Now, let us turn our attention to the other form of evolution we can have in our circuit, which is the adaptive homodyne measurement.
    
    Note that in the injection circuit, we introduce new Gaussian states. Therefore, our state right before the injection gate  becomes
    \begin{align}
    \ket{\psi'_t} = \ket{\psi_t}\otimes \ket{V^{(3)}(\theta);\xi}.
    \end{align}
    Expressing $\ket{V^{(3)}(\theta);\xi} = \sum c'_j \ket{G_j'}$, from \cref{thm:cubic-Gaussian-rank}, we know that each $c'_j=\mathsf{poly}(\xi) = \mathsf{exp}(n)$. So, we get $\sum_i c_i^{(t)} \ket{G_i^{(t)}}\mapsto \sum_{ij} c_i^{(t)} c_j' \ket{G_i^{(t)}, G_j'}$. Assuming $c_i\leq\mathsf{exp}(n)$, we get that the intermediate coeffiente (i.e.~right before the post-selection on $\ket{x=0}$ are at most exponentially large. With abuse of notation, we call these intermediate values by $c_i, \ket{G_i}$ as well.
    Also, without loss of generality, assume have measured the $1$st mode. Note that the update rule for $c_i^{(t)}$ is
    \begin{align}\label{eq:evolution-of-ci-under-meas}
    c_i^{(t)} \mapsto c_i^{(t)}\frac{\norm{(\bra{x=0}\otimes \mathbb I)\ket{G_i^{(t)}}}}{\sqrt{Z_{\psi_t}}},
    \end{align}
    where $Z_t$ is a renormalization factor as in \cref{lem:post-select-on-zero}. Let $C\subseteq[T]$ be the set of time steps that we apply a cubic gate.
    Note that from \cref{lem:post-select-on-zero} we can choose $\lambda =\frac{1}{\mathsf{exp}(n)}$ to have
    \begin{align}
    \abs{\Pi_{t\in C} \frac{1}{\sqrt{Z_{\psi_t}}} - \pi^{|S|/4} \xi^{|S|/2}}\leq \frac{1}{\mathsf{exp}(n)}.
    \end{align}
    Note that this choice of $\lambda$ can still be done by $\xi=\mathsf{exp}(n)$. Hence, we will ignore all rescaling prefactors $Z_t$, and will replace each $\frac{1}{\sqrt{Z_{\psi_t}}}$ by $\pi^{\frac14}\xi^{1/2} = \mathsf{exp}(n)$. For the post-selection overlaps, we use \cref{lem:bounded-overlap}.

    From \eqref{eq:evolution-of-ci-under-meas} and \cref{lem:bounded-overlap} (which implies  $\norm{(\bra{x=0}\otimes \mathbb I)\ket{G_i^{(t)}}} \leq \mathsf{exp}(n)$), we get that
    \begin{align}
    |c_i^{(t+1)}| \leq \mathsf{exp}(n) \cdot |c_i^{(t)}|,
    \end{align}
    with $1-1/\mathsf{exp}(n)$ probability.

For the evolution of covariance matrices under measurement of the first mode in $\ket{x=0}$, consider a Gaussian state $\ket{G}$. In what follows, we prove that if the norm of $\boldsymbol\Sigma$ and $\boldsymbol\mu$ are bounded by an exponential, then the covariance and means are also exponentially bounded as well. Let us write the covariance matrix in a way to put the first mode on the top block:
\begin{align}
\boldsymbol\Sigma = \begin{pmatrix}
\boldsymbol\Sigma_1 & \boldsymbol\sigma^T\\
\boldsymbol\sigma & \boldsymbol\Sigma_{[2..n]}
\end{pmatrix}.
\end{align}
Then, by \cite[Eq. (64)]{weedbrook2012gaussian} we have that 
\begin{align}
\boldsymbol\Sigma_{[2..n]} \mapsto \boldsymbol\Sigma_{[2..n]}'=\boldsymbol\Sigma_{[2..n]} - \frac{1}{\sigma_{11}}\boldsymbol\sigma \Pi \boldsymbol\sigma^T,
\end{align}
where $\sigma_{11} = (\boldsymbol\Sigma_1)_{11}$ is the variance of the position of the first mode, and $\Pi = \mathrm{diag}(1,0,\cdots,0)$ is the projection on to the first element. We know that $\boldsymbol\Sigma_{[2..n]} - \frac{1}{\sigma_{11}}\boldsymbol\sigma \Pi \boldsymbol\sigma^T \geq 0$ (as the covariance matrix is always semi-definite positive). Moreover, as $\frac{1}{\sigma_{11}}\boldsymbol\sigma \Pi \boldsymbol\sigma^T\geq 0$, we get that $\norm{\Sigma_{[2..N]}'}\leq \norm{\Sigma_{[2..N]}}$. Hence, the norm of the covariance matrix cannot increase by post-selecting on a homodyne outcome. 

For the vector of means, we write the vector of means as $\boldsymbol\mu = (\boldsymbol\mu_1, \boldsymbol\mu_{[2..N]})$. We have that \cite{serafini2023quantum}
\begin{align}
\boldsymbol\mu_{[2..N]} \mapsto \boldsymbol\mu_{[2..N]} - \frac{\mu_{1,x}}{\sigma_{11}}\boldsymbol\sigma,
\end{align}
where $\mu_{1,x} = (\boldsymbol\mu_1)_1$. Since $\norm{\boldsymbol\sigma/\sigma_{11}}\leq\mathsf{exp}(n)$ (note that we can upper bound $1/\sigma_{11}\leq (\boldsymbol\Sigma_1)_{22}$ via uncertainty principle) for all steps, we can bound all $\boldsymbol\mu$ by $\mathsf{exp}(n)$.

As a result of the change of the norms during the update rules, we obtain \eqref{eq:bounded-norm-growth}.

    \item[3.] Here, we need to rewrite the coherent decomposition of $\ket{n}$ and finish the proof. Since the parameters for the coherent decomposition of $\ket{k}$ can be approximated with $1/\mathsf{exp}$ precision within $\mathsf{polylog}(n)$ many bits (c.f. \cref{lem:coherent-decomposition-of-number-states}), we can compute \eqref{eq:big-final-overlap} within $1/\mathsf{exp}(n)$ precision. Therefore, for the sum
    \begin{align}
    p=\sum_{n\in A}\sum_{i,j} \sum_{p,q} \overline{c_i}c_j\overline{\omega_p^{(k)}} \omega_q^{(k)} \bra{G_i}(\ketbra{\alpha_p}{\alpha_q}\otimes I^{\otimes n-1})\ket{G_j},
    \end{align}
    we can compute each branch independently and in polynomial time. Hence, we have a circuit of polynomial size evaluating each branch. As a result, $p$ can be approximated by a $\GapP$ function (difference between accepting and rejecting branches of an NTM) and the corresponding decision problem whether $p\geq \frac23$ is in $\PP$.
\end{enumerate}

\end{proof}


\subsection*{Acknowledgements}

We thank Eugene Tang, Cameron Calcluth, and Benjamin Hinrichs for helpful comments on a previous version of this manuscript.
We thank Robert Koenig. Libor Caha, and Lukas Brenner for interesting discussions.
We thank Alessandro Ciani, Qiuting Chen, Benjamin Hinrichs, Jonas Kamminga, Antoine Prouff, Norbert Schuch, and Christine Silberhorn for helpful discussions, and Michael Stefszky for putting up with our many (many) experimental quantum photonics questions.
UC acknowledges inspiring discussions with Varun Upreti, Cameron Calcluth, Giulia Ferrini, Francesco Arzani, Alexander Schuckert, Armando Angrisani, and Zo\"e Holmes, and funding from the European Union’s Horizon Europe Framework Programme (EIC Pathfinder Challenge project Veriqub) under Grant Agreement No.~101114899. 
SG was supported by the DFG under grant numbers 563388236 (Priority Programme ``Quantum Software, Algorithms and Systems – Concepts, Methods and Tools for the Quantum Software Stack'' (SPP 2514)), and 450041824, the BMFTR within the funding program ``Quantum Technologies - from Basic Research to Market'' via project PhoQuant (grant
number 13N16103), and the project ``PhoQC'' from the programme ``Profilbildung 2020'', an initiative of the Ministry of Culture and Science of the State of North Rhine-Westphalia.
DS acknowledges funding from the project PhoQuant (grant
number 13N16103) of the BMFTR, Deutschland. 
We thank Eugene Tang for very insightful comments and suggestions. SM acknowledges funding provided by NSF CCF-2013062. 
HN acknowledges funding from the project “PhoQC” from the programme “Profilbildung 2020”, an initiative of the Ministry of Culture and Science of the State of North Rhine-Westphalia.

\subsection*{Authorship statement}

All authors contributed to the research, calculations, and manuscript preparation. The development of the technical details were led by DR (upper bounds, lower bounds, energy analysis), AM (upper bounds and lower bounds), DS (upper bounds) and HN (noise analysis). UC, SG and SM supervised the project and contributed to outlining proof strategies, explicit calculations, and verifying results.

The authors used a large language model to assist with language editing, drafting and revising text, and brainstorming possible proof strategies and research directions. All mathematical arguments, proofs, calculations, code, and scientific conclusions were independently verified and validated by the authors. All authors take full responsibility for the content of the manuscript. 
\printbibliography

\appendix
\crefalias{section}{appendix}

\input{schwartz.tex}

\input{phasespacesim}


\section{Infinite energy, truncation \& hope for simulation}\label{app:conjecture}

\begin{conjecture}\label{conj:decidable}
    $\CVBQP$ is decidable for any gate set and without energy bounds.
\end{conjecture}

Our result $\PTOWER \subseteq \CVBQP$ already places $\CVBQP$ in the upper echelons of the complexity zoo.
Additionally, we have seen that $\CVBQP$ circuits can produce infinite energy within arbitrarily short time, even with a single gate applied to the vacuum state (\cref{prop:inf_three}).
We further showed that energy ``finiteness'' is undecidable (\cref{thm:energy-undecidable}).

On the other hand, we have shown in \cref{thm:decidable} that $\CVBQP$ \emph{with} bounded energy (including bounds for higher moments) is decidable, even if the energy bound is not known a priori.
The reason why we require the energy and moment bounds is the truncation lemma (\cref{lem:trunc}), which allows us to simulate Hamiltonians via truncation in the Fock basis (\cref{rem:truncation} illustrates the difference between energy bounds and truncation).
Unfortunately, the error scales with the chosen cutoff, which requires that the approximation error is inverse polynomial in the chosen cutoff.
That is precisely what bounded moments $\langle\N^k\rangle$ guarantee.
However, as we will argue next, we do not actually need moment bounds, as long as we can calculate Fock basis entries of the unitaries.
In fact, we show that any $\CVBQP$ has a \emph{uniform truncation threshold} for any accuracy $\epsilon$.
That is, for we for any desired accuracy $\epsilon$, there exists a threshold $E$, such that
\begin{equation}
    \norm{\Pi_{\le E}\ket{\psi(t)}} \ge 1-\epsilon \;\;\forall t\in [0,T],
\end{equation}
where $\ket{\psi(t)}$ is the state at a continuous timestep (i.e.~also ``during'' each gate).
It suffices to consider a single gate since we only have a finite number of gates.

\begin{proposition}\label{prop:uniform-truncation}
    Let $\ket{\psi}\in L^2(\RR^m)$, and $H$ a self-adjoint Hamiltonian with a dense domain.
    Define $\ket{\psi(t)} = e^{-itH}\ket{\psi}$.
    For any $\epsilon>0$, there exists a cutoff $E$, such that $\norm{\Pi_{\le E}\ket{\psi(t)}} \ge 1-\epsilon$ for all $t\in[0,T]$ (arbitrary $T>0$).
\end{proposition}
\begin{proof}
    Define $Q_E(t) = \bra{\psi(t)}\Pi_{>E}\ket{\psi(t)}$.
    $Q_E(t):[0,T]\to[0,1]$ is a well-defined continuous function by Stone's theorem (e.g.~\cite[Theorem 5.1]{teschl2014mathematical}).
    We trivially have $Q_{E+1}(t) \le Q_{E}(t)$ for all $t,E$.
    Therefore the $Q_E$ form a monotonically decreasing sequence that uniformly converges to the constant function $0$ by Dini's theorem, i.e.
    \begin{equation}
        \lim_{E\to\infty}\max_{t\in[0,T]} Q_E(t) = 0.
    \end{equation}
\end{proof}

That means, for any $\CVBQP$ computation, we can find a sufficiently good cutoff in Fock basis, even if the energy becomes infinite.
So in theory all that remains is computing the Fock basis entries of the unitaries, i.e.,
\begin{equation}
    \bra{{\bfi}} e^{-itH} \ket{{\bfj}},
\end{equation}
However, note the infinite series $\sum_{n=0}^\infty \frac{-it\bra{\bfi}H^n\ket{\bfj}}{n!} \quad \text{for } \bfi,\bfj \in \NN_0^m$ is not guaranteed to converge since Fock states are generally not analytic for polynomial Hamiltonians of degree $>2$ (i.e. non-Gaussian).
In principle, our leakage lemma (\cref{lem:leakage}) can guarantee simulatability of short time slices, but these slices could potentially get ever shorter, so that the simulation never finishes.

\subsection{Infinite energy example}\label{rem:truncation}

    Every state in $L^2$ can be arbitrarily well approximated by truncating in the Fock basis.
    For example, the following unit state has infinite ``energy'', although it looks perfectly harmless \cite{cahill1969ordered}:
    \begin{subequations}
    \begin{align}
        \ket{\psi} &= \frac{\sqrt{6}}{\pi} \sum_{n=1}^\infty \frac{1}{n}\ket{n},\\
        \ev{\N}{\psi} &= \frac{6}{\pi^2} \sum_{n=1}^\infty \frac{1}{n} = \infty.
    \end{align}
    \end{subequations}
    Hence, $\ket{\psi}$ is not in the domain of the number operator $\N$, so if we strictly view $\N$ as a function on $L^2$, the object $\N\ket{\psi}$ is not well defined.

\section{Insights about implementing measurement in CV systems}\label{scn:model}
 In DV systems, we define the computational model by starting in a trivial start state, applying a sequence of ``feasible'' operations (defined as sequences of $2$-qubit unitary gates), and measuring in a ``feasible'' measurement basis (the standard basis). While the first two of these have rather natural analogues in the CV models considered so far, the latter is less obvious. Measuring in the Fock basis seems a natural analogue to a standard basis measurement, but in principle it allows superexponential size numbers as outputs, which clearly cannot be efficiently implemented, as a poly-time CV machine cannot even write out the answer\footnote{In principle, one can consider succinct representations of the output, like streaming bits of the photon number read, but usually we like to avoid such succinct representations as they tend to be very context-specific (read: hardware-specific in this case).}. This unbounded nature of the "axiomatic" measurement observable (photon number) is troublesome when applying standard deviation bounds, since it no longer suffices to bound the trace distance of an ideal CV computation versus some non-ideal simulating CV computation (H\"older inequality fails).

    Here we discuss what a ``computationally meaningful'' definition of CV measurement would be. Such a measurement should only allow exponentially many possible outputs\footnote{One possible weakening of this claim, which I assume is more experimentally friendly, is to also allow the CV device to output either one of exponentially many possible outputs, or an ``error'' flag to indicate the measurement has failed.}. There are various ways to try to model this. The most trivial seems to be to simply truncate the number operator, as in practice, this is effectively what one can hope to implement anyway. Another nice idea is to move to a hybrid model, where one couples CV degrees of freedom (e.g.~oscillators) to a set of qubits. Such hybrid oscillator-qubit devices have been studied in various works (see for instance \cite{liuHybridOscillatorQubitQuantum2025}). This leads to at least two possible definitions:

\begin{definition}[Hybrid IO]
    In this model, the qubits are only accessed at the beginning to hold/read the input, and at the end to conduct a standard basis measurement. The latter forces the outcome to be ``computationally well-defined.''
\end{definition}
\noindent Note that Hybrid IO does not restrict the energy during the computation itself, as in \cite{chabaudBosonicQuantumComputational2025}.

The main question is: ``Is the Hybrid IO model equivalent to simply truncating the number operator?'' If a polynomial value bounds the energy of the oscillator, then using insights from \cite{chabaudBosonicQuantumComputational2025}, one can show that oscillators can effectively be discretized. But if oscillators have superpolynomial values of energy, can we decode information encoded in such high-energy modes? In the section \cref{sec:CV-DV}, we study some examples where energy is inversely polynomially related to the precision one needs in the DV measurements. 

We can, furthermore, define a fully hybrid model:
\begin{definition}[Fully hybrid]
    In this model, the qubits are accessed as desired throughout the computation and can be used, e.g.~to control operations on the CV system and vice versa.
\end{definition}
We can ask the question ``Is the Fully Hybrid model more powerful than Hybrid IO?'' By adding CV-DV interactions, one gains higher control over CV degrees of freedom. We don't know whether or not this leads to a computational advantage over the Hybrid IO model. 

\subsection{Coupling between oscillators and spins}
\label{sec:CV-DV}

Inspired by the discussion above, we consider the following hybrid IO model for encoding DV information into the input of a CV system and decoding DV information from its output. We consider three parts (1) a CV Hilbert space $S \cong \mathcal{M}_m$ initially in the vacuum state $\ket{0^m}$, (2) a DV Hilbert space $I \cong \mathcal{Q}_l$ to encode DV information into the system and (3) another DV system $O \cong \mathcal{Q}_n$ to decode DV information from the system. We consider a unitary coupling $E_x$ between $I$ and $S$ to load information from $I$ into $S$. Information will be encoded as a string $\ket x$ for $x \in \{0,1\}^l$ in $I$. We then use a CV circuit $\hat U$ to evolve $S$. We then use a unitary coupling between $S-O$, $\hat{D}$, to load information from CV modes into the DV modes $O$. In the end, we will measure an observable $\hat Q$ from $O$. We choose both $\hat{E}_x$ and $\hat D$ to be composed of gates based on Hamiltonian terms $A \otimes B$ where $A$ impacts CV degrees of freedom ($S$) and $B$ impacts DV systems $I$ or $O$. We furthermore assume $d(A) = 2$. 

As an example, suppose $m = l = n$ and we start with the initial input state $\ket{0^m}_{S} \otimes \ket{x}_{I}$. Consider the operator $\hat E_x = \prod_{j = 1}^m e^{i \hX_j \otimes \ket{1}_{j}\bra{1}}$ to load the DV information onto the CV modes. By applying $\hat E_x$ to $\ket{0^m} \otimes \ket{x}$ we obtain $\ket{x}_{S} \ket{x}_{I}$. We could also, without loss of generality, set $\hat E_x = I$ and start with $\ket{0^m}_{CV}$ and encode the $DV$ information $x$ into $S$ directly by updating $\hat U$. As a result, we will only focus on the decoding process.

\paragraph*{The measurement process:}
Suppose at the end of the computation the quantum state $\ket{\psi} = \sum_{\mathbf{n} \in \mathbb{Z}_{\geq 0}^m} \psi_{\mathbf{n}} \ket{\mathbf{n}} = U \ket{0^m}$. We then apply the Hamiltonian $\hH (\theta_1, \ldots, \theta_L) = \sum_{j = 1}^{L} \theta_j \hat \Lambda_j \otimes \hat A_j$, where $H_j$ is a CV Hamiltonian with degree $\leq 2$ affecting system $S$ and $A_j$ is a DV Hamiltonian affecting system $O$. Ultimately, we measure a DV operator $\hQ$ from $O$. Without loss of generality assume $\|\hA_1\|_{op}, \ldots, \|\hA_L\|_{op}, \|\hQ\|_{op} \leq 1$. Suppose we start the DV system in $\ket{0^n}_O$. Let 
$$
q^{(N)}_{\psi}(\theta_1, \ldots, \theta_L) = \bra{\psi}_S\bra{0^n}_O e^{i \hH(\theta_1, \ldots, \theta_L)} (I_S \otimes Q_O) e^{-i \hH(\theta_1, \ldots, \theta_L)} \ket{\psi}_S\ket{0^n}_O.
$$

As an example, suppose $L = m$ and $\hat \Lambda_j = \hN_j$. Then
\asp{
q^{(N)}_\psi(\theta_1, \ldots, \theta_L) 
&= \sum_{\mathbf{n} \in \bbZ_{\geq 0}^m}|\psi_{\mathbf{n}}|^2 \bra{0^n}e^{i \hA_{\mathbf{n}}(\theta_1, \ldots, \theta_L)} \hQ e^{-i \hA_{\mathbf{n}}(\theta_1, \ldots, \theta_L)}\ket{0^n}\\
&=: \sum_{\mathbf{n} \in \bbZ_{\geq 0}^m}|\psi_{\mathbf{n}}|^2 \bra{0^n} \hQ_{\mathbf{n}} (\theta_1, \ldots, \theta_L)\ket{0^n},
}
where $\hA_{\mathbf{n}}(\theta_1, \ldots, \theta_L) = \sum_{j = 1}^m \theta_j n_j \hA_j$ and
$$
\hQ_{\mathbf{n}} (\theta_1, \ldots, \theta_L) := e^{i \hA_{\mathbf{n}}(\theta_1, \ldots, \theta_L)} \hQ e^{-i \hA_{\mathbf{n}}(\theta_1, \ldots, \theta_L)}.
$$
Note that the other choice we could have considered was $\hat \Lambda_j = \hX_j$ in which case we would have obtained 
\asp{
q^{(X)}_\psi(\theta_1, \ldots, \theta_L) 
&= \int_{x \in \bbR^n}|\psi (\mathbf {x})|^2 \bra{0^n}e^{i \hA_{\mathbf{x}}(\theta_1, \ldots, \theta_L)} \hQ e^{-i \hA_{\mathbf{x}}(\theta_1, \ldots, \theta_L)}
\ket{0^n}\\
&=:\int_{x \in \bbR^n}|\psi (\mathbf {x})|^2 \bra{0^n} \hQ_{\mathbf{x}}(\theta_1, \ldots, \theta_L)
\ket{0^n},
}
where $\hA_{\mathbf{x}}(\theta_1, \ldots, \theta_L) = \sum_{j = 1}^m \theta_j x_j \hA_j$ and
$$ 
\hQ_{\mathbf{x}} (\theta_1, \ldots, \theta_L) := e^{i \hA_{\mathbf{x}}(\theta_1, \ldots, \theta_L)} \hQ e^{-i \hA_{\mathbf{x}}(\theta_1, \ldots, \theta_L)}.
$$

\paragraph*{Analysis of sensitivity:} Next, we consider small perturbations in the angular coordinate 
\[
    (\theta_1, \ldots, \theta_L)\mapsto (\theta_1, \ldots, \theta_L) + (\delta\theta_1, \ldots, \delta\theta_L),
\]
and ask how small these perturbations should be to get stability in the output measurement value. We evaluate
\asp{
\frac{\partial}{\partial \theta_j }q^{(N)}_\psi(\theta_1, \ldots, \theta_L) 
&= i \sum_{\mathbf{n} \in \bbZ_{\geq 0}^m}|\psi_{\mathbf{n}}|^2 n_j \bra{0^n} [\hat A_j , \hQ_{\mathbf{n}} (\theta_1, \ldots, \theta_L)]
\ket{0^n},\\
\frac{\partial}{\partial \theta_j }q^{(X)}_\psi(\theta_1, \ldots, \theta_L) 
&= i \int_{\mathbf{x} \in \bbR^m}|\psi(\mathbf{x})|^2 x_j \bra{0^n} [\hat A_j , \hQ_{\mathbf{n}} (\theta_1, \ldots, \theta_L)]
\ket{0^n}.\\
\label{eq:q-derivative}
}
Therefore for any $\boldsymbol{\theta}$, $|\frac{\partial}{\partial \theta_j }q^{(N)}_\psi(\boldsymbol{\theta})| \leq O(\langle \hN_j\rangle)$ and $|\frac{\partial}{\partial \theta_j }q^{(X)}_\psi(\boldsymbol{\theta})| \leq O(|\langle X_j\rangle|) = O(\sqrt{\langle \hN_j\rangle})$.

Next suppose $\bra{\psi} \hN \ket{\psi} = E$. Then, using Markov's quantum inequality discussed in \cite{chabaudBosonicQuantumComputational2025} (aka gentle measurement lemma) then $\|\ket{\psi_{E/\delta^2}} - \ket{\psi}\| = O (\delta)$, where $\ket{\psi_N}$ is the trunction of $\ket{\psi}$ up to $N$ in the Fock bassi. Therefore if we define $|q^{(N)}_{\psi_{E/\delta^2}} (\boldsymbol{\theta}) - q^{(N)}_{\psi} (\boldsymbol{\theta}) | = O(\delta)$. Since $q^{(N)}_{\psi_{E/\delta^2}} (\boldsymbol{\theta})$ is a sum of a finite number of terms, it is real analytic. We can then use the multi-variable mean-value theorem: Let $f (t) = q^{(N)}_{\psi_{E/\delta^2}} (\boldsymbol{\theta} + t \boldsymbol{\theta'})$ for some set of angles $\boldsymbol{\theta'}$. Therefore:
$|f(\delta) - f(0)| \leq \delta \sup_{t\in [0,\delta]} |f'(t)|$. We now evaluate $f'(t) = \sum_j \theta'_j\frac{\partial}{\partial \theta_j} q^{(N)}_{\psi_{E/\delta^2}} (\boldsymbol{\theta})$. Using the evaluation of derivatives in \cref{eq:q-derivative} we can conclude that under perturbation $\delta  \boldsymbol{\theta'_1}$,
\asp{
|q^{(N)}_{\psi_{E/\delta^2}} (\boldsymbol{\theta} + \delta \boldsymbol{\theta}_1) - q^{(N)}_{\psi_{E/\delta^2}} (\boldsymbol{\theta})| \leq \delta \cdot E. 
}
As a result, in order for the DV measurement to reveal a stable outcome, one needs to choose $\delta \ll 1/E$. Then, the precision one needs to obtain meaningful information from this DV measurement needs to scale inversely with the energy of the CV state.

\end{document}

%% file: macros.tex
\usepackage{quantikz}
\usepackage{tikz}
\usepackage{complexity}

\usepackage{xcolor}

\newcommand{\asp}[1]{%
  \begin{align}
    \begin{split}
      #1
    \end{split}
  \end{align}
}

\newcommand{\hA}{\hat{A}}
\newcommand{\hB}{\hat{B}}
\newcommand{\hC}{\hat{C}}
\newcommand{\hD}{\hat{D}}
\newcommand{\hDelta}{\hat{\Delta}}

\newcommand{\hF}{\hat{F}}
\newcommand{\hG}{\hat{G}}
\newcommand{\hH}{\hat{H}}
\newcommand{\hI}{\hat{I}}

\newcommand{\hK}{\hat{K}}
\newcommand{\hL}{\hat{L}}
\newcommand{\hM}{\hat{M}}
\newcommand{\hN}{\hat{N}}
\newcommand{\hO}{\hat{O}}
\newcommand{\hP}{\hat{P}}
\newcommand{\hQ}{\hat{Q}}
\newcommand{\hR}{\hat{R}}
\newcommand{\hS}{\hat{S}}
\newcommand{\hT}{\hat{T}}
\newcommand{\hU}{\hat{U}}
\newcommand{\hV}{\hat{V}}

\newcommand{\hX}{\hat{X}}
\newcommand{\hY}{\hat{Y}}

\newclass{\CVBQP}{CVBQP}
\newcommand{\CVBQPpoly}{\CVBQP_{\textup{poly}}}
\newcommand{\CVBQPexp}{\CVBQP_{\textup{exp}}}

\newcommand{\BQPspoly}{\BQP/_{\textup{poly}}}

\newclass{\PrecQMA}{PrecQMA}
\newclass{\TOWER}{TOWER}
\newclass{\PTOWER}{PTOWER}
\newclass{\AzPP}{A_0PP}

\newcommand{\ha}{\hat{a}}
\renewcommand{\E}{\mathbb{E}}

\newcommand{\ad}{\mathrm{ad}}

\newcommand{\bbZ}{\mathbb{Z}}
\newcommand{\bbR}{\mathbb{R}}
\newcommand{\bbC}{\mathbb{C}}

\newcommand{\SUM}{\hat{\text{SUM}}}

\newcommand\calH{\mathcal{H}}

\newcommand\calS{\mathcal{S}}
\newcommand\calT{\mathcal{T}}
\newcommand\calG{\mathcal{G}}

\newcommand\calD{\mathcal{D}}

\newcommand\calC{\mathcal{C}}

\newcommand\Span{\mathrm{span}}
\newcommand\NN{\mathbb{N}}
\newcommand\CC{\mathbb{C}}
\newcommand\RR{\mathbb{R}}
\newcommand\ZZ{\mathbb{Z}}

\newcommand\uuarr{\mathbin{\upuparrows}}
\newcommand{\bfalpha}{\vb*{\alpha}}
\newcommand{\bfbeta}{\vb*{\beta}}
\newcommand{\bfk}{\mathbf{k}}

\newcommand{\bfpartial}{\vb*{\partial}}
\newcommand{\bfl}{\mathbf{l}}
\newcommand{\bfj}{\mathbf{j}}
\newcommand{\bfi}{\mathbf{i}}
\newcommand{\bfx}{\vb*{x}}
\newcommand{\bfn}{\mathbf{n}}
\newcommand{\bfm}{\mathbf{m}}

\newcommand{\Ayes}{A_{\mathrm{yes}}}
\newcommand{\Ano}{A_{\mathrm{no}}}

\newcommand\Hvar{H^{\mathrm{var}}}
\newcommand\Hweight{H^{\mathrm{wgt}}}
\newcommand\HHlog{\calH^{\mathrm{log}}}
\newcommand\HHvar{\calH^{\mathrm{var}}}
\newcommand\wtrho{\widetilde{\rho}}
\newcommand\wtpsi{\widetilde{\psi}}
\newcommand\Dom{\mathrm{Dom}}
\newcommand{\ctl}{\mathrm{ctl}}
\newcommand{\tgt}{\mathrm{tgt}}
\newcommand\X{{\hat X}}
\renewcommand\P{{\hat P}}
\newcommand\N{{\hat N}}
\renewcommand\a{{\hat a}}
\newcommand\Var{\mathrm{Var}}
\newcommand\diag{\mathrm{diag}}
\newcommand\n{{\bar n}}
\newcommand\out{{\mathrm{out}}}

\renewcommand{\poly}{{\mathrm{poly}}}

\newcommand{\calP}{{\mathcal{P}}}
\newcommand{\Dfin}{{\calD_{\mathrm{fin}}}}

\newcommand\Poisson{\mathrm{Poisson}}
\newcommand\nmax{n_{\mathrm{max}}}

%% file: adiabatic2.tex
\begin{theorem}\label{thm:CVBQP-NP}
  There exists a finite gateset $\calG$, such that $\NP \subseteq \CVBQP[\calG]$.
  The CV computation requires exponential energy and a constant number of modes.
\end{theorem}

\begin{theorem}\label{thm:CVBQP-ELEMENTARY}
  There exists a finite gateset $\calG$, such that for all $k\in \NN\colon\NTIME(\exp^{(k)}(n))\subseteq \CVBQP[\calG]$ with $O(k)$ modes.
\end{theorem}

\begin{theorem}\label{thm:CVBQP-TOWER}
  There exists a finite gateset $\calG$, such that $\PTOWER \subseteq \CVBQP[\calG]$.
\end{theorem}

The key point of \cref{thm:CVBQP-NP} is that we can use energy as a resource to solve $\NP$-hard problems in polynomial time with a bosonic circuit with a constant number of modes.

Relaxing the energy bounds allows us to solve vastly harder problems.
Any language \\ $L\in \ELEMENTARY$ can be decided with a constant number of modes (\cref{thm:CVBQP-ELEMENTARY}).
Allowing a polynomial number of modes even gives $\CVBQP$ the power to simulate Turing machines whose runtime is a power tower of polynomial height in the input size (\cref{thm:CVBQP-TOWER}).
These results come with a crucial caveat: The intermediate states during the computation technically have infinite energy (see \cref{sscn:infenergydorian}).
The states remain well defined in the Hilbert space $L^2$, but the expectation value of the photon number observable becomes infinite.
The infinite energy really just comes from a ``fat tail''; the intermediate states will be exponentially close in $L^2$-norm to states that have at most $\exp^{O(r)}(n)$ photons, since each gate roughly exponentiates energy,
where $r$ is the number of ``photon number controlled squeezers'' that have been applied.
Another way to look at it, is that if we were to measure certain ancilla modes during the computation, the energy of the resulting mixed state would be bounded by $\exp^{O(r)}(n)$ with high probability (see \cref{lem:tower}).

\subsubsection{Hardness of Diophantine equations}

We prove \cref{thm:CVBQP-NP,thm:CVBQP-ELEMENTARY,thm:CVBQP-TOWER} by developing a $\CVBQP$ algorithm for solving Diophantine equations.
By the celebrated Matiyasevich–Robinson–Davis–Putnam or MRDP theorem \cite{Mat93}, Diophantine equations (or Hilbert's tenth problem) are undecidable.

Adleman and Manders give upper bounds on the solutions to the Diophantine equation encoding the computation of a nondeterministic Turing machine in terms of its runtime.
This allows to define an $\NTIME(T(n))$-hard promise problem based on a fixed Diophantine equation:

\begin{theorem}\label{thm:Diophantine}
  There exists an integer polynomial $F(x_1,\dots,x_k)$, such that the promise problem $A^{(F,E)}$ is $\NTIME(T(n))$-hard for $E(n) = 2^{2^{2^{10 T^2(n)}}}$:
  \begin{subequations}
    \begin{align}
      \Ayes^{(F,E)} &= \{x\mid \exists x_2,\dots,x_k \in \{0,\dots,E(\lceil\log x\rceil)\}\colon F(x,x_2,\dots,x_k)=0\},\\
      \Ano^{(F,E)} &= \{x \mid \forall x_2,\dots,x_k\in \NN_{0}\colon F(x,x_2,\dots,x_k)\ne 0\},
    \end{align}
  \end{subequations}
\end{theorem}
\begin{proof}
  Let $M$ be a universal NTM that instead of rejecting enters an infinite loop. We say $M$ accepts an input $x$ in time $T$, if there is at least one branch, on which $M$ accepts within time $T$. Other branches do not need to halt.
  By \cite[Lemma VI]{AM76}, we can construct an integer polynomial $F(x_1,\dots,x_k)$, such that $F_x(x_2,\dots,x_k)\coloneq F(x,x_2,\dots,x_k)$ has an integer root if and only if $M$ has an accepting branch on input $x$.
  Further, if $M$ halts within time $t$ (i.e. in the YES case), there exists a solution bounded by $2^{2^{2^{10t^2}}}$.

  Now let $L\in \NTIME(T(n))$ be accepted by an NTM $M_L$.
  The reduction from $L$ to $A^{(F,E(n))}$ (with $E(n)$ as in the theorem statement) simply maps the input $x$ to $\langle M_L\rangle x$ (i.e. the Gödel number of $M_L$ concatenated with the input).
\end{proof}

Note that we do not get completeness since the upper bound on the solution is \emph{triply exponential} in the runtime.
For $\NP$, Manders and Adleman show completeness of a much simpler Diophantine system:

\begin{theorem}[\cite{MA78}]\label{thm:MA78}
  The language $L_{\mathrm{MA}} = \{(\alpha,\beta,\gamma)\in\NN_0^3 \mid \exists x_1,x_2\in\NN_0\colon \alpha x_1^2+\beta x_2 - \gamma = 0\}$ is $\NP$-complete.\footnote{An elegant reduction from the Subset Sum problem appears in \cite[p. 150]{Nature_of_Computation}.}
\end{theorem}
\noindent
Containment in $\NP$ is trivial as it suffices to test $x_1,x_2\in \{0,\dots,\gamma\}$.

\subsubsection{Challenges for adiabatic computation}

Writing a Diophantine equation as a CV Hamiltonian is very simple \cite{Kieu2003}:
The solution to a Diophantine equation $F(x_1,\dots,x_k) = 0$ is exactly the $0$-energy ground state of the $k$-mode Hamiltonian $F(\N_1,\dots,\N_k)$.
To prove \cref{thm:CVBQP-NP}, we will use adiabatic computing to (approximately) prepare the ground state of $F(\N_1,\dots,\N_k)$ inside the $E$-photon subspace.

There are several challenges with simulating an adiabatic algorithm inside $\CVBQP[\calG]$ with a fixed gateset:
\begin{enumerate}
  \item \emph{Simulating a time-dependent Hamiltonian with time-independent gates.}
  \begin{itemize}
    \item Use a CV mode as a continuous time register, encoding the time in position basis. Then $H = (1-\X)\otimes H_0 + \X \otimes H_1 + \P\otimes I$ acts like a time-dependent Hamiltonian $A(t) = (1-t)H_0 + tH_1$, where the momentum operator $\P$ increases the time register, and the $\X$ operator ``reads'' the time.
    \item Concretely, simulate the unphysical position eigenstates $\ket{x=t}$ (i.e. not in $L^2$) with displaced squeezed states.
    The position variance of these states is exponentially small in the squeezing parameter, which suffices for our clock construction.
  \end{itemize}
  \item \emph{Bounding the spectral gap.}
  \begin{itemize}
    \item \cite{Kieu2003} already suggests an adiabatic evolution that interpolates between the projector onto a coherent state with the diagonal Hamiltonian $F(\N_1,\dots,\N_k)$.
    The difficulty with that approach is that there appears to be no explicit upper bound on the required evolution time, to the best of our knowledge.
    Since the Hamiltonian is infinite dimensional, even if level crossings are avoided, it is a priori not clear whether finite runtime suffices.
    \item Therefore, we design our Hamiltonian so that the time-dependent part $A(t)$ preserves photon number, i.e., a finite dimensional subspace.
    This circumvents unbounded operator issues and opens up the possibility to use standard tools from linear algebra for gap analysis
    \item We can construct a Hamiltonian $H_0$ that looks like the Laplacian of a path graph in a certain subspace.
    Hence, its ground state is a uniform superposition of Fock states.
    The next logical step would be to interpolate $H_0$ adiabatically with the diagonal Hamiltonian $H_1 = F(\N_1,\dots,\N_k)$.
    It can be shown that the gap is nonzero at all time-steps.
    However, without a strict analysis, the gap can be exponentially small in the dimension, i.e., for an $\NP$ problem we would have doubly-exponentially small spectral gap.\footnote{There are examples of adiabatic computation with Laplacian and diagonal potential where the gap becomes exponentially small in the dimension \cite{JJ15}. Generally, an ``exponentially small spectral gap'' is considered with respect to the input size (e.g. \cite{Altshuler_2010}), which would even suffice for our purposes, since we do not aim to solve $\NP$-hard problems efficiently. Proving that the gap cannot get any worse than inverse exponential in the \emph{dimension} is non-trivial, but can be done by noting that the gap (and its inverse) can be expressed in the first-order theory of the reals, which has a bounded solution by \cite[Proposition 1.3]{Renegar92}. We include a proof in \cref{lem:gap} since it may be of independent interest.}
    In the case of a path graph and unit weights, a polynomial gap in the dimension can be shown \cite{ADKLLR07}, but unit weights do not seem possible in our model, since we can only use Hamiltonians constructed from $\a_i^\dagger,\a_i$.
    \item Our workaround for the above issues is to use a nonlinear adiabatic Hamiltonian. The idea is to construct a ``weighted history state'' similar to \cite{Bausch_2018}.
    Note that the projector $(w_1\ket{0} - w_0\ket{1})(w_1\bra{0}-w_1\bra{1})$ has the ground state $w_0\ket{0} + w_1\ket{1}$.
    Setting $w_1(t) = 1 + t$ and $w_0(t) = 1 + tF(\N_1,\dots,\N_k)$ will start to weight the ground state towards $\ket{1}$ over time if $F(\cdot) =0$ and to $\ket{0}$ otherwise.
    The projector is now quadratic in $t$.
  \end{itemize}
  \item \emph{Energy generation.}
  \begin{itemize}
    \item Solving problems beyond $\NP$ will require a gigantic search space and extremely long evolution time.
    By using Hamiltonians of the form $\N H$, we can effectively use high energy states to fast-forward Hamiltonian solution.
    \item Using this trick, we can construct a ``photon number controlled squeezing operator'' whose repeated application \emph{exponentiates} the energy.
    With this, we can create a state with photon number around $2\uuarr n$ with just $O(n)$ gates.
  \end{itemize}
  \item \emph{Output measurement.}
  \begin{itemize}
    \item $\CVBQP$ requires polynomially bounded photon number in the output (with high probability). In general, the solutions to the Diophantine equation can be exponentially large.
    We use the Diophantine equation as a \emph{detuning parameter} on a squeezing operator, which will act like a regular squeezing operator if $F(\cdot)=0$, but if $F(\cdot)\ge 1$, the output has bounded squeezing for any evolution time.
  \end{itemize}
\end{enumerate}

\medskip
\noindent
\textbf{Outline} (simplified; full argument in \cref{sec:complete-proof}):
\begin{enumerate}
  \item Prepare the initial state $\ket{\Psi_0} = \bigotimes_{j=1}^k\ket{0,0,E_j,E_j}$, using $4$ modes for each variable (\cref{sec:high-energy}).
  Since we use a photon-number preserving Hamiltonian, the $E_j$ bound the search space for solutions to the Diophantine equations.
  A logical $x\in\NN_0$ would be represented as $\ket{x,x,E_j-x,E_j-x}$.
  The reason for ``doubling up'' is to get rid of $\sqrt{\cdot}$ factors generated by creation and annihilation operators.
  \item Use an adiabatic evolution (\cref{sec:adiabatic}) to prepare the superposition \[\ket{\Psi_1} \propto \sum_{\bfx: F(\bfx)=0} \bigotimes_{j=0}^k\ket{x_j,x_j,E_j-x_j,E_j-x_j}\] (Hamiltonian defined in \cref{sec:hamiltonian-def}).
  \item Check whether a solution was prepared successfully, while keeping the photon number in the output mode bounded with high probability (\cref{sec:extract-output}).
\end{enumerate}

\subsubsection{Hamiltonian: Concept}\label{sec:Hamiltonian-concept}

First, we show how the Hamiltonian solving a univariate Diophantine equation looks like.
The multivariate generalization will be straightforward.
We use $6$ modes to encode the state.
Fix a number $n\in \NN$. Then our ``search space'' for the Diophantine equation will be $\{0,\dots,n\}$.
We encode a logical state $\ket{j,0}_L$ as $\ket{j,j,n-j,n-j,1,0}$ and logical state $\ket{j,1}_L$ ass $\ket{j,j,n-j,n-j,0,1}$, so that all logical states have the same photon number.
The reason for using $4$ modes to encode a single integer is so that creation and annihilation operators do not introduce square roots. Now we have
\begin{equation}
  \begin{aligned}
\a_1^\dagger \a_2^\dagger \a_3\a_4\ket{j,j,n-j,n-j} &= (j+1)(n-j)\ket{j+1,j+1,n-j-1,n-j-1},\\
\ket{j}_L&\mapsto (j+1)(n-j)\ket{j+1}_L.
  \end{aligned}
\end{equation}
This is advantageous because we can build a path graph Laplacian using the above Hamiltonian to make edges, and photon number operators to make the diagonal.
Our Hamiltonians will preserve photon number in modes $1..4$ and $5,6$.
The additional bit encoded in modes $5,6$ allows us to attach another vertex to each vertex in the graph.
In the univariate case, the graph will then be a path on vertices $\ket{j,0}_L$, to which $\ket{j,1}_L$ vertices are attached (see \cref{fig:walk}).
We will impose ``weighted transitions'' between $\ket{j,0}_L$ and $\ket{j,1}_L$.
The other transitions are unweighted, so that all $\ket{j,0}_L$ vertices have the same weight.
This avoids long chains of decreasing relative weights that can lead to an exponentially small spectral gap.

The basic path Laplacian encoding the search space of a single variable is given by
\begin{equation}\label{eq:hvar}
  \Hvar = \N_1(\N_3+I) + (\N_1 + I)\N_3 - \a_1^\dagger\a_2^\dagger \a_3\a_4 - \a_1\a_2\a_3^\dagger\a_4^\dagger.
\end{equation}
Restricted to the logical space $\Span\{\ket{j}_L \coloneq \ket{j,j,n-j,n-j} : 0\le j \le n\}$, we get the Laplacian of an integer-weighted path graph:
\begin{equation}
\Hvar\upharpoonright_{\HHvar_5} = \begin{pmatrix}
5 & -5 & 0 & 0 & 0 & 0 \\
-5 & 13 & -8 & 0 & 0 & 0 \\
0 & -8 & 17 & -9 & 0 & 0 \\
0 & 0 & -9 & 17 & -8 & 0 \\
0 & 0 & 0 & -8 & 13 & -5 \\
0 & 0 & 0 & 0 & -5 & 5
\end{pmatrix}.
\end{equation}
To build the full Hamiltonian, we now add the additional weighted transitions of the form
\begin{equation}
  \begin{aligned}
  \Hweight(t) &=\sum_{j} (W_1\ket{j,0}_L - W_0\ket{j,1})(W_1\bra{j,0}_L - W_0\bra{j,1})\\
   &= \N_5 W_1^2 + \N_6 W_0^2 - W_0W_1\bigl(\a_5\a_6^\dagger - \a_5^\dagger \a_6\bigr)\\
  W_0(t) &= I + tF(\N_1)\\
  W_1(t) &= I + t.
  \end{aligned}
\end{equation}
The time $t$ will later be replaced by a position operator $\X$ in the clock register.
On the span of $\ket{j,0}_L$ and $\ket{j,1}_L$, $\Hweight$ looks like
\begin{equation}\label{eq:Hweightblock}
  \begin{pmatrix}
    W_1^2 & -W_0W_1\\
    -W_0W_1 & W_0^2
  \end{pmatrix},
\end{equation}
with ground state 
\begin{equation}
  W_0\ket{j,0} + W_1\ket{j,1}  \propto \ket{j,0} + \frac{1+t}{1+tF(\N_1)}\ket{j,1}.
\end{equation}
If $F(\N_1)=0$, then the ground state will approach $\ket{j,1}$ as the time $t$ increases.
The combined adiabatic Hamiltonian is
\begin{equation}
  \begin{aligned}
  A(t) &= \Hvar\N_5 + \Hweight(t) \\
  &=\bigl(\N_1(\N_3+I) + (\N_1 + I)\N_3 - \a_1^\dagger\a_2^\dagger \a_3\a_4 - \a_1\a_2\a_3^\dagger\a_4^\dagger\bigr)\N_5\\
  &+ \N_5 W_1^2 + \N_6 W_0^2 - W_0W_1\bigl(\a_5\a_6^\dagger - \a_5^\dagger \a_6\bigr),
  \end{aligned}
\end{equation}
where multiplying $\Hvar$ by $\N_5$ makes it so only the $\ket{j,0}_L$ vertices are connected in a path, the $\ket{j,1}_L$ vertices each only have a single edge to $\ket{j,0}_L$.
See \cref{fig:walk} for a graphical depiction. 
Let $\HHlog_n$ be the logical subspace
\begin{equation}
  \HHlog_n = \Span\bigl\{ \ket{j,b}_L \coloneq \ket{j,j,n-j,n-j,1-b,b}\bigm| j\in\{0,\dots,n\}, b\in\{0,1\}\bigr\}.
\end{equation}
\begin{figure}[t]
  \begin{center}
  \includegraphics[]{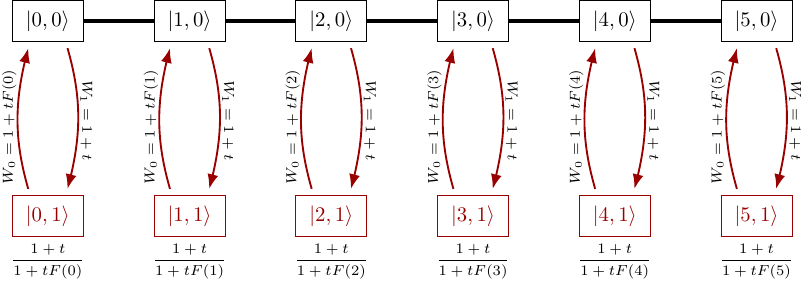}
  \end{center}
  \caption{The adiabatic Hamiltonian $A(t)$ in a logical subspace. Note that the weights refer to weights in the weighted history state picture \cite{Bausch_2018} ($A(t)$ is not a Laplacian). The weights below the red vertices refer the relative weight of those states compared to the base chain in the ground state of $A(t)$.}
  \label{fig:walk}
\end{figure}
The ground state $A(t){\upharpoonright}_{\HHlog_n}$ is unique (\cref{lem:unique-groundstate}) and given by
\begin{equation}
  \ket{\psi(t)} \propto \sum_{j=0}^n \left(\ket{j,0}_L + \frac{1+t}{1+tF(j)}\ket{j,1}_L  \right).
\end{equation}
Let $Z = F^{-1}(0) \cap \{0,\dots,n\}$ and assume $Z\ne \emptyset$.
Define the unnormalized state
\begin{equation}\label{eq:psitilde}
  \ket{\psi(t)}\propto\ket*{\wtpsi(t)} = \underbrace{|Z|^{-1/2}\sum_{j\in Z}\ket{j,1}}_{\eqcolon \ket{Z}} + \underbrace{\sum_{j=0}^n\frac{1}{1+tF(j)}\ket{j,1} + \sum_{j}\frac{1}{1+t} \ket{j,0}}_{\eqcolon \ket{\eta(t)}}.
\end{equation}
Then, $1\le\norm*{\ket*{\wtpsi(t)}} \le 1 + \norm{\eta(t)} \le 1+\sqrt{2n+1}/(1+t)$.
Thus, $\ket{\psi(t)} - \ket*{\wtpsi(t)} \le \sqrt{2n+1}/(1+t)$.
Therefore,
\begin{equation}\label{eq:psit-Z}
  \norm{\ket{\psi(t)} - \ket{Z}} \le \norm{\ket{\psi(t)} - \ket*{\wtpsi(t)} } + \norm{\ket*{\wtpsi(t)} - \ket{Z}} \le 2\norm{\ket{\eta(t)}} \le \frac{2\sqrt{2n+1}}{1+t} = O\!\left(\frac{\sqrt{n}}{1+t}\right).
\end{equation}
Hence, for large enough $t$, $\ket{\psi(t)}$ is approximately the superposition over all integer roots of $F$.
At $t=0$, $\ket{\psi(0)}$ is the uniform superposition.
Since we do not know how to prepare the uniform superposition, we prepare the vacuum state $\ket{0,1}_L$.
We also modify $W_0$ in the next step, so that at $t=0$ the ground state will be close to $\ket{0,1}_L$.
To prepare the Fock state $\ket{1}$ exactly, we use the methods developed in \cite{arzani2025can}.
Since they do not explicitly discuss self-adjointness, we remark here that their operators can be made self-adjoint by adding a high-degree polynomial in $\N$, which gives self-adjointness by Kato--Rellich.

\subsubsection{Hamiltonian: Formal definition and spectral gap}\label{sec:hamiltonian-def}

Going from univariate to $k$-variate Diophantine equations is straightforward.
We really only need to add multiple copies of $\Hvar$ \eqref{eq:hvar}. Define
\begin{subequations}\label{eq:adiabatic-hamiltonian}
\begin{align}
  \Hvar_i & = \N_{j}(\N_{j+2}+I) + (\N_{j} + I)\N_{j+2} - \a_j^\dagger\a_{j+1}^\dagger \a_{j+2}\a_{j+3} - \a_j\a_{j+1}\a_{j+2}^\dagger\a_{j+3}^\dagger\ \ (j\coloneq 4i-3)\\
  \Hvar &= \sum_{i=1}^k \Hvar_i\\
  \Hweight(t) &= \N_{4k+1} W_1^2 + \N_{4k+2} W_0^2 -\bigl(\a_{4k+1}\a_{4k+2}^\dagger - \a_{4k+1}^\dagger \a_{4k+2}\bigr) W_0W_1\label{eq:Hweightt}\\
  W_0(t) &= I + \bigl((1-t)F_0(\N_1,N_5,\dots,\N_{4k-3}) + tF_1(\N_1,\N_5,\dots,\N_{4k-3})\bigr)\hDelta\label{eq:W0}\\
  \hDelta &= \prod_{j=1}^k (\N_j+1)\\
  W_1(t) &= I + t\\
  A(t) &= \Hvar\N_{4k+1} + \Hweight(t).
\end{align}
\end{subequations}
For $k=1$ we get the univariate case presented previously.
The main difference here is that we added a second polynomial to $W_0(t)$. In practice $F_1$ will be the Diophantine equation we wish to solve, and $F_0(x_1,\dots,x_k) = x_1+\dotsm+x_k$ (i.e.~a Diophantine equation with $F_0(0,\dots,0)=0$).
Additionally, we multiply by $\Delta$, so that the entire adiabatic evolution can be performed in time $1$.

Let $\bfn \in \NN^k$. Define the logical space
\begin{subequations}
\begin{align}
  \HHlog_{\bfn} &= \Span\bigl\{\ket{\bfj,b}_L\bigm| \bfj\in S_{\bfn},b\in\{0,1\}\bigr\}\\
  \ket{\bfj,b}_L &= \left(\bigotimes_{i=1}^k\ket{j_i,j_i,n_i-j_i,n_i-j_i}\right)\otimes\ket{1-b,b}\\
  S_\bfn &= \{0,\dots,n_1\}\times\dotsm\times\{0,\dots,n_k\},
\end{align}
\end{subequations}
where $S_\bfn$ is our ``search space'' for solutions to the Diophantine equation $F(x_1,\dots,x_k)$.
Each individual $\Hvar_i$ implements the Laplacian of a path graph along the $i$-th dimension.

\begin{lem}\label{lem:Hvar1}
  Let $H = \Hvar_1{\upharpoonright}_{\HHvar_n}$, where $\HHvar_n = \Span\{\ket{j}_L\coloneq \ket{j,j,n-j,n-j} \mid 0 \le j \le n\}$.
  Then $H$ is the Laplacian of an integer-weighted path graph.
\end{lem}
\begin{proof}
  Consider the matrix elements $H_{j,k} = \bra{j,j,n-j,n-j}H\ket{k,k,n-k,n-k}$.
  Clearly, $H_{j,k} = 0$ if $|j-k|\ge 2$, and
  \begin{align}
    H_{j,j} &= j(n-j+1) + (j+1)(n-j)\\
    H_{j,j+1} &= H_{j+1,j} = - \bra{j,j,n-j,n-j}\a_1\a_2\a_3^\dagger\a_4^\dagger\ket{j+1,j+1,n-j-1,n-j-1}\\
    &=-(j+1)(n-j)\nonumber\\
    H_{j-1,j} &= -j(n-(j-1)) = -j(n-j+1),
  \end{align}
  so $H$ is tri-diagonal.
  We now verify that the rows sum to zero.
  The rows $j=0$ and $j=n$ are boundary cases.
  \begin{align}
    H_{0,0} + H_{0,1} &= n - n = 0\\
    H_{n,n} + H_{n,n-1} &= n - n = 0\\
    H_{j,j} + H_{j,j-1} + H_{j,j+1} &= j(n-j+1) + (j+1)(n-j) - j(n-j+1) -(j+1)(n-j)=0.
  \end{align}
  Therefore, $H$ is the Laplacian of a weighted path graph.
  Since the graph is connected, the Laplacian has a one-dimensional kernel, spanned by the all-ones state $\ket{\psi}$.
\end{proof}

Note that the Laplacian of the Cartesian product of two graphs is given by $L(G_1\square G_2) = L(G_1)\otimes I + I \otimes L(G_2)$.
Thus, $\N_{4k+1}\Hvar{\upharpoonright}_{\HHlog_\bfn}$ is the Laplacian of an $(n_1+1)\times\dotsm\times (n_k+1)$ grid.
The $\Hweight(t)$ term is not exactly a Laplacian, since it is supposed to weight the ground state.
As we see in the next lemma, some gap preserving linear transformations can turn $A(t){\upharpoonright}_{\HHvar_\bfn}$ into a Laplacian with the structure of \cref{fig:grid}.

\begin{figure}[t]
  \begin{center}
    \includegraphics[scale=0.9]{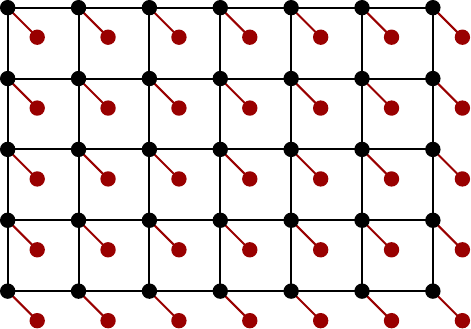}
  \end{center}
  \caption{Graph structure of $A(t){\upharpoonright}_{\HHlog_{6,4}}$, where the red edges are weighted as in \cref{fig:walk}.}\label{fig:grid}
\end{figure}

\begin{lem}
  Let $H=A(t){\upharpoonright}_{\HHvar_\bfn}$. Then $H$ has unique ground state with energy $0$
  \begin{equation}
    \ket{\psi(t)} \propto \sum_{\bfj\in S_\bfn}\left(\ket{\bfj,0}_L + \frac{1+t}{1+\bigl((1-t)F_0(\bfj)+ tF_1(\bfj)\bigr)|S_\bfn|}\ket{\bfj,1}_L\right),
  \end{equation}
  and spectral gap
  \begin{equation}
    \lambda_2(H) = \Omega(\nmax^{-2}),
  \end{equation}
  where $\nmax = \max_{i} n_i$.
\end{lem}
\begin{proof}
  Verifying $H\ket{\psi(t)}=0$ is a straightforward calculation.
  On the $\ket{\bfj,0}$ subspace, $H$ is a weighted Laplacian of a $k$-dimensional grid graph.
  The ground state is the uniform superposition.
  It remains to check $\Hweight\ket{\psi(t)}=0$.
  $\Hweight$ is block diagonal with respect to $2$-dimensional blocks that look like \cref{eq:Hweightblock}.
  Thus, we get 
  \begin{equation}
    \Hweight \left(\ket{\bfj,0} + \frac{W_1}{W_0}\ket{\bfj,1}\right) = \left(W_1^2 - W_0W_1\cdot\frac{W_1}{W_0}\right)\ket{\bfj,0} - \left(-W_0W_1 + W_0^2\cdot \frac{W_1}{W_0}\right)\ket{\bfj,1} = 0.
  \end{equation}
  By \cref{lem:unique-groundstate} the smallest eigenvalue of $H$ is unique.
  We also have $H\succcurlyeq 0$ because $H$ can be written as $H = L\otimes \ketbra0 + B$, where $L$ is the Laplacian of the grid, and $B$ is a sum of (scaled) projectors \eqref{eq:Hweightblock}.
  Thus $\ket{\psi(t)}$ is the unique ground state.
  Define the diagonal matrix
  \begin{equation}
    \begin{aligned}
    D &= \N_{4k+1} + W_0 \N_{4k+2} = \sum_{\bfj\in S_\bfn} \ketbra{\bfj,0} + w_0(\bfj)\ketbra{\bfj,1},\\
    w_0(\bfj) &= 1+\bigl((1-t)F_0(\bfj) + tF_1(\bfj)\bigr)|S_\bfn|,\\
    w_1 &= 1+t.\\
    \end{aligned}
  \end{equation}
  Let $H' = L\otimes\ketbra0 + w_1^{-2}B$.
  By \cref{cor:gapscale}, we have $\lambda_2(H) \ge \lambda_2(H')$.
  Then $L^* = D^{-1} H' D^{-1}$ is the Laplacian of the ``whiskered grid graph'' (see \cref{fig:grid}).
  This can be seen by transforming the $\Hweight$ $2\times2$ submatrices:
  \begin{equation}
    \begin{aligned}
    D^{-1}w_1^{-2}\Hweight D^{-1}\upharpoonright_{2\times 2}\ &= \begin{pmatrix}
      1&0\\
      0& w_0(\bfj)^{-1}
    \end{pmatrix}\cdot w_1^{-2}\begin{pmatrix}
    w_1^2 & -w_0(\bfj)w_1\\
    -w_0(\bfj)w_1 & w_0(\bfj)^2
    \end{pmatrix}\cdot \begin{pmatrix}
      1&0\\
      0& w_0(\bfj)^{-1}
    \end{pmatrix}\\
    &= \begin{pmatrix}
      1&-1\\-1&1
    \end{pmatrix}.
    \end{aligned}
  \end{equation}
  A single $n$-vertex path graph has smallest eigenvalues $2-2\cos(\pi/n) \sim \pi^2/n^2$ \cite{jiang12}.
  Then $\lambda_2(L) \sim \pi^2/(\nmax+1)^2$ as it is the Cartesian product of paths.
  Attaching the whiskers does not make the Laplacian much worse. In fact, \cite[Theorem 3.2(iii)]{BPS07} provides an explicit formula:
  \begin{equation}
    \lambda_2(L^*) = \frac{\lambda_2(L)+2-\sqrt{\lambda_2(L)^2 + 4}}{2} = \Theta(\nmax^{-2}).
  \end{equation}
  Finally, we have $\lambda_2(H') \ge \lambda_2(L^*)$ by \cref{lem:gapdiag}.
\end{proof}

\begin{lem}\label{lem:adiabatic-groundstate}
  Let $Z = F_1^{-1}(0) \cap S_{\bfn}$. Assuming $Z\ne\emptyset$,
  \begin{equation}
    \begin{rcases}
    \norm{\mathrlap{\ket{Z,1}_L}\phantom{\ket{0^k,1}_L} - \ket{\psi(1)}}\  \\
    \norm{\ket{0^k,1}_L - \ket{\psi(0)}}
    \end{rcases}
    \le O\!\left(\abs{S_{\bfn}}^{-1/2}\right),
  \end{equation}
  where $\ket{Z}$ is the subset state over solutions as in \eqref{eq:psitilde}, and the constant is independent of $F$.
\end{lem}
\begin{proof}
  At $t=1$, $\ket{\psi(t)}$ looks like \cref{eq:psitilde} with $t=|S_\bfn|$ and $F=F_1$, since $(1-t)F_0$ vanishes.
  Then we can use \cref{eq:psit-Z} to bound the difference.
  At $t=0$, the same holds with $F=F_0$, which has the unique all-zero solution.
\end{proof}

Therefore, we can prepare a good approximation of the ground state at $t=0$, and get a good approximation of the solution state at $t=1$.
We can raise $\hDelta$ to higher powers in $W_0$ \eqref{eq:W0} to improve the error if desired.
The current parameterization suffices since the error is already exponentially small in the input size.

We will later use two-mode squeezed vacuum (TMSV) states to initialize the logical registers, as well as an additional special gate to prepare the single photon state for the binary register.
The TMSV states have the right structure with paired photon numbers, albeit in superposition. The superposition is not an issue, since $A(t)$ leaves total photon number in each logical register invariant, i.e., it commutes with the observables $O_i = \sum_{j=4i-3}^{4i}\N_j$ for all $i\in[k]$, and $O' = \N_{4k+1}+\N_{4k+2}$.

\subsubsection{Adiabatic computation with time-independent Hamiltonian}\label{sec:adiabatic}

\begin{theorem}[Adiabatic Theorem, adapted from {\cite{Rei04}}]\label{thm:adiabatic}
  Let $k\ge1$. Let $H(s)$ with $s\in[0,1]$ be a time-dependent, finite-dimensional, $(k+1)$-times differentiable Hamiltonian, with non-degenerate continuous eigenstate $\ket{\psi(t)}$ that has gap $\gamma(t)>0$ from all other eigenvalues.
  Let $U(s)$ be the unitary evolution satisfying the Schrödinger equation corresponding to evolution of $H(s=t/\tau)$ from time $t=0$ to time $t=\tau$:
  \begin{equation}\label{eq:adiabatic-schroedinger}
    \frac{d}{ds}U(s) = -i\tau H(s)U(s),\quad U(0)=I,
  \end{equation}
  which we can also write as a time-ordered integral
  \begin{equation}
    U(s) = \calT \exp(\int_{0}^{s} -i\tau H(t)\, dt).
  \end{equation}
  Let $h(s) \ge \norm{(\frac{d}{ds})^l H(s)}$ for $l=1,\dots,k$, and $P(s)$ the projector onto $\ket{\psi(s)}$.
  Then
  \begin{equation}
    \norm{(I-P(s))\,U(s)\, P(0)} = O\!\left(\frac{1}{\tau}\left(\frac{h(0)}{\gamma(0)} + \frac{h(s)}{\gamma(s)}\right) + \max_{\sigma\in [0,s]}\frac{h^{k+1}(\sigma)}{\tau^k \gamma^{2k+1}(\sigma)}\right).
  \end{equation}
\end{theorem}

Here, we can just use $k=1$ and get
\begin{equation}
  \tau = \Omega\!\left(\frac{h^2}{\epsilon\gamma^3} \right).
\end{equation}
In order to keep the total evolution time to $1$, we effectively multiply $H$ by $\tau$ (as in \cref{eq:adiabatic-schroedinger}).
Practically, we initialize an additional ancilla $4k+3$ with a highly squeezed state, multiply $A(s)$ with $\N_{4k+3}$.

Next, we define the time-independent adiabatic Hamiltonian $\hH$ by designating an additional mode $0$ as time register, giving a Hamiltonian on now $4k+4$ modes.
\begin{equation}\label{eq:hH}
  \hH = \P_0 + A(\X_{0})\N_{4k+3} \eqcolon \P_0 + A^{(0)} + \X_{0}A^{(1)} + \X_{0}^2A^{(2)},
\end{equation}
where $\N_{4k+3}$ is absorbed into the $A^{(j)}$ terms.
Note that $\hH$ leaves the Hilbert spaces for $\bfn\in\NN_0^{4k+4}$
\begin{equation}\label{eq:logical-Hn}
  \calH_\bfn \coloneq \Span\left\{\ket{\bfx}\middle| \bfx\in \NN_0^{4k+2},\; \forall i\in [k]\colon{\textstyle\sum_{j=4i-3}^{4i}x_i = n_i,\; x_{4k+1}+x_{4k+2}=n_{k+1}, x_{4k+3}=n_{4k+3}}  \right\}
\end{equation}
invariant.
Since $\hH$ is technically an unbounded operator, we need to careful in its analysis.
The first step is proving essential self-adjointness, which means $\hH$ has a unique self-adjoint closure. Thus unitary evolution is well-defined by Stone's theorem \cite[Theorem 6.2]{Schmdgen2012}.

\begin{lem}\label{lem:H-self-adjoint}
  $\hH$ is essentially self-adjoint.
\end{lem}
\begin{proof}
  We can verify that $\Dfin$ (finite Fock state superpositions) is a dense set of analytic vectors for $\hH$.
  It suffices to check for a single Fock state $\ket{\bfx}$, $\bfx\in \NN_0^{4k+4}$, that there exists a constant $C_{\bfx}$ with
  \begin{equation}\label{eq:check-analytic}
    \norm{\hH^n \ket{\bfx}} \le C_{\bfx}^n n!.
  \end{equation}
  We can expand $\hH=\sum_{j=1}^m c_jM_j$ into monomials $M_j$ of $\a_i,\a_i^\dagger$, such that each monomial preserves total photon number in modes $1,\dots,4k+3$, and has degree $2$ in $\a_0,\a_0^\dagger$ and $d$ in the other modes.
  Let $c = \max_j \abs{c_j}$. Then
  \begin{equation}\label{eq:analytic-vector}
    \begin{aligned}
      \norm{\hH^n\ket{\bfx}} &\le \sum_{\bfj\in [m]^n}c^n\norm{ M_{j_n}\cdots M_{j_1}\ket{\bfx}} \le c^n\cdot m^n\cdot \sqrt{(2n+x_0)!}\cdot (x_1+\dotsm+x_{4k+3}+1)^{dn}\\
      &\le C_{\bfx}^n n!,
    \end{aligned}
  \end{equation}
  for some constant $C_{\bfx}$ since we can bound $\sqrt{(2n+x_0)!} \le c_{x_0}^n n!$ for some constant $c_{x_0}$, using upper and lower bounds on Stirling's approximation \cite{Robbins55}.
  Essential self-adjointness then follows from Nelson's theorem \cite[Theorem 7.16]{Schmdgen2012}.
\end{proof}

Next, we can show that evolution by $H$ preserves the Schwartz space in the clock register, for fixed photon number in the workspace.

\begin{lem}\label{lem:schwartz-invariance-subspace}
  $e^{-itH}$ preserves the Hilbert space $\calS\otimes \calH_{\bfn}$, where $\calS = \bigcap_{j=1}^\infty\calD(\N_0^j)$ is the Schwartz space in mode $0$ (see \cref{app:schwarz}).
\end{lem}
\begin{proof}
  Using $\hH$  from \cref{eq:hH}, we get
  \begin{equation}\label{eq:commute-H-N0}
    [\hH,\N_0] = -i\left(\X_0 + i\P_0A^{(1)} + i\{\X_0,\P_0\}A^{(2)}\right),
  \end{equation}
  where notably the degree in mode $0$ does not increase.
  Since the other modes are bounded in $\calH_{\bfn}$, there exists a constant $c$ (depending on $\bfn$), such that for all $\ket\psi\in\Dfin\otimes \calH_\bfn$
  \begin{equation}
    \abs{\bra{\psi}[\hH,\N_0]\ket{\psi}} \le c\bra{\psi}\N_0\ket{\psi}+1,
  \end{equation}
  analogously to \cref{lem:form-bound}.
  We can show the analogous statement for all powers of $\N_0$, as
  \begin{equation}
    [\hH,\N_0^j] = \sum_{i=0}^{j-1} \N_0^{i}[H,\N_0]\N_0^{j-1-i},
  \end{equation}
  where the maximum degree in $\a_0,\a_0^\dagger$ is at most $2j$ due to \eqref{eq:commute-H-N0}.
  Thus, the evolution of $\hH$ preserves $\calD(\N_0^j)$ for all $j$ by \cite[Theorem 2]{Faris1974}.
\end{proof}

\newcommand{\tH}{\tilde{H}}
As we have seen in \cref{thm:schwartz}, the above lemma does guarantee that evolution by $H$ preserves the Schwartz space on all modes.
However, we can still bound the energy growth under $H$, and leave the question of Schwartz space preservation open.

\begin{lem}\label{lem:adiabatic-hamiltonian-energy}
  Let $\ket{\psi_0}\in\calS(\RR^{4k+5})$ and define $\ket{\psi_t} = e^{-it\hH}\ket{\psi_0}$.
  There exists constants $c,d$, such that
  \begin{equation}\label{eq:adiabatic-hamiltonian-bounds}
    \ev{\N}{\psi_t} \le e^{c|t|}\ev{(\N+1)^d}{\psi_0},
  \end{equation}
  where $\N = \sum_{i=0}^{4k+4}\N_{i}$ is the total photon number observable.
\end{lem}
\begin{proof}
  We begin by defining the comparison operator
  \begin{equation}
    \Lambda = 1 + \N_0 + \X_0^4 + \hM^d,\qquad \hM = \sum_{i=1}^{4k+4} \N_i.
  \end{equation}
  We argue that there exists a constant $c$, such that as quadratic forms on the Schwartz space for all $\psi\in\calS$,
  \begin{equation}
    \abs{\ev{[\hH,\Lambda]}{\psi}} \le c\ev{\Lambda}{\psi}.
  \end{equation}
  We have
  \begin{equation}
    \begin{aligned}
      [\hH,\Lambda] &= [\P_0+\X_0A^{(1)}+\X_0^2A^{(2)},\N_0+\X_0^4]\\
      &=i\left([\P_0,\N_0] + [\P_0,\X_0^4]+[\X_0,\N_0]A^{(1)}+[\X_0^2,\N_0]A^{(2)}\right)\\
      &=i\left(-\X_0 - 4\X_0^3 + \P_0A^{(1)} + \{\X_0,\P_0\}A^{(2)} \right).
    \end{aligned}
  \end{equation}
  Next, we apply the inequality $\pm\ev{(ST+TS)}{\psi}\le \ev{S}{\psi} + \ev{T}{\psi}$ for self-adjoint $S,T$ on $\calD(S)\cap\calD(T)$ \cite[Lemma 4.3]{golse2021quantumsemiquantumpseudometricsapplications}.
  Then we get for all $\ket{\psi}\in\calS$
  \begin{equation}\label{eq:Lambda-commutator-bound}
    \begin{aligned}
      \abs{\ev{[\hH,\Lambda]}{\psi}} &\le \ev{\X_0}{\psi} + 4\abs{\ev{\X_0^3}{\psi}} + 8\ev{\P_0^2}{\psi} + 2\ev{A^{(1)}}{\psi} \\
      &\ \ + \ev{\bigl( (\X_0A^{(1)})\P_0 + \P_0 (\X_0A^{(1)})\bigr)}{\psi}\\
      &\le 8\ev{\Lambda}{\psi} + \ev{\P_0^2}{\psi} + \ev{\X_0^2A^{(1)2}}{\psi}\\
      &\le 9\ev{\Lambda}{\psi} + \ev{\X_0^4}{\psi} + \ev{A^{(1)4}}{\psi}\\
      &\le 10\ev{\Lambda}{\psi},
    \end{aligned}
  \end{equation}
  where the final inequality follows by setting $d$ to at least twice the maximum degree of the $A^{(j)}$.

  The next step will be to bound the $\ev{\Lambda}_t$.
  We start by decomposing $\ket{\psi_0}$ into the logical spaces, so that we can apply \cref{lem:schwartz-invariance-subspace}:
  \begin{equation}
    \ket{\psi_0} = \sum_{\bfn}\alpha_{\bfn}\underbrace{\ket*{\phi_0^{\bfn}}\ket*{\eta^{\bfn}_0}}_{\eqcolon\ket{\psi_0^\bfn}}, \quad \ket*{\eta^{\bfn}_0}\in\calH_\bfn,
  \end{equation}
  with unit $\ket{\phi_0^\bfn},\ket{\eta_0^\bfn}$.
  This decomposition can always be done by decomposing the tail in Fock basis.
  Additionally, we have $\ket{\phi_0^{\bfn}}\in\calS$ for all $\bfn$, since we know that $\ev{\N_0^j}{\psi_0}<\infty$ for all $j$, and 
  \begin{equation}
    \ev*{\N_0^j}{\psi_0} = \sum_{\bfn} \ev*{\N_0^j}{\phi_0^{\bfn}},
  \end{equation}
  hence every summand must be finite.
  There exist some constants $c_1,c_2,c_3>0$ (depending on $\bfn$), such that
  \begin{subequations}
    \begin{align}
    \norm{\Lambda e^{-it\hH}\psi_0^\bfn}^2 &\le c_1\norm{(\N+1)^d \psi_t^\bfn}^2 \\
    &= c_1 \ev{{(\N+1)^{2d}}}{\psi_t^\bfn}\nonumber\\
    &\le c_2 \ev*{(\N_0+1)^d}{\psi_t^\bfn} \\
    &\le e^{c_3|t|}\ev*{(\N_0+1)^d}{\psi_0^\bfn},
    \end{align}
  \end{subequations}
  where (a) is by \cref{lem:form-bound}, (b) follows because $\N_j,j\ge 1$ are bounded on $\calH_{\bfn}$, and (c) holds by applying the form bound derived in the proof \cref{lem:schwartz-invariance-subspace} together with the observable bound in \cite[Theorem 2, Corollary 1.1]{Faris1974}.
  Therefore, we fulfill all requirements of the abstract Ehrenfest theorem \cite{FK09}, so that $\ev{\Lambda}{\psi_t^\bfn}$ is continuously differentiable in $t$.
  Then by Grönwall and \cref{eq:Lambda-commutator-bound},
  \begin{equation}\label{eq:Lambda-time-bound}
    \ev{\Lambda}{\psi^\bfn_t} \le e^{10|t|} \ev{\Lambda}{\psi^\bfn_0},
  \end{equation}
  which holds independently of $\bfn$.
  Since the $\ket{\psi_t^\bfn}$ remain orthogonal under evolution, this inequality extends to all of $\ket{\psi_0}$:
  \begin{subequations}
    \begin{align}
    \ev{\N}{\psi_t} &= \sum_{\bfn}|\alpha_\bfn|^2 \ev{\N}{\psi^\bfn_t}\\
    &\le \sum_{\bfn}|\alpha_\bfn|^2\ev{\Lambda}{\psi^\bfn_t}\\
    &\le \sum_{\bfn}|\alpha_\bfn|^2 e^{10|t|} \ev{\Lambda}{\psi^\bfn_0}\\
    &\le ce^{10|t|}\sum_{\bfn}|\alpha_\bfn|^2 \ev{(\N+1)^{d}}{\psi_0^\bfn}\\
    &= ce^{10|t|}\ev{(\N+1)^{d}}{\psi_0},
    \end{align}
  \end{subequations}
  where (b) follows from the straightforward quadratic form bound $\N \le \Lambda$, (c) from \cref{eq:Lambda-time-bound}, (d) from \cref{lem:form-bound}, and (a),(e) from $\mel*{\psi_t^{\bfn}}{\N^d}{\psi_t^{\bfm}}=0$ for $\bfn\ne\bfm$ and all $t$.
\end{proof}

Next, we derive the action of evolution by the time-independent adiabatic Hamiltonian.

\begin{lem}\label{lem:adiabatic-wavefunction}
  Let $\Psi_0\in\calS(\RR,\calH_{\bfn})\cong \calS\otimes\calH_\bfn$.
  Then $e^{it\hH}\Psi_0 = \Psi_t$, where 
  \begin{equation}\label{eq:Psi_t(x)}
    \Psi_t(x) = U(x-t,x)\Psi_0(x-t), \qquad U(a,b) = \calT \exp(-i\int_{a}^b A(s)\,ds),
  \end{equation}
  where $A(s)$ is the adiabatic Hamiltonian acting on $\calH_\bfn$.
\end{lem}
\begin{proof}
  Since $\Psi_0\in \calD(\hH)$, we have by \cite[Lemma 1.3]{EN2000}, that $e^{-it\hH}\Psi_0$ is continuously differentiable.
  Thus, all we need to verify is that \cref{eq:Psi_t(x)} satisfies the Schrödinger equation
  \begin{equation}
    i\partial_x \Psi_t(x) = \hH\Psi_t(x) = (i\partial_x + A(x))\Psi_t(x)\ \Longleftrightarrow\ (\partial_t+\partial_x)\Psi_t = iA(x)\Psi_t(x),
  \end{equation}
  using the fact that $\P = i\partial_x$.
  By the product rule, we have
  \begin{equation}
    \begin{aligned}
    (\partial_t+\partial_x)\Psi_t(x) &= [(\partial_t+\partial_x)U(x-t,t)]\Psi_0(x-t) + U(x-t,t)\underbrace{[(\partial_t+\partial_x)\Psi_0(x-t)]}_{=0}\\
    &= \bigl[-iU(x-t,x)A(x-t)-iA(x)U(x-t,x)+iU(x-t,x)A(x-t)\bigr]\Psi_0(x-t)\\
    &= -iA(x)U(x-t,x)\Psi_0(x-t) = -iA(x)\Psi_t(x),
    \end{aligned}
  \end{equation}
  where we use the endpoint derivatives of $U(a,b)$ (see \cite[Example 5.9]{EN2000}):
  \begin{subequations}
    \begin{align}
      \partial_a U(a,b) &= iU(a,b)A(a)\\
      \partial_b U(a,b) &= -iA(b)U(a,b)\\
      \partial_x U(x-t,x) &= \partial_x U(m,x)U(x-t,m) = -iA(x)U(x-t,x) + iU(x-t,x)A(x-t)\\
      \partial_t U(x-t,x) &= \partial_t U(\underbrace{a(t)}_{=x-t},x) = a'(t)iU(a(t),x)A(a(t)) = -iU(x-t,x)A(x-t).
    \end{align}
  \end{subequations}
\end{proof}

Using the position eigenstate $\ket{x=0}$ to initialize clock register, we would get
\begin{equation*}
  e^{-itH}\ket{x=0}\ket{\psi_0} = \ket{x=t}\otimes U(t,0)\ket{\psi_0}.
\end{equation*}
Since the $\ket{x=0}$ state is unphysical (in fact not even in the Hilbert space), we instead simulate it with a squeezed vacuum state $\ket{\phi_r} = \hS(r)\ket{0}$.
\begin{lem}\label{lem:adiabatic-approx}
  Let $\ket{\Psi_0} = (\hS(r)\ket{0})\otimes \ket{\psi_0}$ with $\ket{\psi_0}\in \calH_\bfn$.
  Let $\ket{\psi_t} = U(0,t)\ket{\psi_0}$,
  $\rho_t = \ketbra{\psi_t}$, and $\wtrho_t = \tr_0(\ketbra{\Psi_t})$.
  Then $D(\rho_t,\wtrho_t) \le O((t+t^2)\nmax^{O(1)}e^{-r})$.
\end{lem}
\begin{proof}
  The wave function in position representation of $\ket{\phi_r} \coloneq \hS(r)\ket0$ is a Gaussian with $0$ mean \cite{Schumaker1986}
  \begin{equation}
    \phi_r(x) = \frac1{(2\pi\sigma^2_X)^{1/4}}\exp(-\frac{x^2}{4\sigma_X^2}),
  \end{equation}
  where $\sigma_X^2 = e^{-2r}/2$ is the position variance.
  \cref{lem:adiabatic-wavefunction} gives the wavefunction of $\Psi_t$ (interpreted as $\calS(\RR,\calH_\bfn)$) as
  \begin{equation}
    \Psi_t(x) = \phi_r(x-t)\cdot U(x-t,x)\ket{\psi_0}.
  \end{equation}
  Thus, tracing out the clock register gives
  \begin{equation}
    \begin{aligned}
    \wtrho_t = \tr_0\ketbra{\Psi_t}\;&= \int_\RR \abs{\phi_r(x-t)}^2U(x-t,x)\ketbra{\psi_0}U^\dagger(x-t,x)\,dx \\
    &= \int_\RR \abs{\phi_r(x)}^2U(x,x+t)\ketbra{\psi_0}U^\dagger(x,x+t)\,dx.
    \end{aligned}
  \end{equation}
  We bound the distance between $\wtrho_t$ and $\rho_t = \ketbra{\psi_t}$:
\begin{equation}\label{xeq:Drho}
  \begin{aligned}
  D(\wtrho_t, \rho_t)&=\frac12\norm{\int_\RR \abs{\phi_r(x)}^2 \bigl(U(x,x+t)\ketbra{\psi_0}U^\dagger(x,x+1)-\ketbra{\psi_t}\bigr)}_1\,dx \\
  &\le \frac12\int_\RR \abs{\phi_r(x)}^2 \norm\big{U(x,x+t)\ketbra{\psi_0}U^\dagger(x,x+1)-U(0,t)\ketbra{\psi_0}U^\dagger(0,t)}_1\,dx\\
  &\le \int_\RR \abs{\phi_r(x)}^2 \norm\big{U(x,x+t)\ket{\psi_0}-U(0,t)\ket{\psi_0}}\,dx,
  \end{aligned}
\end{equation}
where the last inequality follows from the bound on trace distance of pure states
\begin{equation}
  D(\ketbra{u},\ketbra{v})=\sqrt{1-\abs{\braket{u}{v}}^2}\le \norm{\ket{u}-\ket{v}}.
\end{equation}
Now we bound the norm of the state difference via the difference of the unitaries
\begin{subequations}
  \begin{align}
  \norm{\bigl(U(x,t+t)-U(0,t)\bigr)\ket{\psi_0}} &\le \norm\big{U(x,x+t)-U(0,t)}\\
  &\le \int_0^t\norm\big{A(x+s)-A(s)}\,ds\\
  &\le \int_0^t\norm\big{xA^{(1)} + ((x+s)^2-s^2)A^{(2)}}\,ds\\
  &\le \int_0^t\norm\big{xA^{(1)} + (2sx+x^2)A^{(2)}}\,ds\\
  &\le \int_0^t\abs{x}\norm\big{A^{(1)}} +(\abs{2sx} + \abs{x^2}) \norm\big{A^{(2)}}\,ds\\
  &\le \abs{tx} \norm\big{A^{(1)}} +(\abs{2t^2x} + \abs{tx^2}) \norm\big{A^{(2)}}\\
  &\le \bigl((\abs{t}+t^2)\abs{x} +\abs{t}x^2 \bigr)\nmax^{O(1)},
  \end{align}
\end{subequations}
where (a) and all subequent inequalities are understood on the subspace $\calH_\bfn$, and (b) follows from \cite[Appendix B]{PhysRevX.9.031006} (see also \cite[Lemma 1]{Berry2020timedependent}).
Plugging this back into \cref{xeq:Drho}, we have
\begin{equation}
  \begin{aligned}
   D(\wtrho_t, \rho_t) &\le \bigl(\abs{t}+t^2\bigr)\nmax^{O(1)}\left(\int_\RR \abs{\phi_r(x)}^2\abs{x}\,dx + \int_\RR \abs{\phi_r(x)}^2x^2\,dx\right)\\
   &\le\nmax^{O(1)}(\abs{t}+t^2)e^{-r},
  \end{aligned}
\end{equation}
where we use the fact that the first integral is exactly the mean deviation of a normal distribution, i.e. $\mathbb{E}[\abs*{\X}] = \sqrt{\frac2\pi}\sigma_X$ \cite{weisstein_meandev}, and the second integral is $\mathbb{E}[\X^2] = \sigma_X^2$.
\end{proof}

The final error is therefore exponentially small in the squeezing parameter and the required squeezing parameter is logarithmic in the runtime.

\paragraph*{Preparing the solution state.}

Starting with the right initial state, we can prepare the uniform superposition over the solution states to the Diophantine equation $F(x_1,\dots,x_k)$ in the search space, using adiabatic computation.
The search space is determined by the input state.
Below we denote these input parameters by the letter $E$ to emphasize the the photon number of the input state relates to its energy.

\begin{lem}\label{lem:compute-output-ideal}
  Given initial state
  \begin{equation}
    \ket{\psi} = \ket{0}_{0}\otimes \bigotimes_{i=1}^k\ket{0,0,E_i,E_i}_{4i-3,\dots,4i}\otimes\ket{0,1}_{4i+1,4i+2}\ket{E_T}_{4i+4}\ket{E_S}_{4i+5},
  \end{equation}
  where $E_T \ge E^{O(k)}$, and $E_S \ge \Omega(\log E_T)$, with $E = \max_i E_i$,
  we can prepare a Fock state $\ket{x_1,\dots,x_k}$ within trace distance $E^{-\Omega(1)}$, such that $F(x_1,\dots,x_k)=0$ and $x_i\le E_i$ for all $i\in[k]$ (or rather a superposition over all such solutions as in \cref{lem:adiabatic-groundstate}).
  The circuit uses $2$ gates (one depends on $F$) and constant runtime.
\end{lem}
\begin{proof}
  The first step is to apply a controlled squeezing gates with Hamiltonians $H = \frac{i}{2}\N_{\ctl}\otimes (\a_{\tgt}^{\dagger 2}-\a_{\tgt}^{2})$ with control $4i+4$ and target $0$.
  This gives us suitable approximations of the $\ket{x=0}$ clock register to perform adiabatic evolution for time $E_T$ (\cref{lem:adiabatic-approx}).
  Next, we apply the time-independent adiabatic evolution (\cref{eq:adiabatic-hamiltonian}), using $F_0(x_0,\dots,x_k) = x_0+\dots+x_k$ and $F_1=F$ (\cref{eq:W0}) as weights at times $t=0$ and $t=1$.
  This guarantees that our initial state, logically encoding $\ket{0^k,1}_L$ (\cref{eq:logical-Hn}),
  has overlap $E^{-\Omega(1)}$ with the ground state of $A(0)$.\footnote{The construction could easily be adapted to handle arbitrarily small error, as discussed around \cref{eq:hH}.}
  The $E_T$ parameter takes the role of $\tau$ in \cref{thm:adiabatic} and is chosen sufficiently large to achieve error $E^{-\Omega(1)}$.
\end{proof}

Unfortunately, we are not aware of a method to prepare the input state of \cref{lem:compute-output-ideal} (even approximately) with a fixed gateset, considering the (super-) exponentially large photon number states required for \cref{thm:CVBQP-TOWER,lem:tower}.
Additionally, \cref{thm:CVBQP-NP,thm:CVBQP-TOWER} require some variables of the Diophantine equation to be controlled by the input.

\subsubsection{Preparing the input for the Diophantine equation}\label{sec:Diophantine-input}

To solve the problems in \cref{thm:Diophantine,thm:MA78}, we need to set some coefficients in the Diophantine equations based on the input.
We treat these coefficients as variables of the polynomial and prepare a suitable input.
These fixed inputs will be left invariant by the time-independent adiabatic Hamiltonian $\hH$ (see \cref{eq:adiabatic-hamiltonian,eq:hH}).
For simplicity, assume we want to fix the first variable.
So we would like to prepare Fock state $\ket{x}$ in mode $1$.
However, there does not appear to be an immediate way to prepare arbitrary Fock states with a fixed gateset.\footnote{Observe that our Diophantine gadget can in fact prepare Fock states, but of course we cannot use that to prepare its input.}
Instead, we use coherent states, which have a photon number distribution $\Poisson(\n),\,\n=\abs{\alpha}^2$ \cite[Eq. (3.25)]{Gerry_Knight_2004}.

\begin{lem}[\cite{poisson}]
  Let $X\sim \Poisson(\lambda)$. Then $\Pr[|X-\lambda| > \delta] \le 2e^{-\frac{\delta^2}{2(\lambda+\delta)}}$ for $\delta>0$.
  In particular $\Pr[|X-\lambda| > \lambda^{1/2+\epsilon}] \le 2e^{-\lambda^{2\epsilon}/3}$ for $\epsilon>0$.
\end{lem}

Therefore, if we prepare the coherent state $\hD(x^2)\ket{0}$, the photon number will be $x^4 + o(x^3)$ with high probability.
Fortunately, the difference between two adjacent fourth powers is $(x+1)^4 - x^4 = 4x^3+6x^2+4x+1$.
Hence, a Diophantine equation can recover $x$ from $y=x^4 + o(x^3)$ by encoding $|x^4-y| \le x^3$:
\begin{equation}\label{eq:Diophantine-error-correction}
  f(x,y,s) = \bigl(x^6 - (x^4 - y)^2 - s\bigr)^2,\quad x,y,s\in\NN_0.
\end{equation}
The remaining challenge is to prepare $\hD(x^2)$ without actually running the displacement gate for time $x^2$ (which is exponential in the input size).
The idea is to prepare $\ket{x_1 = x^2}$ using SUM gates, and approximate the Dirac eigenstate with a squeezed state as in \cref{lem:adiabatic-approx}.
Then a final SUM gate prepares onto a vacuum mode prepares $\hD(x^2)\ket{0}$.

The SUM gate $U_{ij}(t)= e^{-it\X_i\P_j}$ acts as 
\begin{equation}
  \begin{aligned}
  \X_j &\mapsto \X_j + t\X_i,& \X_i&\mapsto \X_i,\\
  \P_i&\mapsto \P_i - t\P_j,& \P_j&\mapsto\P_j.
  \end{aligned}
\end{equation}
This directly follows from BCH as $[H,\X_j] = -it\X_i$ and $[H,\X_i]=0$ for $H = t\X_i\P_j$.
Since higher order commutators vanish, we get
\begin{equation}
  U_{ij}^\dagger(t) \X_j U_{ij}(t) = \X_j + i[H,\X_j] + \frac{i^2}{2!}[H,[H,\X_j]] + \dots = \X_j + t\X_i.
\end{equation}
Therefore, we can use $O(n)$ such gates to prepare $\ket{x_1=2^n}$:
\begin{equation}\label{eq:prepare-2n}
  \bigl(U_{21}(1)\,U_{12}(1/2)\bigr)^n\ket{x_1=1,x_2=1/2} = \ket{x_1=2^n,x_2=2^{n-1}}.
\end{equation}
In our actual computation, we would like to prepare the input $x\in\NN$ to the Diophantine equation.
By inserting additional gates $U_{13}(1)$ in \cref{eq:prepare-2n} at appropriate locations, we get a unitary $V$ of $O(n)$ gates, where $n$ is the number of bits in $x$, such that
\begin{equation}
  V\ket{x_1=1,x_2=1/2,x_3=0}=\ket{x_1=2^n,x_2=2^{n-1},x_3=x}.
\end{equation}

\begin{lem}\label{lem:Diophantine-input}
  Let $x\in\NN$ with $x = \sum_{i=0}^n b_i2^{i}$ with $b_0,\dots,b_{n-1}\in\{0,1\}$.
  There exists a Gaussian circuit of $O(n)$ gates that given input $\ket{\psi} = \hS_1(r)\hS_2(r)\hS_3(r)\ket{0,0,0}$, prepares a state in mode $3$, such that $\E[\X_3] = x$ and $\Var(\X_3)=O(e^{-2r}x^2)$.
\end{lem}
\begin{proof}
  One round is $R = U_{21}(1)U_{12}(1/2)$. 
  In the Heisenberg pictures, $R$ implements the following transformation:
  \begin{equation}
    \begin{aligned}
    \begin{pmatrix}
      \X_1\\\X_2
    \end{pmatrix}\mapsto
    M\begin{pmatrix}
      \X_1\\\X_2
    \end{pmatrix},
    \qquad M &= \begin{pmatrix}
      \frac32&1\\
      \frac12 &1
    \end{pmatrix} = P\,\diag(2,2^{-1})\,P^{-1},\\P &= \begin{pmatrix}
      2&1\\1&-1
    \end{pmatrix},\;
    P^{-1} = \frac13\begin{pmatrix}
      1&1\\1&-2
    \end{pmatrix}.
    \end{aligned}
  \end{equation}
  Thus we have
  \begin{equation}\label{eq:Mk}
    M^k = \frac13\begin{pmatrix}
      2^{k+1}+2^{-k} & 2^{k+1} - 2^{1-k}\\
      2^k - 2^{-k} & 2^k + 2^{1-k}
    \end{pmatrix}.
  \end{equation}
  To prepare the output state, we do for $k=0,\dots,n-1$: If $b_k=1$, apply $U_{13}(1)$. Apply $R$.
  Then the output mode will have
  \begin{equation}
    \X_3^{(n)} = \X_3^{(0)} + \sum_{k=0}^{n-1}b_k\X_1^{(k)},
  \end{equation}
  where $\X_i^{(j)}$ denotes the Heisenberg operator after the $j$-th round.
  By \cref{eq:Mk}, we get
  \begin{equation}
    \X_1^{(k)} = \alpha_k \X_1^{(0)} + \beta_k\X_2^{(0)}, \; \alpha_k = \frac13(2^{k+1}+2^{-k}),\;\beta_k = \frac13(2^{k+1}-2^{1-k}),
  \end{equation}
  and thus the output mode is
  \begin{equation}
    \X_3^{(n)} = \X_3^{(0)} + A\X_1^{(0)} + B\X_2^{(0)},\quad A =  \sum_{k=0}^{n-1}b_k\alpha_k,\; B = \sum_{k=0}^{n-1}b_k\beta_k.
  \end{equation}
  Since the input modes are independent and $A,B \in O(x)$, we get
  \begin{equation}
  \Var(\X_3^{(n)}) = O(x^2)\Bigl(\Var(\X_1^{(0)}) + \Var(\X_2^{(0)}) + \Var(\X_3^{(0)})\Bigr) = O(x^2 e^{-2r}).
  \end{equation}
  By Chebyshev's inequality, a squeezing of $r=\Omega(n)$ suffices to get error $e^{-\Omega(r)}$ with high probability.
\end{proof}

\begin{corollary}{\label{cor:displacement-x2}}
  Let $\rho_4$ be the reduced state in mode $4$ after applying $U_{3,4}(t)$ onto the output of \cref{lem:Diophantine-input} (with $x^2$ instead of $x$) and a fresh vacuum ancilla.
  Then
  \begin{equation}
  \int_\RR \hD(s) \ketbra{0}\hD^\dagger(s) \frac{1}{\sqrt{2\pi\sigma^2}}e^{-\frac{(s-\mu)^2}{2\sigma^2}} \,ds,
  \end{equation}
  where $\mu=x^2$ and $\sigma^2 = O(e^{-2r}x^4)$.
\end{corollary}
\begin{proof}
Since the final SUM gate commutes with measuring the output mode of \cref{lem:Diophantine-input} in position basis, treat mode $3$ as a random position eigenstate sampled from a Gaussian distribution.
\end{proof}

In practice, we can choose $r$ so that $e^{-r}x^4 \le 1/x$, so that the Diophantine equation \cref{eq:Diophantine-error-correction} can correct the error with high probability.

\subsubsection{Preparing high energy states}\label{sec:high-energy}

With just Gaussian gates, we only get to exponential energy using a polynomial number of gates and time.
However, the controlled squeezing gate with Hamiltonian $H = \frac{i}{2}\N_{\ctl}\otimes (\a_{\tgt}^{\dagger 2}-\a_{\tgt}^{2})$ roughly exponentiates energy with each application as $e^{iH}\ket{n,0}=\ket{n}\otimes \hS(n)\ket{0}$ and $\hS(n)\ket{0}$ has average photon number $\sinh^2 n = \Theta(e^n)$.
We actually need two-mode squeezed vacuum states (TMSV) in order to prepare the states $\ket{E_i,E_i}$ for \cref{lem:compute-output-ideal}.
TMSVs also simplify the analysis of repeated squeezing since we can trace out one half and treat the other as a Fock state.

\begin{lem}[{\cite[Section 7.7]{Gerry_Knight_2004}}]\label{lem:TMSV}
  Let $\hS_2(\xi) = \exponential(\xi^*\a_1\a_2 - \xi \a_1^\dagger \a_2^\dagger)$.
  Then the TMSV is
  \begin{equation}
    \hS_2(re^{i\theta})\ket{0,0} = \frac1{\cosh r}\sum_{n=0}^{\infty} e^{in\theta}(-\tanh r)^n \ket{n,n},
  \end{equation}
  and the reduced density operators of each mode are given by
  \begin{equation}\label{eq:TMSV-mode}
    \frac{1}{\cosh^2 r} \sum_{n=0}^{\infty} \tanh^{2n} r \ketbra{n}{n} = (1-\lambda^2)\sum_{n=0}^\infty \lambda^{2n}\ketbra{n}{n},\;\lambda^2=\tanh^2 r=\frac{\n}{\n+1},
  \end{equation}
  with mean photon number $\n=\trace(\N\rho)=\sinh^2 r$.
\end{lem}

Note that the photon number follows a geometric distribution with $p=1/(1+\n)$.
We can give upper and lower bounds that hold with high probability.

\begin{lem}\label{lem:TMSV-lower}
  Let $\rho$ as in \cref{eq:TMSV-mode} with $\n\ge 1$, and $m = \lceil \n^c\rceil$.
  Then $\Probability[\N < m] \le m/\n \le  2\n^{c-1}$ for $c\in (0,1)$.
\end{lem}
\begin{proof}
  Let $m=\lceil\n^c\rceil$.
  We have $\Probability[\N < m] = 1 - (1-p)^{m} \le 1-(1-mp)\le m/\n$.
\end{proof}

Thus, squeezing $r$ gives at least $e^{r}$ photon with high probability.
On the other hand, the photon number will not be bigger than $e^{3r}$ with high probabilty:

\begin{lem}\label{lem:TMSV-upper}
  Let $\rho$ as in \cref{eq:TMSV-mode} with $\n\ge 1$.
  Then $\Probability[\N > k\n] \le e^{-k/2}$.
\end{lem}
\begin{proof}
  $\Probability[\N > k\n] = (1-p)^{\lfloor k\n\rfloor+1} \le e^{-k\n/(\n+1)}$.
\end{proof}

We also have a similar upper bound for the SMSV, which will be helpful later for bounding the energy, since the probability of measuring $k\n$ photons decays exponentially fast in $k$.
\begin{lem}\label{lem:SMSV-upper}
  Let $\ket{\psi} = S(r)\ket{0}$ be the SMSV state with squeezing parameter $r>0$ and $k>0$.
  Then 
  \begin{equation}
    \Probability[\N > 2k\n] \le \sqrt{\frac{\n+1}{\pi k\n}}\exp(-k\n\ln(1+1/\n)) \lesssim \frac{1}{\sqrt{\pi k}} e^{-k},
  \end{equation} with $\n = \sinh^2 r$.
\end{lem}
\begin{proof}
  We have \cite[Eq. (7.69)]{Gerry_Knight_2004}
  \begin{equation}
    \Probability[\N = 2m] = \frac1{\sqrt{\n+1}}\binom{2m}{m} \left(\frac{t}{4}\right)^m,\quad t\coloneq\tanh^2 r = \frac{\n}{\n+1}.
  \end{equation}
  By Stirling's approximation, we have $\binom{2m}{m} \le 4^m/\sqrt{\pi m}$ for $m\ge1$.
  Let $K= \lfloor k\n\rfloor+1$. Then
  \begin{equation}
    \begin{aligned}
    \Probability[\N >2k\n] &= \sum_{m=K}^\infty \Probability[N=2m] \le  \frac1{\sqrt{\pi(\n+1)}} \sum_{m\ge K} \frac{t^m}{\sqrt{m}} \le \frac{1}{\sqrt{\pi(\n+1)}}\cdot\frac{1}{\sqrt{K}}\cdot \frac{t^K}{1-t}\\
    &=\sqrt{\frac{1+\n}{\pi K}}t^K \le \sqrt{\frac{1+\n}{\pi K}}\exp(-K\ln\frac{\n+1}{\n}) \le \sqrt{\frac{1+\n}{\pi k\n}} e^{-k\ln(1+1/\n)}.
    \end{aligned}
  \end{equation}
\end{proof}

Finally, we can generate states whose energy corresponds to a power tower.
With this, we can prepare sufficient approximations to the input state for \cref{lem:adiabatic-approx}.
The inputs to the ``variable registers`` receive both halves of the TMSV, which we can still treat as a Fock state $\ket{E_i,E_i}$ since the adiabatic evolution preserves photon number, i.e., can be diagonalized with respect to subspaces of fixed photon number (see \cref{lem:schwartz-invariance-subspace}).
Also recall that we can prepare the $\ket{1}$ Fock states with \cite{arzani2025can} as discussed in \cref{sec:Hamiltonian-concept}.

\begin{lem}\label{lem:tower}
  Suppose we have initial state $\hS_2(r)\ket{0,0}$ and $2m$ vacuum ancillas.
  Apply the photon number controlled two mode squeezer $\exponential(\N_{0}(\a_1\a_2-\a_1^\dagger \a_2^\dagger))$, starting in modes $2,3,4$ and then shifting two modes to the right in reach round.
  Then with probability $1-2^{-\Omega(r)}m$, the last mode has photons in range
  \def\rddots#1{\cdot^{\cdot^{\cdot^{#1}}}}
  \begin{equation}\label{eq:tower}
    e^{e^{\rddots {e^r}}}, \quad b^{b^{\rddots {b^r}}},\,b=e^3,
  \end{equation}
  where the tower has total height $m+1$.
\end{lem}
\begin{proof}
  By \cref{lem:TMSV-lower,lem:TMSV-upper}, a TMSV (with sufficiently large $r$) has photon number between $e^{r}$ and $e^{3r}$.
  If $A,B$ are the current lower and upper bounds, then one controlled squeezer gives $e^{A},e^{3B}=b^B$.
  After $m$ iterations, we get \cref{eq:tower}.
\end{proof}

\subsubsection{Extracting the output}\label{sec:extract-output}

So far, we have the tools to prepare a state $\rho$ which is close (in trace distance) to a superposition  of solutions $\ket{\bfx}$ with $F(\bfx) = 0$ (assuming such a solution exists within the search space).
The remaining task can be phrased as follows: Given Fock state $\ket{\bfx}$, prepare an output mode, such that for $a,b<\poly(n)$,
\begin{equation}\label{eq:Fout}
  \begin{aligned}
    F(\bfx)&=0 \quad&\Longrightarrow\quad& \Pr[\N_\out \in [a+1,b]] \ge \frac23,\\
    F(\bfx)&\ne 0 \quad&\Longrightarrow\quad& \Pr[\N_\out \in [0,a]] \ge \frac23,
  \end{aligned}
\end{equation}

Define Hamiltonian $G^{(c)}$ with real parameter $c\ge0$:
\begin{equation}
  G^{(c)} =  c \a^\dagger \a + \frac i2\left(\a^{\dagger2} - \a^2\right).
\end{equation}
$G^{(c)}$ is the Hamiltonian of the \emph{detuned degenerate parametric amplifier} \cite{WC19}.
For $c>1$ the dynamics are oscillatory (bounded energy) and for $c<1$ exponential ($e^{-itG^{(0)}} = \hS(-t)$).

\begin{lem}\label{lem:dpa}
  Let $\ket{\psi(t)} = e^{-itG^{(c)}}\ket{0}$.
  For $c=0$, $\ket{\psi(t)} = \hS(-t)\ket{0}$.
  For $c\ge2$, $\ket{\psi(t)}$ is a rotated squeezed vacuum state with $0$ displacement and squeezing $r(t)<1$.
\end{lem}
\begin{proof}
  The $c=0$ case holds trivially.
  The Heisenberg equations of motion are given by \cite[Sec. II B]{WC19}
  \begin{equation}
    \frac{d}{dt}\begin{pmatrix}
      \a(t)\\\a^\dagger(t)
    \end{pmatrix} = \underbrace{\begin{pmatrix}
      -ic&1\\1&ic
    \end{pmatrix}}_{M}
    \begin{pmatrix}
      \a(t)\\\a^\dagger(t)
    \end{pmatrix}= (-ic\sigma_z + \sigma_x) \begin{pmatrix}
      \a(t)\\\a^\dagger(t)
    \end{pmatrix}.
  \end{equation}
  We have $M^2 = (1-c^2)I$.
  With $\mu \coloneq 1-c^2$ we therefore get $M^{2n} = \mu^nI,M^{2n+1} = \mu^nM$.
  Let $c>1$.
  Grouping $e^{Mt}$ into even and odd terms, we get for $\omega = \sqrt{-\mu}=\sqrt{c^2-1}$
  \begin{equation}
    e^{Mt} = \sum_{n=0}^\infty \frac{(Mt)^n}{n!} = \sum_{n=0}^\infty \frac{\mu^nt^{2n}}{(2n)!}I + \sum_{n=0}^\infty \frac{\mu^nt^{2n+1}}{(2n+1)!}M = \cos(\omega t)I + \frac{\sin(\omega t)}\omega M.
  \end{equation}
  Thus we get,
  \begin{equation}
    \begin{aligned}\label{eq:dpa:uv}
    \a(t) &= u(t) \a + v(t) \a^\dagger, \\
    u(t) &= \cos(\omega t) - i\frac c\omega\sin(\omega t),\quad v(t) = \frac{\sin(\omega t)}\omega.
    \end{aligned}
  \end{equation}
  Hence, $\ket{\psi(t)}$ is a Gaussian state with zero displacement.
  Let $U = e^{-itG^{(c)}}$.
  We can write $U=\hR(\theta)\hS(\zeta)$ since it is Gaussian with displacement $0$ \cite{weedbrook2012gaussian}.
  To determine the parameters $\theta\in\RR$ and $\zeta=e^{i\varphi}r$ we match the expressions for $\a(t)$:
  \begin{equation}
    U^\dagger a U = \hS^\dagger(\zeta)\hR^\dagger(\theta) \a \hR(\theta) \hS(\zeta) = e^{i\theta}\left(\a \cosh r - e^{i\varphi}\a^\dagger \sinh r\right),
  \end{equation}
  which follows from 
  \begin{equation}
    \hR^\dagger(\theta) \a \hR(\theta) = e^{i\theta}\a,\qquad \hS^\dagger(\zeta) \a \hS(\zeta) = a \cosh r - e^{i\varphi}\a^\dagger\sinh r.
  \end{equation}
  Solving for $r$ gives
  \begin{equation}
    r(t) = \arcsinh\frac{\abs{\sin(\omega t)}}{\omega} \le \arcsinh\frac1\omega = \frac12\ln\frac{c+1}{c-1} < 1,
  \end{equation}
  assuming $c\ge 2$.
\end{proof}

\begin{corollary}\label{cor:Hout}
  Let $H_\out =  \a_0^\dagger \a_0\cdot 2F(\N_1,\dots,\N_k) + \frac i2(\a_0^{\dagger2} - \a_0^2)$, and $\ket{\psi} = e^{-tH_\out}\ket{0}\ket{\bfx}$ with $\bfx\in \NN_0^k$.
  Then $\ket{\psi} = \ket{\eta}\otimes\ket{\bfx}$ with $\ket{\eta} = \hS(-t)\ket{0}$ if $F(\bfx)=0$.
  Otherwise $\ket{\eta} = \hS(re^{i\varphi})\hR(\theta)\ket{0}$ with $r \le 1$.
\end{corollary}

So applying $H_\out$ for time $t$ to a potential solution either gives a state with squeezing $t$ (yes case), or a state with squeezing $\le1$ (no case).
The upper bound for the photon number large $t$ follows from \cref{lem:SMSV-upper}, and the lower bound by the below \cref{lem:smsv:lefttail}.
For small $r$, we can get an easy exponential tail bound more directly.

\begin{lem}\label{lem:smsv:smalltail}
  Squeezed state $\hS(r)\ket{0}$ with $r\in [0,1]$ has $\Probability[\N \ge k] = \exp(\Omega(-k))$.
\end{lem}
\begin{proof}
  Let $\tau = \tanh r$ and recall
  \begin{equation}
    \Probability[\N = 2m] \le \frac{1}{\cosh r}\frac{(2m)!}{2^{2m}(m!)^2}\tau^{2m} \le \frac{1}{\cosh r} \tau^{2m}.
  \end{equation}
  Thus with $M = \lceil k/2\rceil$
  \begin{equation}
    \Probability[\N \ge k] = \frac{1}{\cosh r}\sum_{m\ge M} \tau^{2m} = \frac{1}{\cosh r}\cdot\frac{\tau^{2M}}{1-\tau^2} \le \cosh r\,\tau^{2M} \le \cosh r \,(\tanh r)^k < 1.6 (0.8)^k,
  \end{equation}
  where $1-\tau^2 = 1-\tanh^2 r = \sech^2 r = 1/\cosh^2 r$.
\end{proof}

We also need a left tail bound for the highly squeezed state:
\begin{lem}\label{lem:smsv:lefttail}
  Squeezed state $\hS(r)\ket{0}$ with $r>0$ has 
  \begin{equation}
    \Probability[\N < k]\le \frac{1}{\sqrt{\n+1}} + \sqrt{\frac{2k}{\pi(\n+1)}},
  \end{equation}
  which gives $\Probability[\N < \sqrt{\n}] = O(\n^{-1/4})$.
\end{lem}
\begin{proof}
  Let $\n = \sinh^2 r$, $t \coloneq \tanh^2 r = \n/(\n+1)$ and recall (see \cref{lem:SMSV-upper})
  \begin{equation}
    \Probability[\N=2m] = \frac{1}{\cosh r} \frac{\binom{2m}{m}}{4^m} \tanh^{2m} r \le \frac{t^m}{\cosh r\sqrt{\pi m}}.
  \end{equation}
  Set $M \coloneq \lfloor (k-1)/2 \rfloor$.
  Then
  \begin{equation}
    \Probability[\N < k] \le \frac{1}{\sqrt{\n+1}}\left(1 + \frac{1}{\sqrt\pi}\sum_{m=1}^M\frac{t^m}{\sqrt{m}} \right) \le \frac{1}{\sqrt{\n+1}}\left(1 + \frac{2}{\sqrt\pi}\sqrt{M} \right) \le \frac{1}{\sqrt{\n+1}} + \sqrt{\frac{2k}{\pi(\n+1)}},
  \end{equation}
  where we $\sum_{m=1}^M t^{m}m^{-1/2}\le \sum_{m=1}^M m^{-1/2} \le 2\sqrt{M}$ as $t \in (0,1)$ and $m^{-1/2}\le 2(\sqrt{m}-\sqrt{m-1})$ for $m\ge1$.
\end{proof}

\begin{lem}
  State $\ket{\eta}$ from \cref{cor:Hout} has $\Probability[\N < \n] > 1 - e^{-\Omega(\n)}$ in the NO case, and $\Probability[\N \in [\sqrt{\n},\n^{3/2}]] \ge 1- \n^{-1/4}$ for $\n=\sinh^2 t$.
\end{lem}
\begin{proof}
  Follows from \cref{lem:SMSV-upper,lem:smsv:smalltail,lem:smsv:lefttail}.
\end{proof}

Thus we can improve the success probability of $2/3$ in \cref{eq:Fout} to $1-1/p(n)$ for an arbitrary polynomial $p$, setting $t=\Theta(\ln n)$.

\subsubsection{Completing the proof}\label{sec:complete-proof}

Now that we have tools, we summarize how they fit together.

\begin{proof}[Proof of \cref{thm:CVBQP-NP}]
  To show $\NP\subseteq \CVBQP[\calG]$, we construct a gateset, such that the $\NP$-complete language $L_{\mathrm{MA}}\in \CVBQP[\calG]$ (\cref{thm:MA78}).
  Given $\alpha,\beta,\gamma\in\NN_0$, we need to decide whether there exist $x_1,x_2\in\NN_0$, such that $\alpha x_1^2+\beta x_2-\gamma=0$.
  To make our gateset independent from the input, we define the polynomial 
  \begin{equation}
    F(x_1,x_2,\alpha,\beta,\gamma) \coloneq \bigl(\alpha x_1^2 + \beta x_2 - \gamma \beta\bigr)^2.
  \end{equation}
  Note that the equation we will actually realize contains $6$ additional variables for decoding the coherent state inputs, and the decoding terms \eqref{eq:Diophantine-error-correction} as discussed in \cref{sec:Diophantine-input}.
  The $\CVBQP$ circuit then works as follows:
  \begin{enumerate}
    \item Prepare tensor product of two-mode squeezed states for the variable modes corresponding to $x_1,x_2$ (see \cref{sec:Diophantine-input}), and one for the time parameter $E_T$ in \cref{lem:compute-output-ideal}.
    Choose squeezing parameter $\poly(n)$, where $n$ is the total bit length of $\alpha,\beta,\gamma$, so that \cref{lem:TMSV-lower} guarantees $\ge \gamma$ photons with probability $1-\exp(-n)$, while \cref{lem:TMSV-upper} gives energy bound $\exp(n)$.
    \item Prepare single mode squeezed states for the clock (see \cref{lem:adiabatic-approx}), as well as for the coherent state input preparation routine (see \cref{sec:Diophantine-input}).
    \item Use the repeated addition routine of \cref{lem:Diophantine-input,cor:displacement-x2} to displace the three input modes by exactly $\alpha^2,\beta^2,\gamma^2$.
    \item Apply the time-independent adiabatic Hamiltonian $\hH$ \eqref{eq:hH}. The trace distance error can be bounded as $\exp(-n)$ via \cref{lem:adiabatic-approx} by choosing sufficiently high $\poly(n)$ squeezing.
    \item If there exists $x_1,x_2$, such that $F(x_1,x_2,\alpha,\beta,\gamma)=0$, then \cref{lem:compute-output-ideal} guarantees that the adiabatic evolution finds it.
    \item The output gadget of \cref{cor:Hout} prepares a suitable output mode with $O(1)$ squeezing in the NO case, and $\Theta(\ln n)$ squeezing in the YES case, which can easily be distinguished by a photon number measurement.
  \end{enumerate}
  By \cref{lem:adiabatic-hamiltonian-energy} energy throughout the entire evolution is bounded by $O(\ev*{\N^d}_0)$ since we only evolve for constant time (the ``time parameter mode'' is considered to be part of the Hamiltonian in \cref{lem:adiabatic-hamiltonian-energy}), and the evolution is well-defined as the Hamiltonians are essentially self-adjoint (\cref{lem:H-self-adjoint}).
  The initial state will be a tensor product of (two-mode) squeezed vacuum states.
  From $\hS^\dagger(r)\a \hS(r) = \a\cosh r - \a^\dagger\sinh r$ \cite[Eq. (7.12)]{Gerry_Knight_2004}, we get a rough bound for the moments of the squeezed state for large $r$,
  \begin{equation}
    \bra{0}\hS^{\dagger}(r)\N^d\hS(r)\ket{0} \le e^{O(kr)} = \n^{O(d)}.
  \end{equation}
  We get an analogous bound for the two-mode squeezed state, as $\hS_2^\dagger(r)\a_1 \hS(r) = \a_1\cosh r - \a_2^\dagger\sinh r$ and $\hS_2(r)\a_2 \hS(r)=\a_2\cosh r-\a_1\sinh r$ \cite[Eq.~(7.159)]{Gerry_Knight_2004}.
  So specifically for the $\NP\subseteq\CVBQP$ containment, where we start with polynomial squeezing, energy will be bounded exponentially.
  Also note that the output extraction gadget (\cref{cor:Hout}) commutes with photon number measurement in modes $1,\dots,k$, and its output energy is polynomially bounded.
\end{proof}

\begin{proof}[Proof of \cref{thm:CVBQP-ELEMENTARY,thm:CVBQP-TOWER}]
  The only differences to the $\NP$ containment proof above are that we need more energy and we solve the Diophantine equation based on a universal Turing machine (\cref{thm:Diophantine}).
  \cref{lem:tower} lets us produce energy $2\uuarr n$ with $O(n)$ gates and modes.
\end{proof}

%% file: infinite-energy.tex
\begin{theorem}\label{thm:energy-undecidable}
  It is undecidable whether a Gaussian state evolved under a polynomial Hamiltonian has finite energy.
\end{theorem}
\begin{proof}
  Recall the Hamiltonian $G^{(c)}$ with real parameter $c\ge0$ from \cref{lem:dpa}:
  \begin{equation}
    G^{(c)} =  c \a^\dagger \a + \frac i2\left(\a^{\dagger2} - \a^2\right).
  \end{equation}
  By \cref{eq:dpa:uv}, we have
  \begin{equation}\label{eq:Nct}
    \N_c(t) = \ev{\a^\dagger(t)\a(t)}{0} = \abs{v(t)}^2 = \frac{\sin^2(\omega t)}{\omega} \le \frac{1}{c^2-1}.
  \end{equation}

  The result now follows from the undecidability of Diophantine equations \cite{Mat93}.
  Consider a Diophantine equation $f(x_1,\dots,x_k) = 0$ with $x_1,\dots,x_k \in \ZZ_{\ge0}$ without loss of generality.
  Define Hamiltonian on $k+2$ modes
  \begin{equation}\label{Heq-energy}
    H = \N_{k+1}^2\left(2f(\N_1,\dots,\N_k)\N_{k+2} + \frac i2\left(\a_{k+2}^{\dagger2} - \a_{k+2}^2\right)\right).
  \end{equation}
  $H$ is essentially self-adjoint because it can be written as a direct sum of Gaussian Hamiltonians.
  Let the initial state be $\ket{\phi} =(D(1)\ket{0})^{\otimes(k+1)}\otimes \ket{0}$.
  Then 
  \begin{equation}\label{eq:Heqphi}
    \begin{aligned}
    \ket{\psi} &\coloneq e^{-iH}\ket{\phi} = e^{-k/2}\sum_{\bfn \in \NN_0^k} \frac{1}{\sqrt{\bfn!}}\ket{\bfn}\otimes \exp(-i\N_{k+1}G^{(2f(\bfn))}_{k+2})D_{k+1}(1)\ket{0}_{k+1}\ket{0}_{k+2}, \\
    &=\sum_{\bfn \in \NN_0^{k},m\in\NN_0} b_{\bfn,m}\ket{\bfn,m}\otimes \exp(-i m G^{(2f(\bfn))})\ket{0},\quad b_{\bfn,m} = e^{-(k+1)/2}\frac{1}{\sqrt{\bfn!\,m!}},
    \end{aligned}
  \end{equation}
  where $\bfn! = \prod_{j=1}^k n_j!$
  Let $\N = \sum_{j=1}^{k+2} \N_j$. 
  The energy of $\ket{\psi}$ is given by $\ev{\N}{\psi} = \sum_{j=1}^{k+2}\ev{\N_j}{\psi}$.
  For $j=1,\dots,k+1$, we have $\ev{\N_j}{\psi} = \ev{\N_j}{\phi} = 1$ since $[\N_j,H]=0$.
  It remains to compute $\ev{\N_{k+2}}{\psi}$.
  We have
  \begin{equation}\label{eq:Nk2}
    \ev{\N_{k+2}}{\psi} = \sum_{\bfn\in \NN_0^k,m\in\NN_0} b_{\bfn,m}^2\, \N_{2f(\bfn)}(m).
  \end{equation}
  If $f$ has no solution, i.e., $f(\bfn)\ne 0$ for all $\bfn$, then $\ev{\N_{k+2}}{\psi} \le 1/3$ by \cref{eq:Nct}.
  Now suppose there exists $\bfn^*$ with $f(\bfn^*)=0$.
  Then
  \begin{equation}
    \begin{aligned}
    (\bra{\bfn^*}\otimes I_{k+1,k+2})\ket{\psi} &= \exp(-i\N_{k+1}G_{k+2}^{(0)})D_{k+1}(1)\ket{0}_{k+1}\ket{0}_{k+2}\\
    &=\exp(\frac12\N_{k+1}(\a_{k+2}^{\dagger2}-\a_{k+2}^{2}))D_{k+1}(1)\ket{0}_{k+1}\ket{0}_{k+2},
    \end{aligned}
  \end{equation}
  which has unbounded energy by \cref{prop:Dinf}.
  Note that the squeezing parameter above is negative, which does not affect the energy.
\end{proof}

\begin{theorem}\label{thm:adjointUndecidable}
  It is undecidable whether the evolution of a given essentially self-adjoint polynomial Hamiltonian preserves the Schwartz space.
\end{theorem}
\begin{proof}
  The proof of \cref{thm:energy-undecidable} already shows that $e^{-itH}$ for $H$ defined in \cref{Heq-energy} does not preserve the Schwartz space if $f$ has integer roots.
  This follows from the characterization  $\calS(\RR^{k+2}) = \bigcap_{j}\calD(\N^j)$ of \cref{lem:schwartz}.
  What remains to show is that if $f$ has no integer roots, then $H$ preserves the Schwartz space.
  Let $\ket{\phi}\in \calS(\RR^{k+2})$ and write
  \begin{equation}
    \ket{\phi} = \sum_{\bfn\in\NN_0^{k+1}}c_{\bfn}\ket{\bfn}\otimes \ket{\phi_{\bfn}}.
  \end{equation}
  Let $A = \prod_{j=1}^{k+1} \N_j(\N_{k+2}+1)$ and note $\bigcap_j\calD(\N^j) = \bigcap_j\calD(A^j)$ (see \cref{lem:schwartz}).
  Then
  \begin{equation}
    \ev{A}{\phi} = \sum_{\bfn}c_{\bfn}^2\prod_{j=1}^{k+1}n_j \cdot \ev{\N^d}{\phi_{\bfn}}.
  \end{equation}
  The evolved state is
  \begin{equation}
    \ket{\psi} = \exp(-itH)\ket{\phi} = \sum_{\bfn} c_{\bfn}\ket{\bfn}\otimes\ket{\psi_{\bfn}},\quad \ket{\psi_{\bfn}}=\exp(-in_{k+1}tG^{(f(n_1,\dots,n_k))})\ket{\phi_\bfn}.
  \end{equation}
  We argue that $\ev{A^d}{\psi} \le C_d\ev{A^d}{\phi}$ for some constant $C_d$ independent of $t$.
  It suffices to show $\ev{(\N+1)^d}{\psi_{\bfn}} \le C_d\ev{(\N+1)^d}{\phi_{\bfn}}$.
  Similar to \cref{eq:Nct}, it follows that the time evolved Hamiltonian in mode $k+2$, denoted $(\N_{k+2}+2)^d(t)$, can be written as a polynomial in $a,\a^\dagger$ with coefficients depending on $u(t),v(t)$.
  Since $u(t),v(t)$ are bounded for all $c\ge2$, the existence of $C_d$ follows from \cref{lem:form-bound}.
\end{proof}

\begin{theorem}\label{thm:schwartz}
  It undecidable whether a given symmetric polynomial Hamiltonian is essentially self-adjoint on the Schwartz space.
\end{theorem}
\begin{proof}
  Let $B$ be a polynomial symmetric Hamiltonian that is not essentially self-adjoint, here $B=\P^2-\X^4$ \cite[Theorem 9.41]{hall2013quantum}.\footnote{The theorem assumes $\Dom(B)=C_c^\infty(\RR)$, i.e. compactly support smooth functions. The result can be extended to the Schwartz space since $C_c^\infty(\RR)\subseteq \calS(\RR)$ is a core for $B$. The proof is analogous to \cite[Lemma 7.9]{teschl2014mathematical}.}
  Again, consider a Diophantine equation $f(n_1,\dots,n_k)$ and define Hamiltonian
  \begin{equation}
    H = 2f(\N_1,\dots,\N_k)(\N_{k+1}+1)^2 + B_{k+1}.
  \end{equation}
  If $f(\bfn)=0$, then $H$ is not essentially self-adjoint in the subspace $\Span(\ket{\bfn})\otimes \calS(\RR)$.
  However, if $f$ has no solution, then $H$ is self-adjoint by the Kato-Rellich Theorem \cite[Theorem 6.4]{teschl2014mathematical}.
  For $A = 2f(\N_1,\dots,\N_k)(2\N_{k+1}+1)^2$ the statement easily follows as $2\N+1 = \X^2+\P^2$.
\end{proof}

%% file: schwartz.tex
\renewcommand\K{{\hat K}}
\section{Mathematical Tools}
\label{app:mathtools}

In this appendix, we include basic mathematical tools used in the main body of the paper.
For a deeper background, we refer to the following text books: Engel--Nagel \cite{EN2000}, Hall \cite{hall2013quantum}, Reed--Simon \cite{Reed1981-kj,reed1975ii}, Schmüdgen \cite{Schmdgen2012}, and Teschl \cite{teschl2014mathematical}.

\subsection{Polynomial Hamiltonians}

In this paper, we almost exclusively deal with Hamiltonians of the form
\begin{equation}
  H \in \CC[\a_1,\a_1^\dagger,\dots,\a_m,\a_m^\dagger],
\end{equation}
of finite degree $d = \deg(H)$.

\begin{definition}
  Denote by $\calP_{m,d}$ the set of degree-$d$ operators in $\CC[\a_1,\a_1^\dagger,\dots,\a_m,\a_m^\dagger]$, and let $\calP_m = \bigcup_{d\in\NN}\calP_{m,d}$.
\end{definition}

$H$ is a linear operator on the hilbert space $\calH = L^2(\RR^m)$.
Recall that on the Fock basis, we have for all $n\in \NN_0$
\begin{equation}
  \a\ket{n} = \sqrt{n}\ket{n-1}, \quad \a^\dagger\ket{n} = \sqrt{n+1}\ket{n+1}.
\end{equation}
Also recall that these operators define the momentum and position operators 
\begin{equation}
  \P = \frac{1}{\sqrt{2}i}(\a - \a^\dagger),\quad \X = \frac{1}{\sqrt2}(\a + \a^\dagger).
\end{equation}
The number operator is given by 
\begin{equation}
  \N = \sum_{i=1}^m \N_i,\quad \N_i = \a_i^\dagger \a_i.
\end{equation}

Unbounded operators are generally not defined on the entire Hilbert space, which is why we write $T\colon \calD(T)\to \calH$, where $\calD(T)\subseteq \calH$ is called the domain of $T$.
Since the creation and annihilation operators are defined on the entire Fock basis of $\calH$, any $T\in \calP_m$ is a well-defined linear operator with domain $\calD(T) = \{\psi\in \calH \colon \norm{T\psi}<\infty\}$.
The expression ``$T\psi$'' is a slight abuse of notation, since not necessarily $T\psi \in \calH$. However, in the case of polynomial operators
$T\psi$ is by linearity a well-defined function $\RR^m\to \CC$ (or equivalently $\NN_0^m\to \CC$).
Hence, $\psi \in \calD(T)$ iff $T\psi$ is square-integrable.

Since this notion of domain is rather unwieldly, we can choose a suitable \emph{core} $\calC\subseteq \calD$, which is a dense subset of the domain $T$ whose image under $T$ is also dense.
The linear subspace of finite Fock superpositions is a convenient choice for Gaussian and $\X^d$ Hamiltonians:
\begin{equation}
  \Dfin = \bigcup_{n=0}^\infty\calH_k,\quad \calH_{n} = \Span\Bigl\{\ket{n_1,\dots,n_m}\colon \sum_{i=1}^m n_i \le n\Bigr\}.
\end{equation}

\subsubsection{Tools}

\begin{lem}\label{lem:permute-f(N)}
  Let $f\colon \NN_0 \to \CC$.
  Then on $\Dfin$,
  \begin{equation}
    \a_i f(\N) = f(\N+1)\a_i, \quad \text{and }\; \a_i^\dagger f(\N)=f(\N-1)\a_i^\dagger.
  \end{equation}
\end{lem}
\begin{proof}
  Let $\{\Pi_n\}_{n\in\NN_0}$ be the projections onto the $n$-particle sectors (i.e. spectral projections of $\N$).
  Then 
  \begin{equation}
  \N = \sum_{n=0}^\infty n\Pi_n,\quad f(\N) = \sum_{n=0}^\infty f(n)\Pi_n.
  \end{equation}
  Using the convention $\Pi_{-1}=0$, we have for all $n\in\NN_0$
  \begin{equation}
    \Pi_n \a_i = \a_i \Pi_{n+1},\quad \Pi_n \a_i^\dagger = \a_i^\dagger \Pi_{n-1}.
  \end{equation}
  Thus,
  \begin{equation}
    \a_i f(\N) = \sum_{n=0}^\infty f(n)\a_i \Pi_n = \sum_{n=0}^\infty f(n)\Pi_{n-1}\a_i = \sum_{n=0}^\infty f(n+1)\Pi_{n} \a_i = f(\N+1)\a_i,
  \end{equation}
  where $\Pi_{-1}=0$ allows us to shift the indices. Similarly,
  \begin{equation}
    \a_i^\dagger f(\N) = \sum_{n=0}^\infty f(n)\a_i^\dagger \Pi_n = \sum_{n=0}^\infty f(n) \Pi_{n+1} \a_i^\dagger = \sum_{n=1}^\infty f(n-1) \Pi_{n} \a_{i}^\dagger = f(\N-1)\a_{i}^\dagger,
  \end{equation}
  where we used $\Pi_0 \a_i^\dagger=0$.
\end{proof}

\begin{lem}\label{lem:form-bound}
  Let $T\in\calP_{m,d}$.
  There exists $c$, such that $|\ev{T}{\psi}|\le c\ev{(\N+1)^{d/2}}{\psi}$ for all $\psi\in \calD(\N^{d/4})$.
\end{lem}
\begin{proof}
  Note, we do not require $T$ to be symmetric.
  It suffices to verify the inequality for a single monomial of the form $T=\prod_{j=1}^m \a_j^{\dagger d_j}\a_j^{e_j}$ with $\sum_{j=1}^m d_j+e_j\le d$.
  Additionally, there exists $c'$, such that for all $\psi\in \calD(\N^{d/4})$
  \begin{equation}
    \ev{(\N+1+2d)^{d/2}}{\psi}\le c'\ev{(\N+1)^{d/2}}{\psi}.
  \end{equation}
  Define the operator
  \begin{equation}
    B = (\N+2d+1)^{-d/4} T (\N+2d+1)^{-d/4} = (\N+d+1)^{-d/4}(\N+2d+1+\delta) T,
  \end{equation}
  with $|\delta|\le d$ \cref{lem:permute-f(N)}.
  $B$ is bounded on all of $L^2(\RR^m)$ since
  \begin{equation}\label{eq:form-bound:bounded}
    T\sum_{\bfx\in\NN_0^m}\alpha_i\ket{\bfx} = \sum_{\bfx\in\NN_0^m}\alpha_i'\ket{\bfx+\mathbf{\Delta}},
  \end{equation}
  with $\abs{\alpha_i'}\le \abs{\alpha_i}$ and $\mathbf{\Delta}\in\ZZ^m$ with $\sum_{i=1}^m|\Delta_i| \le d$, defining $\ket{\bfx+\Delta}$ with negative entries as $0$.
  Then for $\psi\in\calD(\N^{d/4})=\calD((\N+2d+1)^{d/4})$,
  \begin{equation}
    \begin{aligned}
    \ev{T}{\psi} &= \ev{(\N +2d+1)^{d/4} B(\N+2d+1)^{d/4}}{\psi} \le \norm{B}\norm{(\N+2d+1)^{d/4}\psi}^2 \\
    &= \norm{B}\ev{(\N+2d+1)^{d/2}}{\psi} \le c'\norm{B}\ev{(\N+1)^{d/2}}{\psi}.
    \end{aligned}
  \end{equation}
\end{proof}

\begin{lem}\label{lem:norm-bound}
  Let $T\in\calP_{m,d}$.
  There exists $c$, such that $\norm{T\psi}\le c\norm{(\N+1)^{d/2}\psi}$ for all $\psi\in \calD(N^{d/2})$.
\end{lem}
\begin{proof}
  The proof is analogous to \cref{lem:form-bound}. Assume $T$ is a monomial. Define
  \begin{equation}
    B = T(\N+2d+1)^{-d/2},
  \end{equation}
  which is bounded by the same argument as \cref{eq:form-bound:bounded}.
  Then
  \begin{equation}
    \norm{T\psi} = \norm{B(\N+2d+1)^{d/2}\psi} \le \norm{B}\norm{(\N+2d+1)^{d/2}\psi} \le c\norm{B}\norm{(\N+1)^{d/2}\psi}.
  \end{equation}
\end{proof}

\subsubsection{Schwartz space}\label{app:schwarz}

In the following, we give a (mostly) self-contained proof that the Schwartz space can be characterized precisely as the set of states with finite moments $\ev*{\N^k}$ for all $k\in\NN$ (i.e., \Cref{lem:schwartz}), where $\N$ is the photon number observable as below.
Recall the Schrödinger representation of $L^2(\RR^m)$:
\begin{equation}\label{eq:L2}
  \a_j=\frac1{\sqrt2}(x_j + \partial_{x_j}),\quad \a_j^\dagger = \frac1{\sqrt2}(x_j-\partial_{x_j}),\quad \N_j = \a_j^\dagger \a_j, \quad \N=\sum_{j=1}^m \N_j, \quad \K\coloneq \N+I
\end{equation}

\begin{definition}
  The Schwartz space on $\RR^m$ is defined as
  \begin{equation}\label{eq:Schwartz}
    \begin{aligned}
    \calS(\RR^m) &= \bigl\{ f\in C^\infty(\RR^m,\CC)\bigm| \forall \bfalpha,\bfbeta\in\NN^m\colon \norm{f}_{\bfalpha,\bfbeta}<\infty\bigr\},\\
    \norm{f}_{\bfalpha,\bfbeta} &= \sup_{\bfx\in\RR^m} \abs{\bfx^{\bfalpha}(\bfpartial^{\bfbeta} f)(\bfx)},
    \end{aligned}
  \end{equation}
  where $\bfx^{\bfalpha} = x_1^{\alpha_1}\dotsm x_m^{\alpha_m}$ and $\bfpartial^\beta = \partial_{x_1}^{\beta_1}\dotsm \partial_{x_m}^{\beta_m}$.
\end{definition}

\begin{lem}\label{lem:S-closed}
  If $\psi\in \calS(\RR^m)$, then also $\bfx^{\bfalpha}\bfpartial^{\bfbeta}\psi\in\calS(\RR^m)$ for all $\bfalpha,\bfbeta\in\NN^m$.
\end{lem}
\begin{proof}
  Follows directly from \cref{eq:Schwartz}.
\end{proof}

\begin{lem}\label{lem:S-L2}
  $\calS(\RR^m) \subseteq L^2(\RR^m)$.
\end{lem}
\begin{proof}
  We need to show that all functions in the Schwartz space are square integrable.
  Let $\psi\in\calS(\RR^m)$.
  By definition of the Schwartz space, there exists a constant $C$, such that $\abs{\psi(\bfx)} \le C/(1+\norm{\bfx})^m$.
  One can easily verify that $\int_{\RR^m}d\bfx \abs{\psi(\bfx)}^2 < \infty$ by integration in polar coordinates.
\end{proof}

\begin{definition}[Sobolev space]
  $W^{k,p}(\RR^m) = \{f\in L^p(\RR^m)\mid \bfpartial^{\bfalpha}f \in L^p\}$.
\end{definition}

\begin{lem}\label{lem:schwartz}
  $\calS(\RR^m) = \bigcap_{k=0}^\infty \calD(\N^k)$.
\end{lem}
\begin{proof}
  First note $\calD(\N^k) = \calD(\K^k)$ since
  \begin{equation}
    \norm{\K^k\psi} = \norm{(\N+1)^k\psi} \le \norm{(2^k\N+1)\psi} \le 2^k\norm{\N^k\psi}+\norm{\psi},
  \end{equation}
  which is bounded if $\psi\in\calD(\N^k)$.

  Let $\psi\in\calS(\RR^m)$. Using \cref{eq:L2} 
  we can expand $\K^k$ as $\K^k = \sum_{\abs{\bfalpha}+\abs{\bfbeta}\le 2k} c_{\bfalpha,\bfbeta} \bfx^{\bfalpha}\bfpartial^{\bfbeta}$ in the Weyl algebra.
  Then $\norm{\K^k\psi} \le \sum_{\bfalpha,\bfbeta} c_{\bfalpha,\bfbeta}\norm{\bfx^{\bfalpha}\bfpartial^{\bfbeta}\psi}$, which is bounded since $\bfx^{\bfalpha}\bfpartial^{\bfbeta}\psi\in \calS(\RR^m)\subseteq L^2(\RR^m)$ (\cref{lem:S-closed,lem:S-L2}).

  First, observe that $\bfx^{\bfalpha}\bfpartial^{\bfbeta}\psi\in \bigcap_k\calD(\N^k)$ by writing $\bfx^{\bfalpha}\bfpartial^{\bfbeta}$ as a polynomial in $a_j,a_j^\dagger$ and applying \cref{lem:norm-bound}.
  Thus, we have that $\norm{\bfx^{\bfalpha}\bfpartial^{\bfbeta}\psi}_2$ is bounded, however we still need to show that the $\infty$-norm is bounded.
  We use the Sobolev embedding \cite[Corollary 9.13]{Brezis2011}, which gives $W^{s,2}(\RR^m) \subseteq L^{\infty}(\RR^m)$ for $s > m/2$ with
  \begin{equation}
    W^{s,p}(\RR^m) =\, \bigr\{ u\in L^p(\RR^m) \bigm| \bfpartial^{\bfalpha} u\in L^p(\RR^m)\;\forall \bfalpha,\abs{\bfalpha}\le s \bigl\}.
  \end{equation}
  By previous argument $\bfx^{\bfalpha}\bfpartial^{\bfbeta}\psi \in L^2(\RR^m)$ and all of its derivatives.
  Thus, $\bfx^{\bfalpha}\bfpartial^{\bfbeta}\in W^{s,2}(\RR^m)$ with $s > m/2$.
\end{proof}

\subsection{Spectral gaps}

We say an $n\times n$ matrix $A$ is irreducible if it is not similar via a permutation to a block upper triangular matrix.
Equivalently, the directed graph $G_A$ with vertices $V(G_A) = [n]$ and edges $E(G_A) = \{(i,j) \colon A_{ij} \ne 0\}$ is strongly connected \cite[Theorem 6.2.24]{Horn_Johnson_2012}.
Given Hermitian matrix $A\in \CC^{n\times n}$, we denote its eigenvalues by $\lambda_1(A)\le \dotsm\le\lambda_n(A) \in \RR$.
The spectral gap is $\gamma(A)\coloneq \lambda_2(A) - \lambda_1(A)$.

\begin{lem}\label{lem:unique-groundstate}
  Let $A \in \RR^{n\times n}$ be symmetric and irreducible, such that $A_{ij} \le 0$ for all $i\ne j$.
  Then the spectral gap $\gamma(A) \coloneq \lambda_2(A) - \lambda_1(A) > 0$ and $A$ has an eigenvector with strictly positive entries and eigenvalue $\lambda_1(A)$.
\end{lem}
\begin{proof}
  Let $B = cI - A$, where $c > \lambda_n(A)$ and $c>|A_{ij}|$ for all $i,j$.
  Then $B$ is irreducible, nonnegative, and $B \succcurlyeq 0$.
  Thus $\lambda_n(B) = \rho(B)$, where the spectral radius is defined as $\rho(B) \coloneq \max_{i}\abs{\lambda_i(B)}$.
  By the Perron--Frobenius theorem \cite[Theorem 8.4.4]{Horn_Johnson_2012}, $\rho(B)$ is a simple eigenvalue.
  Thus $\lambda_2(A) - \lambda_1(A) = \lambda_{n}(B) - \lambda_{n-1}(B) > 0$
\end{proof}

\begin{lem}\label{lem:gapscale}
  Let $A, B \succcurlyeq 0 \in \CC^{n\times n}$ and $C=A+B$, such that $\lambda_1(C)=0$ and $\lambda_2(C)>0$.
  Let $s>0$ and $C' = A+ sB$.
  Then $\lambda_1(C')=0$ and $\lambda_2(C') \ge \min\{s,1\}\cdot \lambda_2(C)$
\end{lem}
\begin{proof}
  We have $\ker C=\ker C'$.
  By Courant--Fischer \cite[Theorem 4.2.6]{Horn_Johnson_2012},
  \begin{equation}
    \lambda_2(C) = \max_{\dim S=n-1}\ \min_{x\in S,\norm{x}=1} x^\dagger C x,\quad \lambda_2(C') = \max_{\dim S=n-1}\ \min_{x\in S,\norm{x}=1} x^\dagger C' x.
  \end{equation}
  Suppose $s\le 1$. Then
  \begin{equation}
    x^\dagger C' x \ge x^\dagger (C' - (1-s)A) x = x^\dagger (sA - sB)x = s x^\dagger C x.
  \end{equation}
  For $s \ge 1$, we have
  \begin{equation}
  x^\dagger C' x \ge x^\dagger (C' - (1-s)B)x = x^\dagger C x.
  \end{equation}
\end{proof}

\begin{corollary}\label{cor:gapscale}
  $\lambda_2(C) \ge \min\{1/s,1\}\cdot \lambda_2(C')$.
\end{corollary}
\begin{proof}
  Apply \cref{lem:gapscale} to $\tilde{C} = \tilde{A} + \tilde{B}$ and $\tilde{C'} = \tilde{A} + \tilde s\tilde B$ with $\tilde{A} = A$, $\tilde B = sB$, and $\tilde s= 1/s$.
\end{proof}

\begin{lem}\label{lem:gapdiag}
  Let $A \succcurlyeq 0 \in \CC^{n\times n}$ with $\lambda_2(A) > 0 = \lambda_1(A)$, and $D = \diag(d_1,\dots d_n) \succ 0$.
  Let $B = D^{-1}AD^{-1}$.
  Then $\lambda_1(B)=0$ and 
  \begin{equation}
    \frac{\lambda_2(A)}{\beta^2}\ \le\ \lambda_2(B)\ \le\ \frac{\lambda_2(A)}{\alpha^2},
  \end{equation}
  where $\alpha = \min_i d_i$ and $\beta = \max_i d_i$.
\end{lem}
\begin{proof}
  Clearly $\lambda_1(B) = 0$ as for $Ax = 0$, we have $B(Dx)=0$.
  By Courant--Fischer,
  \begin{equation}
    \lambda_2(B) = \min_{\dim S = 2}\ \max_{x\in S\setminus\{0\}} \frac{x^\dagger B x}{\norm{x}^2}\ =\ \min_{\dim S = 2}\ \max_{y\in S\setminus\{0\}} \frac{y^\dagger A y}{\norm{Dy}^2},
  \end{equation}
  where the equality holds by substituting $x$ with $y = Dx$, which is legal since $D$ is invertible.
  The lemma then follows as
  \begin{equation}
     \frac{y^\dagger A y}{\beta^2\norm{y}^2}\ \le\ \frac{y^\dagger A y}{\norm{Dy}^2}\ \le\ \frac{y^\dagger A y}{\alpha^2\norm{y}^2}.
  \end{equation}
\end{proof}

\begin{lem}\label{lem:laplacian-sum}
  Let $L_1,L_2$ be the Laplacians of two undirected $n$-vertex graphs with nonnegative edge weights.
  Let $L = L_1+L_2$.
  Then $\lambda_1(L) = 0$ and $\lambda_2(L) \ge \max\{\lambda_2(L_1),\lambda_2(L_2)\}$.
\end{lem}
\begin{proof}
  We apply a similar Courant--Fischer argument. We have for $x\in\RR^n$,
  \begin{equation}
    x^\dagger L x = \frac12\sum_{u,v} (w_1(u,v) + w_2(u,v))(x_u- x_v)^2 = x^\dagger L_1 x  + x^\dagger L_2 x.
  \end{equation}
  Thus, $\lambda_2(L) \ge \max\{\lambda_2(L_1),\lambda_2(L_2)\}$.
\end{proof}

\begin{lem}
  Let $L$ be the Laplacian of an $n$-vertex graph diameter $D$ and edge weights $\ge\alpha>0$.
  Then $\lambda_2(L) \ge 1/\alpha nD$. 
\end{lem}
\begin{proof}
  For an unweighted graph, we have $\lambda_2 \ge 1/nD$ \cite{Mohar1991}.
  By \cref{lem:laplacian-sum}, increasing edge weights does not decrease the spectral gap.
  Thus, $\alpha^{-1}\lambda_2(L) = \lambda_2(\alpha^{-1} L) \ge 1/nD$.
\end{proof}

\begin{lem}\label{lem:gap}
  Let $n,B\in\NN$,
  \begin{itemize}[parsep=0pt,itemsep=2pt,topsep=0pt,leftmargin=2em]
    \item $H_0 = \mathrm{diag}(d_1,\dots,d_n)$ with $d_1,\dots,d_n\in \{0,\dots, B\}$ such that there exists a unique $k\in[n]$ with $d_k=0$,
    \item $H_1 \in \{-B,\dots,B\}^{n\times n}$ be symmetric and irreducible,
    \item $H(s) \coloneq (1-s)H_0 + sH_1$ with spectral gap $\gamma(s) \coloneq \lambda_2(H(s))-\lambda_1(H(s))$.
  \end{itemize}
  Then $\gamma(s) \ge \exp((n\log B)^{-O(1)})$ for all $s\in [0,1]$.
\end{lem}
\begin{proof}
  At $s=0$, $H(s)=H_0$ is trivially gapped.
  For $s>0$, $H(s)$ is irreducible since $H_1$ is irreducible and $H_0$ diagonal.
  Define the nonnegative matrix $A(s) \coloneq BI- H(s)$.
  By the Perron-Frobenius theorem for real irreducible matrices \cite[Theorem 8.4.4]{Horn_Johnson_2012}, the spectral radius $\rho(A) = \max_i \abs{\lambda_i(A)}$ is a simple (non-degenerate) eigenvalue.
  Thus, $\lambda_1(H(s))$ must also be a simple eigenvalue of $H(s)$.
  
  By \cite[Theorem 5.1]{Alekseevsky1998}, there exists a set of real analytic (therefore continuous) functions $\{f_1(s),\dots,f_n(s)\}$ (not necessarily ordered) describing the roots of the characteristic polynomial of $H(s)$ at time $s$.
  Thus $\lambda_1(s)=\min_i\lambda_i(s)$ is also continuous, as well as $\lambda_2(s) = \min\{\max\{f_i(s),f_j(s)\}\mid i,j\in[n]\}$.
  Thus the gap $\gamma(s) = \lambda_2(s)-\lambda_1(s)>0$ is continuous on $[0,1]$, and thus lower bounded, i.e., there exists $b>0$, such that $\gamma(s)\ge b$ for all $s\in[0,1]$.
  We can express this gap (or rather its inverse $b^{-1}$) as a sentence in the first-order theory of the reals with a constant number of variables:
  \begin{equation}
    \begin{aligned}
    \exists y>0\; \forall s \in [0,1]\;\forall x_1,x_2\;\exists x_3&\colon [\chi(s,x_1) = 0,\chi(s,x_2)=0,x_1< x_2]\\
    &\rightarrow [x_3 < x_2, x_3\ne x_1, \chi(s,x_3)=0] \vee [((x_2 - x_1)y)^2 \ge  1].
    \end{aligned}
  \end{equation}
  In words, there exists $b>0$ ($b=y^{-1}$), such that for each pair of roots $x_1,x_2$ of the characteristic polynomial $\chi(s,\cdot)$ at time $s$, either $x_1,x_2$ is not the pair of smallest roots, or $|x_1-x_2| \ge y^{-1}$.
  By the preceding argument, this $y$ exists, and due to \cite[Proposition 1.3]{Renegar92}, there is a solution with $\log y \le (n\log B)^{O(1)}$, as the system only has integer coefficients.
  We also remark that the symbolic characteristic polynomial $\chi(s,x)$ can be efficiently computed by \cite[Lemma 6.3]{kamminga2025complexitypurestateconsistencylocal}.
  Therefore, Renegar's algorithm \cite[Theorem 1.2]{Renegar92} can compute the smallest spectral gap in time $\poly(n, \log B)$.
\end{proof}

%% file: phasespacesim.tex
\section{Phase Space Gaussian rank simulation} \label{sec:phase-space-grank}

\cref{lem:hdyne-gaussian-expect} and \cref{lem:Gaussian-complx} use the position or Wigner function based methods to find the required overlaps. In this appendix, we show that the phase-space simulation of the Gaussian paths developed in \cite{diasClassicalSimulationNonGaussian2024} and \cite{hahnClassicalSimulationQuantum2024} also gives the same complexity bounds, only with slightly different analysis. \cite{diasClassicalSimulationNonGaussian2024} already provides methods to calculate the homodyne detection and single-mode gaussian overlap provided the Gaussian state is given as a sequence of the elementary Gaussian gates. Here we extend the formalism to implement arbitrary Gaussian gates and even when the Gaussian state is given in its phase space formalism. We expect that implementation of the Gaussian rank based simulation techniques may benefit from this.

\subsection{Preliminaries and notation}

The paper tracks relative phases of the states by adding information about the overlap with a reference state. Taking the convention in \cite{diasClassicalSimulationNonGaussian2024}, we set the reference state to be $\ket{\alpha}$, which is the coherent state with the same displacement as our Gaussian state $\ket{g}$, which we call $r$. Therefore the description of $\ket{g} = \Delta(g) = (\Gamma, \alpha, r)$

Gaussian operations are described by a scale matrix and shift vector $(S, d)$, and updating the Gaussian state involves calculating $\Gamma^\prime = S\Gamma S^T$, $\alpha^\prime = S\alpha + d$, and $r'$. The calculation of $r'$ is only nontrivial for the squeezing operation, and $r'=r$ otherwise.

\subsection{Block Squeezing operator}

The paper already defines \texttt{squeezing}, which is a single-mode squeezing. However, a very common subroutine is to apply SMS operations on every mode with squeezing parameters $[z_i]$. The paper implements this by applying them in sequence, requiring repeated calls to the expensive \texttt{overlaptriple}. Here we show that the protocol can be generalized to a \texttt{blocksqueezing} operation.

\begin{algorithm}
\caption{\texttt{blocksqueezing}}\label{alg:cap}
\begin{algorithmic}
    \Require $\Delta_n, [z_j]_{j\in[n]}$\\
    where $\begin{cases}
     \Delta_n = (\Gamma, \alpha, r)\in\text{Desc}_n \\
     [z_j \in (0, \infty)] \text{ is a list of squeezing operations}
    \end{cases}$
    
    \Ensure $(\Gamma^\prime, \alpha^\prime, r^\prime) \in \text{Desc}_n$
    
    \State $S \gets diagm([\exp(-z_1), \exp(z_1), ... \exp(-z_n), \exp(z_n)])$ 
    \Comment{S of block squeezing}
    
    \State $\Gamma^\prime \gets S \Gamma S^T$ 
    \Comment{covariance matrix of $S_n([z_j))|\psi\rangle$}
    
    \State $\Gamma^{\prime\prime} \gets S S^T$ 
    \Comment{covariance matrix of $S_n([z_j))|\alpha\rangle$}
    \For {$j \in [n]$}
        \State $\alpha^\prime_j \gets \cosh(z_j)\alpha_j - \sinh(z_j)\overline{\alpha_j}$
    \EndFor
    \State $(d_1, d_2, d_3) \gets \hat{d}(\alpha^\prime)$
    \Comment{displacement vectors of $S([z_j])|\alpha\rangle$, $|\alpha^\prime\rangle$ and $S([z_j])|\psi\rangle$}
    \State $(\Gamma_1¸ \Gamma_2, \Gamma_3) \gets (\Gamma^{\prime \prime}, I, \Gamma^{\prime})$ \\
    \Comment{covariance matrices of $S([z_j])|\alpha\rangle$, $|\alpha^\prime\rangle$ and $S([z_j])|\psi\rangle$}
    \State $u \gets \overline{r}$ \Comment{$\left\langle S([z_j])\psi, S(z_j)\alpha\right\rangle$}
    \State $v \gets \prod_{j \in [n]} \frac{1}{\sqrt{\cosh(z_j)}}$ \Comment{$\langle \alpha^\prime, S(z_j)\alpha\rangle$}
    \State $r^\prime \gets \texttt{overlaptriple}(\Gamma_1,\Gamma_2,\Gamma_3, d_1, d_2, d_3, u, v, 0)$
    \Comment{$\langle \alpha^\prime, S(z_j)\psi\rangle$}\\
    \Return $(\Gamma^\prime, \alpha^\prime, r^\prime)$
\end{algorithmic}
\end{algorithm}

\begin{lem}
$\texttt{blocksqueezing}(\Delta(\psi), [z_j]) = \Delta(S([z_j])\ket{g}) \forall\text{ Gaussian } \ket{g}$ and runs in $O(n^3)$
\end{lem}
\begin{proof}
Note that $S$ is the symplectic matrix of the operation $\prod_{j\in[n]}S_j(z_j) = S([z_j])$, and the displacement caused by the operation is $\vec{0}$, so $\Gamma^\prime$ and $\Gamma^{\prime\prime}$ are correctly calculated using the usual update rules of Gaussian operations; $\alpha^\prime$ can easily be seen as the displacement vector resulting by sequential operation of $\prod_{j\in[n]}S_j(z_j) = S([z_j])$; $u$ is correctly calculated by definition; $v$ is correctly calculated, since a multimode coherent state can be written as a product state of single mode coherent states, and $S([z_j])$ as a tensor product of single mode operations. Therefore, the overlap is simply the product of mode-wise overlaps.

The runtime is $O(n^3)$ due to the single call to \texttt{overlaptriple}.

\end{proof}

\subsection{Implementing arbitrary operations}

Any Gaussian operation can be decomposed via the Euler/ Bloch-Messiah decomposition \cite{houdeMatrixDecompositionsQuantum2024} in the following way:

$$
\begin{aligned}
\mathcal{G}(S, s) =& D(d) \mathcal{G}(S, 0) \\
=& D(s) \mathcal{G}(O_1ZO_2, 0)\\
=& D(s) \mathcal{G}(O_2, 0)\mathcal{G}(Z, 0)\mathcal{G}(O_1, 0),\\
\end{aligned}
$$
with $O_1$ and $O_2$ are passive operations, which do not change $r$. $\mathcal{G}(Z, 0)$ is a block squeezing operator with the update rule given above. The covariance matrix and displacement vector update rules remain the same.

\subsection{Computing $r$ for arbitrary states}

An arbitrary Gaussian state $G = (\Gamma, d)$ can be written as $D(d)\mathcal{G}(V, 0)|0\rangle$, where $V$ is calculated by the Williamson decomposition. $\mathcal{G}$ can be decomposed as $O_2ZO_1$ and $\ket{0}$ can be updated as before. 

However, the calculation of $r$ can be greatly simplified by noting that the passive gate $O_1$ operates on $|0\rangle$ as identity, $Z$ is a block squeezing operator, and the passive $O_2$ update is trivial. Therefore, the calculation is greatly simplified since most of the terms in \cite[Eqn 30]{diasClassicalSimulationNonGaussian2024} are 0, or are simpler. 

Since $\Delta(|0\rangle_n) = (\mathbf{I}_{2n}, \vec{0}, 1)$, the numerator in the call to \texttt{overlaptriple} is $\exp(0) = 1$, and the calculation of $\Gamma_4$ does not require a matrix inversion since it is diagonal. Therefore, 

$$r = \frac{v^{-1}}{\sqrt{\det((\Gamma_2+\Gamma_3)/2)\det((\Gamma_1+\Gamma_4)/2)}}.$$

Now that we have a description of a state with phase information, we can calculate the quantities in \cref{lem:hdyne-gaussian-expect} and \cref{lem:hdyne-gaussian-expect} using the protocols in \cite{diasClassicalSimulationNonGaussian2024}.

\subsection{Complexity analysis}

The complexity of the protocol depends on the complexity of the Williamson Decomposition and Euler Bloch decomposition. Numerical recipes in \cite{houdeMatrixDecompositionsQuantum2024} require only linear algebra operations such as SVD, inversions and multiplications of $\poly(n)$ sized matrices to implement these decompositions which have the same complexity as the operations used in these lemmas. Therefore, the complexity of calculating the values does not change.